\documentclass[a4paper,twocolumn,11pt,accepted=2025-12-04]{quantumarticle}
\pdfoutput=1
\usepackage{amsfonts}
\usepackage{amsmath}
\usepackage{amssymb}
\usepackage{amsthm}
\usepackage{mathtools}
\usepackage{dsfont}
\usepackage{bm}
\usepackage{babel}
\usepackage[numbers,sort&compress]{natbib}

\usepackage{xcolor}
\colorlet{myPurple}{blue!40!red}
\colorlet{myPurplee}{blue!10!red}
\colorlet{myCyan}{cyan!80!gray}
\colorlet{myRed}{red!90!black}
\usepackage[colorlinks=true,citecolor=myRed,urlcolor=myRed,linkcolor=myRed]{hyperref}

\usepackage{float}
\usepackage{multirow}
\usepackage{caption}
\usepackage{subcaption}
\usepackage{arydshln}
\usepackage{cellspace}
\usepackage{tikz}
\usetikzlibrary{quantikz2}

\theoremstyle{plain}
\newtheorem{thm}{Theorem}
\newtheorem{lem}{Lemma}
\newtheorem{refexp}{Reference Experiment}
\newtheorem{prop}{Proposition}

\newtheorem{defi}{Definition}

\newcommand{\figtikzscale}{.85}
\newcommand{\X}{\mathsf{X}}
\newcommand{\Z}{\mathsf{Z}}

\begin{document}

\title{Self-Testing Graph States Permitting Bounded Classical Communication}
\author{Uta Isabella Meyer}
\author{Ivan Šupić}
\author{Frédéric Grosshans}
\author{Damian Markham}
\affiliation{Sorbonne Université, CNRS, LIP6, F-75005 Paris, France}

\begin{abstract}
Self-testing identifies quantum states and correlations that exhibit nonlocality, distinguishing them, up to local transformations, from other quantum states.
Due to their strong nonlocality, it is known that all graph states can be self-tested in the standard setting - where parties are not allowed to communicate.
Recently it has been shown that graph states display nonlocal correlations even when bounded classical communication on the underlying graph is permitted, a feature that has found applications in proving a circuit-depth separation between classical and quantum computing.
In this work, we develop self-testing in the framework of bounded classical communication, and we show that certain graph states can be robustly self-tested even allowing for communication. In particular, we provide an explicit self-test for the circular graph state and the honeycomb cluster state - the latter known to be a universal resource for measurement based quantum computation.
Since communication generally obstructs self-testing of graph states, we further provide a procedure to robustly self-test any graph state from larger ones that exhibit nonlocal correlations in the communication scenario.
\end{abstract}

\maketitle

\section{Introduction}

Bell nonlocality describes a phenomenon in which entangled systems display correlations that elude classical explanations~\cite{bell,bellreview}.
Such correlations are exposed in the so-called Bell scenario where two or more non-communicating parties receive some classical inputs and return classical outputs.
A description of the outputs' correlations through local-hidden-variable models (LHV) leads to the formulation of a Bell inequality
whose violation implies that the involved parties share some non-classical resources.
Thus, the Bell scenario is considered a prime stage for witnessing quantumness.

\paragraph*{Self-testing \& Communication.}
Using correlations to certify states and operations, commonly referred to as `device-independence', has been a productive direction over the past two decades~\cite{colbeck,colbeckkent,randomnessreview,pironio2010random,Acin2006,Acin2007}.
Any Bell-inequality violation witnesses entanglement and LHV models cannot reproduce all Bell-nonlocal correlations, but some are uniquely tied to a particular state (up to well-defined transformations).
When correlations are unattainable by LHV models and by any other quantum state, the state is said to be self-tested~\cite{mayers98,mayers2003self,Supic2020selftestingof}.
Observing such correlations yields a device-independent certificate of that state.
Beyond entanglement witnessing, many quantum-information protocols are device-independent within the Bell scenario~\cite{VV,arnon2018practical,gallego2010device,coladangelo2019verifier,reichardt2013classical,ribeiro2018fully}.

All device-independent protocols rely on violations of Bell inequalities, which necessitates spatial separation and prohibits communication among the participating parties.
This requirement can be alleviated by ensuring that the parties are shielded, preventing the transmission of information regarding the local Bell input to any other party.
Nevertheless, fulfilling this condition poses a considerable challenge, particularly in scenarios involving the creation of multipartite states on integrated photonics platforms~\cite{vallée2023corrected,fyrillas2023certified}.

An intriguing question arises regarding the possibility of self-testing when the absence of communication among parties cannot be guaranteed.
To date, this scenario has not received any attention, with a notable exception being~\cite{murta2023self}, showing that if $N$ parties are split into $k$ clusters - each possibly harboring dishonest parties - one can still self-test $k$-partite GHZ states.
This outcome aligns with the intuition that within a Bell scenario involving $N$ parties, any form of communication effectively diminishes the certifiable resource size.
Indeed, as elucidated in~\cite{murta2023self}, only a quantum resource shared among non-communicating clusters is susceptible to self-testing.
Our study challenges this intuition by demonstrating that certain quantum states can be fully self-tested even in scenarios where some degree of communication is permitted among subsets of parties, regardless of whether they are disjoint.

From a broader perspective, such a result underscores the significance of interaction, one of the most important concepts in computational complexity theory.
In essence, the Bell scenario resembles a form of multipartite interactive proof, wherein a verifier engages with multiple, ideally non-communicating provers~\cite{cleve2004consequences}.
While provers may conspire by engaging in communication among themselves, the verifier can also thwart dishonest behavior by sending inputs not used in the certification to foil cheating.
This notion of influential yet misguided interaction is essential to our work. 

\paragraph*{Graph States.}
Our main subjects of interest are graph states, which represent a multifaceted resource in quantum information processing.
Being multipartite states with regular entanglement structure and compact form, they proved to be useful in a wide range of quantum technologies, from measurement-based~\cite{1WQC} and fusion-based~\cite{FBQC} quantum computing to topological error correction~\cite{gottesman1997stabilizer,RAUSSENDORF20062242}, quantum metrology~\cite{Nathan}, and cryptography~\cite{damiansharessecrets,christandl2005quantum}.
The experimental realization of graph states has allowed for the demonstration of controlled genuine multipartite entanglement~\cite{lu2007experimental,schwartz2016deterministic,istrati2020sequential}, and the implementation of various quantum information protocols~\cite{bell2014experimental,bell2014experimentalsecret}.
Graph states are strong resources for Bell nonlocality as every entangled graph state admits local measurements that violate a Bell inequality~\cite{Guhne,Toth}.

Using graph states, the first example exploring nonlocality against classical communication was by Barrett et al.~\cite{barrett2007modeling}.
In particular, based on extensions of the triangle graph state (locally equivalent to the GHZ state), they presented correlations based on Pauli operators in a graph state that no LHV model, even assisted by a round of bounded classical communication, can reproduce.
In~\cite{Meyer2023}, this was extended to one based on any graph state by introducing the so-called inflated graph states.

\paragraph*{Applications.}
The work by Barrett et al.~\cite{barrett2007modeling} had application in the design of distributed computation protocols~\cite{gall2018quantum}, randomness extraction~\cite{coudron2018trading}, and on establishing an unconditional proof of a computational quantum advantage~\cite{bravyi2020quantum,bravyi2018quantum}.
Moreover, a bounded communication range is not merely a theoretical device.
Field trials of measurement-device-independent (MDI) repeaters and simple quantum switches impose microsecond-level latency, so a classical signal cannot traverse more than one or two fiber spans; present MDI-QKD links ($100$–$400$ km), MDI conference-key exchange, and photon-heralded entanglement all permit feed-forward only to the next node~\cite{Lo2012,Tang2015,Boaron2018,Wang2023}.
The very same fixed-range limitation underlies constant-depth classical circuits: after $d$ layers of $K$-fan-in gates, any output depends on inputs within neighbors, so information effectively propagates only a constant distance.
Quantum–classical separation proofs such as~\cite{barrett2007modeling,bravyi2018quantum} therefore model the classical adversary as a single round of nearest-neighbor communication.
Together the communication-assisted-LHV framework matches current hardware constraints and mirrors the standard classical benchmark in shallow-circuit complexity.

\paragraph*{Results.}

We extend self-testing to settings with communication.
In particular, we show that families of graph states can be self-tested from Pauli-measurement correlations that no LHV model can explain even using bounded classical communication.
Using techniques akin to non-communicating scenarios, we give two concrete experiments: one on large circular (ring) graph states and one on the honeycomb cluster state, a universal resource for quantum computation~\cite{van2006universal}.
For general graph states, however, the potential for deception by communicating devices renders the self-testing of such states impossible.
To address this challenge, we introduce a novel approach that expands any target into an inflated graph state; observing correlations beyond any communication-assisted LHV model then guarantees that subsequent local measurements on the added vertices yield a post-measurement state locally isometric to the original.

Although the presentation of the results might imply that graph inflation is an essential tool, their significance lies more in the context where a graph - or even a subgraph - can be regarded as inflated versions of smaller ones.
Under bounded communication, full self-testing of a large graph is often out of reach, but our results show that one can still self-test an appropriate \emph{reduced} graph state.
Graph states are especially well-suited to this, since they admit local transformations and can be cropped into smaller resource states for specific tasks.
Although our current results are tailored to specialized measurements on a narrow subset of inflated graph states, we expect that combining our procedures with local complementations and related graph operations will extend self-testing to larger vertex sets.
This enables intermediate self-testing regimes certifying more than the standard core but not necessarily the entire graph, thus bridging highly structured theoretical protocols with practical, scalable implementations. 

Returning to the standard no-communication setting, we also give a more robust self-test using Pauli measurements for any graph state containing a path of at least three connected vertices.
Combined with McKague~\cite{mckague2011self}, this verifies any graph state with such a path using local Pauli devices, with improved robustness over prior results.

The structure of the article is as follows:
We commence with an introduction to self-testing and graph states in Section~\ref{sec:selftestinggraphstates}.
In Section~\ref{sec:conceptselfcommunication}, we discuss the general characteristics of self-testing graph states with communication patterns determined by an underlying graph.
This is followed by technical definitions in Section~\ref{sec:selfinflated}, and the introduction of three types of self-testing correlations in Section~\ref{sec:refexps}.
Our first result is presented in Section~\ref{sec:maintheorem}, where we demonstrate the self-testing of graph states from inflated graph states, allowing for bounded communication.
In Section~\ref{sec:full}, we highlight the ability to entirely self-test the graph state of certain symmetric graphs, even in the presence of bounded communication.
Finally, we offer concluding discussions in Section~\ref{sec:conc}.

\section{Self-testing and Graph States}
\label{sec:selftestinggraphstates}

We begin with the standard Bell scenario where $n$ \textit{non-communicating} parties share a quantum state $\ket{\Psi}$.
Each party $i$ receives a classical input $k_i \in \{0,1,\cdots,m-1\}$ and performs a measurement described by an observable $\mathcal{A}_{k_i}$ with two outcomes.
The expectation values $\langle {\Psi} \vert \mathcal{M}_{k} \vert {\Psi} \rangle$ can be estimated, where $\mathcal{M}_{k} = \bigotimes_{i=1}^n\mathcal{A}_{k_i}$ with ${k} \in \{(0,1, \cdots, m-1)^n\}$.
A tuple $(\ket{\Psi},{\mathcal{M}_{k}})$ is termed a \emph{physical experiment} (PE).
The PE remains fully uncharacterized, with all information deriving solely from the observed measurement statistics.
The goal of \emph{self-testing} is to demonstrate that, under certain circumstances, the PE can be deemed equivalent to another perfectly characterized setup known as the \emph{reference experiment} (RE).
Given the inherent noise in real measurement devices, the self-test must be resilient to minor perturbations.
Informally, a self-testing procedure exhibits robustness if it can ascertain that the PE closely resembles the RE whenever it nearly reproduces the self-testing correlations.

When establishing a relationship between a PE and a RE, the initial step is to assess their comparability. They are deemed \textit{compatible} if they entail the same number of local devices and measurements.
A crucial subsequent condition is that the PE must simulate the RE.
For our purposes, we slightly generalize the standard definition~\cite{mckague2011self}, which matches only full-measurement expectations, to consider marginal correlations as well.
This aligns with the intuition that equivalent experiments should agree on all substatistics.
\begin{defi}[Submeasurements]\label{def:submeasurement}
Given a measurement $\mathcal{M}_{k} = \bigotimes_{i=1}^n\mathcal{A}_{k_i}$ with ${k} \in \{(0,1, \cdots, m-1)^n\}$, a valid submeasurement $\mathcal{C}_{k}$ of $\mathcal{M}_k$ is defined as a tensor product $\mathcal{C}_{k} = (\mathcal{M}_k)_I = \bigotimes_{i \in I}\mathcal{A}_{k_i}$ where $I \subseteq \{1,\dots ,n \}$. 
\end{defi}
\begin{defi}[$\epsilon$-Simulation] \label{def:simrob}
Given a physical experiment $( \vert \Psi \rangle\,, \{ \mathcal{M}_{k} \})$ and a compatible reference experiment $( \vert \psi \rangle\,, \{ M_{k} \})$, the physical experiment $\epsilon$-simulates the reference experiment if, for all ${k}$, 
\begin{equation} \vert \langle \psi \vert C_{k} \vert \psi \rangle - \langle \Psi \vert \mathcal{C}_{k} \vert \Psi \rangle \vert \leq \epsilon \,, \label{eq:robdef0}\end{equation}
for all submeasurements $C_k(\mathcal{C}_{k})$ of $M_{k}(\mathcal{M}_{k})$.
\end{defi}
Features not fixed by Bell experiment, such as global rotations and extra degrees of freedom, are captured by local isometries.
\begin{defi}[$\delta$-Equivalence] \label{def:equirob} 
Let $( \vert \Psi \rangle\,, \{ \mathcal{M}_k \})$ be a physical experiment, where $\ket{\Psi}$ is a pure state on the Hilbert space $\bigotimes_{i=1}^n\mathcal{H}_i$, and $\mathcal{M}_k = \bigotimes_{i=1}^n \mathcal{A}_{k_i}$, where $\mathcal{A}_{k_i}$ are Hermitian and unitary operators on $\mathcal{H}_i$.
Given a compatible reference experiment $( \vert \psi \rangle\,, \{ M_k \})$, where $\ket{\psi}$ is a pure state on $\bigotimes_{i=1}^n\mathcal{H}'_i$ and $M_k = \bigotimes_{i=1}^n A_{k_i}$, which simulates the physical experiment.
We say that the reference experiment is $\delta$-equivalent to the physical experiment if there exist a valid quantum state $\vert \mathit{junk} \rangle$ on Hilbert space $\bigotimes_{i=1}^n\mathcal{H}''_i$ and local isometries $\{\Phi_i:\mathcal{H}_i\rightarrow \mathcal{H}^\prime_i\otimes\mathcal{H}^{\prime \prime}_i \}_{i=1}^n$ constituting $\Phi = \bigotimes_{i=1}^n\Phi_i$, such that
\begin{align} \Vert \Phi \left( \vert \Psi \rangle \right) &-\vert \psi \rangle \otimes \vert \mathit{junk}\rangle \Vert_1 \leq \delta \,,\label{ali:robdef1}\\ \Vert \Phi \left( \mathcal{A}_{k_i} \vert \Psi \rangle \right) &- A_{k_i} \vert \psi \rangle \otimes \vert \mathit{junk}\rangle \Vert_1 \leq \delta \,,\,\,\forall {k_i} \,.\label{ali:robdef2}\end{align}
\end{defi}
\begin{defi}[Robust Self-testing] \label{def:selfrob}
We say that the set of correlations $\{\langle \psi \vert C_{k} \vert \psi \rangle\}_{k}$ robustly self-test the reference experiment $( \vert \psi \rangle\,, \{ C_{k} \})$ if, for every physical experiment that $\epsilon$-simulates the reference experiment, there exists a function $\delta (\epsilon)$ continuous at $0$ such that the reference experiment is $\delta$-equivalent to the physical experiment.
\end{defi}

In the following, we focus on self-testing graph states with observables corresponding to Pauli or Clifford measurements.
A graph state $\ket{G}$ is uniquely defined by a mathematical undirected graph $G=(V,E)$, a set of vertices $V$ and edges $E \subseteq V \times V$.
Every vertex represents a qubit and the graph state is defined by \begin{equation} \vert G \rangle = \prod_{(u,v) \in E} \text{C}Z_{(u,v)} \vert + \rangle^{\otimes V}\,,\end{equation} with the Pauli $X$ eigenstate $\vert + \rangle = ( \vert 0 \rangle + \vert 1 \rangle)/\sqrt{2}$ and the entangling gate $\text{C}Z = (\vert 0 \rangle \langle 0 \vert \otimes \mathds{1} + \vert 1 \rangle \langle 1 \vert \otimes Z)$ and the Pauli $Z$ operator.
The graph state $\ket{G}$ is the simultaneous $(+1)$-eigenstate of an Abelian subgroup of the Pauli group, generated by $g_u = X_u \bigotimes_{v \in N(u)} Z_v$ for every vertex $u$ and its neighbors $N(u) = \{v : (u,v) \in E\}$.
Any stabilizer element can be written as $S = \prod_{u\in U} g_u$ for some subset $U\subseteq V$.

The stabilizers constrain any Pauli measurement's local outcomes.
For any non-trivially connected graph, the resulting correlations are incompatible with any local-hidden-variable model~\cite{Guhne,scarani2005nonlocality}.
McKague established robust self-tests for all connected graph states with bounds $\delta(\epsilon)=\mathcal{O}(\sqrt{\epsilon})$ or $\mathcal{O}(\epsilon^{1/4})$~\cite{mckague2011self,McKague2016}, implicitly using the proposition below.
\begin{prop}\label{prop:selftestgraphstates}
For self-testing graph states, it suffices that the PE:
\begin{enumerate}
\item[(i)] ($\epsilon$-)certify a qubit locally: two binary local observables ($\epsilon$-)anticommute on the state.
\item[(ii)] These observables ($\epsilon$-)reproduce the graph state's stabilizers: at least a generating set (e.g.\@ the generator elements).
\end{enumerate}
\end{prop}
Condition (i) is enforced by the nonlocal correlations and the simulation condition (Def.\,\ref{def:simrob}), which together imply local anticommutation on the PE's state.
For condition (ii), the RE tests the generator elements $g_u$, and the same observables show to reproduce their correlations.
With these properties, one constructs a local isometry that implements a SWAP gate between the PE's state and an auxiliary qubit for every vertex~\cite{mayers98} - its circuit description is shown in Fig.\,\ref{fig:swapcirc}.
The resulting state on the auxiliary qubit is proven to be $(\delta$-)equivalent to the RE graph state.
In Appendix~\ref{app:line}, we provide an self-test for a three-vertex line using Greenberger-Horn-Shimony-Zeilinger (GHSZ)-correlations~\cite{greenberger1990bell}.
This new self-test has robustness of $\delta(\epsilon) \sim \mathcal{O}(\sqrt{\epsilon})$, with a slight improvement in scaling for some cases compared to~\cite{mckague2011self,McKague2016}.
To self-test all graph states, the remaining case of a pair of vertices and those where the graph consists only of triangles, are already covered in~\cite{mckague2011self}.
We use these ideas in the following to show that some graph states can be self-tested even if some classical communication is permitted, which lifts the devices' strict locality.

\begin{figure}
\centering
\begin{quantikz}
\lstick{$\ket{ 0}$}& \gate{H} & \ctrl{1} & \gate{H} & \ctrl{1} & \qw \\
\lstick{$\ket{\psi}$}& \qw & \gate{Z} & \qw & \gate{X} &\qw\\
\end{quantikz}
\caption{Circuit implementing a SWAP gate between $\vert \psi \rangle$ and $\vert 0 \rangle$ using the Hadamard gate $H$ and Pauli operators $X,Z$.}
\label{fig:swapcirc}
\end{figure}

\section{Self-Testing Graph States Against Bounded Classical Communication}

This section states our main results; Subsection~\ref{sec:conceptselfcommunication} adapts self-testing to the setting with a single round of classical communication.
Subsection~\ref{sec:selfinflated} introduces inflated graphs and their inflated generator elements, and Subsection~\ref{sec:example} gives an explicit example of these in terms of a nonlocal paradox.
Subsection~\ref{sec:refexps} presents three reference experiments (REs) covering all graph states and useful for robustness.
Subsection~\ref{sec:maintheorem} gives a general self-test via inflation and deflation.
Finally, Subsection~\ref{sec:full} shows that certain symmetric graph states can be self-tested entirely even with bounded communication.

\subsection{Self-testing against communication}
\label{sec:conceptselfcommunication}

We extend the concepts in Section~\ref{sec:selftestinggraphstates} to one round of communication between devices.
Integrating communication into the self-testing framework primarily entails precisely defining the physical experiment (PE), particularly determining the observables associated with the measurements performed.
In the usual setting, observables are indexed only by local inputs, but with communication, they might also depend on inputs received from neighbors within a communication range.

Following~\cite{barrett2007modeling,Meyer2023}, we let the allowed communication be dictated by the graph state; each device broadcasts its inputs to vertices within graph distance $d$.
The extension of LHV models to this communication setting is denoted as $d$-LHV$^\ast$ models.
Our tests contradict any $d$-LHV$^\ast$ model in the sense of~\cite{Meyer2023}, matching earlier non-classicality frameworks~\cite{barrett2007modeling}.

For the PE, let $G=(V,E)$ and suppose parties at vertices share $\ket{\Psi}$.
Under bounded communication the input to the binary observable $\mathcal{A}_{k^{(d)}_u}$ is the tuple $k^{(d)}_u=(k_u,(k_v)_{v\in N_d(u)})$.
These constitute a global measurement $\mathcal{M}_k=\bigotimes_{u\in V}\mathcal{A}_{k^{(d)}_u}$ with $k\in\{0,1,\dots,m-1\}^{|V|}$, including submeasurements $\mathcal{C}_k=(\mathcal{M}_k)_{U_k}$ for $U_k\subseteq V$.
Compared to the no-communication case, the local setting alphabet is larger.
When comparing a PE to an RE, we can, without loss of generality, index the RE’s measurements by the same input tuples $k$, even if the RE observables themselves are independent of communicated inputs.
Within this framework, we use Definitions~\ref{def:simrob} and~\ref{def:equirob} to analyze self-testing of a target graph state $\ket{G}$. We give two methods:

In the first approach, applicable to graphs exhibiting suitable symmetry, we specify an RE consisting of the graph state $\vert G\rangle$ and a set of local measurements.
We use these correlations to self-test the reference state, even when the devices are allowed to communicate up to distance $d$ on $G$.
Nonetheless, this method is contingent upon certain symmetrical properties, and its applicability to all graph states cannot hold.
This cannot hold for all graphs (e.g., a complete graph reveals all inputs to all parties, enabling classical spoofing).

The second approach is universally applicable to all graph states.
This method inflates the graph to $G'$, defines communication with respect to $G'$, and chooses measurements that both prepare $\vert G \rangle$ from $\vert G' \rangle$ by measuring the added chain vertices (deflation), see Fig.~\ref{fig:graphmap}, and furthermore provides correlations that self-test the resulting state
In this way, if a PE faithfully simulates this RE, the equivalence between the post-measurement (deflated) physical state and the target graph state can be established, even in scenarios permitting bounded communication among the parties involved.
In this sense, the protocol both prepares the target graph state, and self-tests it at the same time.
Since the first approach uses ideas and techniques from the second approach, we start by presenting the second, more general approach.

\subsection{Inflated graph states}
\label{sec:selfinflated}

\begin{figure}
\centering
\def \abovetext{\small Def.\:\ref{def:inflatedgraph}}
\def \belowtext{\small$\{X_{v} \}_{v \in V\setminus V^{\prime}}$}
\def \bbelowtext{\small \hspace{-0.25cm}with outcome $(x)$}

\begin{tikzpicture}[scale=0.8]

\begin{scope}[yshift=-0.5cm,xshift=-2.5cm,scale=0.7]

\node[draw=none] (0) at (-3.5,0) {\large$G^{\prime}$};
\node[draw=none] (0) at (3.5,0) {\large$G$};

\draw[->] (-1.5,-0.3) node[below=0.4cm,anchor=west]{\belowtext}node[below=0.8cm,anchor=west]{\bbelowtext} -- (2,-0.3);
\draw[<-] (-1.5,0.3) -- node[above=0.1cm]{\abovetext} (2,0.3);

\end{scope}
\begin{scope}[xshift=0cm,rotate=90,every node/.style={circle,draw}]

\node[line width=0.4mm,label={[label distance=0.05cm]270: $2$}] (1) at (1.5,6) {$\phantom{x}$};
\node (121) at (2.25,6) {};
\node (122) at (2.75,6) {};
\node (123) at (3.25,6) {};
\node (124) at (3.75,6) {};
\node[line width=0.4mm,label={[label distance=0.05cm]90: $3$}] (2) at (4.5,6) {$\phantom{x}$};

\node (241) at (4.5,5.25) {};
\node (242) at (4.5,4.75) {};
\node (243) at (4.5,4.25) {};
\node (244) at (4.5,3.75) {};

\node (131) at (1.5,5.25) {};
\node (132) at (1.5,4.75) {};
\node (133) at (1.5,4.25) {};
\node (134) at (1.5,3.75) {};

\node[line width=0.4mm,label={[label distance=0.05cm]270: $1$}] (3) at (1.5,3) {$\phantom{x}$};

\node (141) at (2.25,5.25) {};
\node (142) at (2.75,4.75) {};
\node (143) at (3.25,4.25) {};
\node (144) at (3.75,3.75) {};

\node[line width=0.4mm,label={[label distance=0.05cm]90: $4$}] (4) at (4.5,3) {$\phantom{x}$};

\draw[-] (1) -- (121); 
\draw[-] (121) -- (122);
\draw[-] (122) -- (123);
\draw[-] (123) -- (124);
\draw[-] (124) -- (2);

\draw[-] (2) -- (241); 
\draw[-] (241) -- (242);
\draw[-] (242) -- (243);
\draw[-] (243) -- (244);
\draw[-] (244) -- (4);

\draw[-] (1) -- (131);
\draw[-] (131) -- (132);
\draw[-] (132) -- (133);
\draw[-] (133) -- (134);
\draw[-] (134) -- (3);

\draw[-] (1) -- (141);
\draw[-] (141) -- (142);
\draw[-] (142) -- (143);
\draw[-] (143) -- (144);
\draw[-] (144) -- (4);

\end{scope}

\begin{scope}[xshift=3cm,yshift=0.75cm,rotate=90,scale=0.7,every node/.style={circle,draw}]

\node[line width=0.4mm,label={[label distance=0.05cm]270: $2$}] (1) at (1.5,6) {$\phantom{x}$};
\node[line width=0.4mm,label={[label distance=0.05cm]90: $3$}] (2) at (4.5,6) {$\phantom{x}$};
\node[line width=0.4mm,label={[label distance=0.05cm]270: $1$}] (3) at (1.5,3) {$\phantom{x}$};
\node[line width=0.4mm,label={[label distance=0.05cm]90: $4$}] (4) at (4.5,3) {$\phantom{x}$};

\draw[-] (1) -- (2);
\draw[-] (2) -- (4); 
\draw[-] (1) -- (3);
\draw[-] (1) -- (4);

\end{scope}

\end{tikzpicture}
\caption{A graph $G$ with four vertices (rhs) is mapped to an exemplary \emph{inflated} graph $G'$ for $d=2$ (lhs) by means of Def.~\ref{def:inflatedgraph}, i.e., by adding $2d$ `chain' vertices along the edges.
The inverse mapping `deflation' measures the chain vertices projectively (up to local corrections) to recover $\ket{G}$ from $\ket{G'}$.
In Theorem~\ref{theo:main}, the RE uses $\ket{G'}$ and local measurements to both prepare and self-test $\ket{G}$.}
\label{fig:graphmap}
\end{figure}

Following~\cite{Meyer2023}, inflating $G=(V,E)$ inserts $2d$ chain vertices along each edge; original vertices are `power' vertices.
\begin{defi}[Inflated graph]\label{def:inflatedgraph}
Given $G = (V,E)$, its $d$-inflated graph $G^{\prime}=(V^{\prime}, E^{\prime})$ replaces each edge with a connected chain of $2d$ vertices bridging the original edge.
The new vertices are $V^{\prime} = V \cup W$ with $W = \{ r_{(u,v)} ;\,r=1,\cdots, 2d, (u,v)\in E \}$, and the new edges are $(u,1_{(u,v)}),(1_{(u,v)},2_{(u,v)}),\cdots, ((2d)_{(u,v)},v) \in E^{\prime} $.
\end{defi}

Deflation, i.e., measuring chain vertices of $\vert G' \rangle$ in Pauli $X$ or $Y$, yields $\vert G \rangle$ up to local Pauli corrections by standard graph-state measurement rules~\cite{hein2004multiparty}.
For every $d$-inflated graph state, there exist Pauli measurements that generate correlations inconsistent with any $d$-LHV$^\ast$ models~\cite{Meyer2023}.
For a graph with at least three connected vertices, this inconsistency is demonstrated by means of a GHSZ-like paradox~\cite{greenberger1990bell}, and for two vertices via an inequality akin to the Clauser-Horne-Shimony-Holt (CHSH) inequality \cite{chsh1969}. 

We denote the generator elements of the inflated graph states as $\{ g^{\prime}_{u} \}_{u \in V^{\prime}}$ and introduce the \emph{inflated generator elements} that mimics a generator of $\vert G \rangle$ on the power vertices of $\vert G' \rangle$, see Fig.~\ref{fig:graphex}.
\begin{figure}
\centering
\begin{tikzpicture}[scale=0.9]
\tikzstyle{power}=[draw=black, very thick, shape=circle, minimum size=16pt]
\tikzstyle{chain}=[draw=black, shape=circle, minimum size=9pt,inner sep=1.5pt]
\begin{scope}[xshift=0cm,rotate=90,every node/.style={circle,draw}]

\node[line width=0.3mm,label={[label distance=0.05cm]270: $2$}] (1) at (1.5,6) {\footnotesize$Z$};
\node [style=chain] (121) at (2.25,6) {};
\node [style=chain] (122) at (2.75,6) {};
\node [style=chain] (123) at (3.25,6) {};
\node [style=chain] (124) at (3.75,6) {};
\node[line width=0.3mm, label={[label distance=0.05cm]90: $3$}] (2) at (4.5,6) {\footnotesize$Z$};

\node [style=chain] (241) at (4.5,5.3) {\tiny $X$};
\node [style=chain] (242) at (4.5,4.75) {};
\node [style=chain] (243) at (4.5,4.2) {\tiny $X$};
\node [style=chain] (244) at (4.5,3.65) {};

\node [style=chain] (131) at (1.5,5.25) {};
\node [style=chain] (132) at (1.5,4.75) {};
\node [style=chain] (133) at (1.5,4.25) {};
\node [style=chain] (134) at (1.5,3.75) {};

\node[line width=0.3mm,label={[label distance=0.05cm]270: $1$}] (3) at (1.5,3) {\footnotesize$\phantom{X}$};

\node [style=chain] (141) at (2.25,5.25) {\tiny $X$};
\node [style=chain] (142) at (2.75,4.75) {};
\node [style=chain] (143) at (3.25,4.25) {\tiny $X$};
\node [style=chain] (144) at (3.75,3.75) {};

\node[line width=0.3mm,label={[label distance=0.05cm]90: $4$}] (4) at (4.5,3) {\footnotesize$X$};

\draw[-] (1) -- (121); 
\draw[-] (121) -- (122);
\draw[-] (122) -- (123);
\draw[-] (123) -- (124);
\draw[-] (124) -- (2);

\draw[-] (2) -- (241); 
\draw[-] (241) -- (242);
\draw[-] (242) -- (243);
\draw[-] (243) -- (244);
\draw[-] (244) -- (4);

\draw[-] (1) -- (131);
\draw[-] (131) -- (132);
\draw[-] (132) -- (133);
\draw[-] (133) -- (134);
\draw[-] (134) -- (3);

\draw[-] (1) -- (141);
\draw[-] (141) -- (142);
\draw[-] (142) -- (143);
\draw[-] (143) -- (144);
\draw[-] (144) -- (4);

\end{scope}

\begin{scope}[xshift=3cm,yshift=0.75cm,rotate=90,scale=0.7,every node/.style={circle,draw}]

\node[line width=0.3mm,label={[label distance=0.05cm]270: $2$}] (1) at (1.5,6) {\footnotesize$Z$};
\node[line width=0.3mm, label={[label distance=0.05cm]90: $3$}] (2) at (4.5,6) {\footnotesize$Z$};
\node[line width=0.3mm,label={[label distance=0.05cm]270: $1$}] (3) at (1.5,3) {\footnotesize$\phantom{X}$};
\node[line width=0.3mm,label={[label distance=0.05cm]90: $4$}] (4) at (4.5,3) {\footnotesize{$X$}};

\draw[-] (1) -- (2);
\draw[-] (2) -- (4); 
\draw[-] (1) -- (3);
\draw[-] (1) -- (4);

\end{scope}

\end{tikzpicture}
\caption{A graph with four vertices (rhs) and its ($d=2$)-inflated graph (lhs) by means of Def.~\ref{def:inflatedgraph}.
Letters indicate non-trivial Pauli operators
while empty nodes indicate Pauli identity in the Pauli product of the inflated generator element~$f_{4}$ from Def.~\ref{def:infvstab} (lhs), mimicking the generator element~$g_{4}$ (rhs) on the corresponding graph state.}
\label{fig:graphex}
\end{figure}
\begin{defi}[Inflated Generator Element]
The inflated generator element $f_{u}$ is a stabilizer element of the inflated graph state with
\begin{equation}
f_{u} = g^{\prime}_{u} \prod_{\substack{(u,v)\in E,\\s=1,\dots,d}} g^{\prime}_{2s_{(u,v)}}\,.\label{eq:infstabgenerator}
\end{equation}
\label{def:infvstab}
\end{defi}

\subsection{Introductory Example}
\label{sec:example}

The triangle graph state, which is locally equivalent to the GHZ state, can be self-tested via Pauli measurements from its stabilizer elements~\cite{mckague2011self}:
\begin{equation}
\begin{aligned}
g_{1} &= X_{1}Z_{2}Z_{3}\,, \\
g_{2} &= Z_{1}X_{2}Z_{3} \,, \\
g_{3} &= Z_{1}Z_{2}X_{3} \,, \\
- g_{1}g_{2}g_{3} &= X_{1}X_{2}X_{3} \,.
\end{aligned}
\end{equation}
These four measurements together with the triangle graph state constitute the reference experiment (RE).
Then, the Bell test is the protocol obtained directly from the RE: a referee samples a global question from the four measurements; the local input to each party is either `measure $X$' or `measure $Z$'.
The output correlations included in the physical experiment (PE) can, if they are those expected from the RE, certify qubits (local anticommutation) and, by testing the generating set ($g_{1},g_{2},g_{3}$), pin the state up to local isometries.

In the presence of bounded communication, the global non-signaling assumption is broken, but the power vertices remain effectively non-signaling relative to one another.
Applying results from~\cite{Meyer2023}, one obtains an inflated-state paradox robust to $d$-range communication with stabilizer elements, without specifying $V'_{i}$:
\begin{equation}
\begin{aligned}
f_{1} &= X_{1}Z_{2}Z_{3} \, X^{\otimes V'_{1}} \,,\\
f_{2} &= Z_{1}X_{2}Z_{3}\, X^{\otimes V'_{2}}\,, \\
f_{3} &= Z_{1}Z_{2}X_{3}\, X^{\otimes V'_{3}}\,, \\
-f_{1}f_{2}f_{3} &= X_{1}X_{2}X_{3}\, X^{\otimes V'_{1}+V'_{2}+V'_{3}} \,.
\end{aligned}
\end{equation}
Note that the Pauli operators on the power vertices $1,2,3$ match those of the triangle.
In the Bell experiment, a referee samples from these four measurement, and gives the inputs $k_u$ (`measure $X$' or `measure $Z$') to the parties.
With communication each party $u$ also receives the inputs broadcast by all vertices in $N_d(u)$, so its local input is the tuple $k_u^{(d)}=(k_u,(k_v)_{v\in N_d(u)})$.
The output correlations are compared to the RE’s ideal values (with tolerance $\epsilon$).

As subtle point, if we literally instructs the chain vertices to `measure Pauli-$X$' or `measure $\mathds{1}$ (do not measure)', the power vertices can infer which global setting is used from communicated inputs.
To remove that classical advantage, we always ask all chain vertices to measure Pauli-$X$ and then discard outcomes for those positions that should be Pauli-$\mathds{1}$ in the corresponding stabilizer.
This is exactly the submeasurement idea in~\cite{barrett2007modeling}.
For some graphs and paradoxes (including the triangle), classical strategies still gain leverage.
We thwart this by adding paired settings that use the submeasurement method on power vertices as needed.

\subsection{Reference experiments}\label{sec:refexps}

\begin{figure*}
\centering
\input{tikz/tikz-big-new}
\caption{
Four panels of inflated graphs that address the ones in REs~\ref{def:refe1} (a), REs~\ref{def:refe2}, for $\vert V_{s} \vert >2$ (b), for $\vert V_{s} \vert = 2$ (d) and RE~\ref{def:refe3} (c).
The bold nodes depict power vertices.
Within a panel, each graph depicts a measurement setting and its submeasurement.
The measurements are those that challenge any $d$-LHV$^\ast$ model.
The letters in black and light-blue font indicate he bases in the measurement setting.
For the submeasurements, replace the letters in light-blue font by Pauli identity.
Graphs with a circular arrow represent all measurements under invariant rotations of the power vertices.\\
For RE~\ref{def:refe1}, panel (a) shows $F_{v_{1}}$ with $f_{v_{1}}$ from Eq.\;\eqref{eq:infstabgenerator} (upper left), while the circular arrow hints at $F_{v_{2}}$ with $f_{v_{2}}$, and $F_{v_{3}}$ with $f_{v_{3}}$, $M_{V_{c}}$ with $C_{V_{c}}$ (upper right) from Eq.\;\eqref{eq:measrefstab12}, and $M_{v_{1}}^{(X)}$ (lower left) and $M_{v_{1}}^{(Z)}$ (lower right) with $C_{v_{1}}$ from Eq.\;\eqref{eq:measrefstab13}, and the circular arrow hints at $M_{v_{2}}^{(X)}$, $M_{v_{2}}^{(Z)}$ with $C_{v_{2}}$, and $M_{v_{3}}^{(X)}$, $M_{v_{3}}^{(Z)}$ with $C_{v_{3}}$.\\
For RE~\ref{def:refe2} ($\vert N(v_{c}) \vert = 3$), panel (b) shows $\Tilde{F}_{v_{c}}$ with $f_{v_{c}}$ from Eq.\;\eqref{eq:infstabgenerator} (left), and $M_{v_{2},v_{3}}$ with $C_{v_{2},v_{3}}$ from Eq.\;\eqref{eq:ref21}, and the circular arrow hints at $M_{v_{1},v_{3}}$ with $C_{v_{1},v_{3}}$, $M_{v_{1},v_{2}}$ with $C_{v_{1},v_{2}}$.\\
For RE~\ref{def:refe3}, panel (c) shows, from top to bottom, $M_{i}$ with $I_{i}$ in Eq.\;\eqref{eq:ref3} for $i = 1,2,3,4$ and $\Tilde{F}_{v_{l}}$ with $f_{v_{l}}$, $\Tilde{F}_{v_{r}}$ with $f_{v_{r}}$.\\
For RE~\ref{def:refe2} ($\vert N(v_{c}) \vert = 2$), panel (d) shows, from top to bottom, $\Tilde{F}_{v_{c}}$ with $f_{v_{c}}$ in Eq.\;\eqref{eq:infstabgenerator}, $M_{v_{1}}$ with $C_{v_{1}}$ in Eq.\;\eqref{eq:ref22-1}, $M_{v_{2}}$ with $C_{v_{2}}$ in Eq.\;\eqref{eq:ref22-2}, $M_{v_{c}}$ with $C_{v_{c}}$ in Eq.\;\eqref{eq:ref22-3}, $M_{2}^{(X)}$ and $M_{2}^{(Y)}$ with $C_{2}$ in Eq.\;\eqref{eq:ref22-4}.
}
\label{fig:refexp}
\end{figure*}

We introduce three reference experiments (REs) that comprehensively address any inflated graph state.
Each RE consists of the inflated graph state itself, measurements to test the inflated generator element, and an additional set of measurements designed challenge any $d$-LHV$^\ast$ model.
For clarity and simplicity, we opt to present the submeasurements alongside the parent measurements.

The proof strategy follows Prop.~\ref{prop:selftestgraphstates}, use paradox-type submeasurements to establish a local anticommutation pair on at least one power vertex, and then, propagate anticommutation across power vertices via inflated stabilizer submeasurements.
Finally, we build a local SWAP-style isometry.
This process holds true for any degree of inflation $d$.

\begin{defi}[Induced subgraph]
An induced subgraph of a graph $G=(V,E)$ is a subset of vertices $U \subset V$ along with the edges $(U \times U) \cap E$ of the original graph if and only if the edge connects two vertices in the subset $U$. 
\end{defi}
The RE~\ref{def:refe1} is applicable when the corresponding graph contains an odd circle (in particular a triangle) as an induced subgraph.
The RE~\ref{def:refe1} handles graphs with an induced star subgraph, i.e., a central vertex whose nearest neighbors do not share an edge.
These two encompass any graph with at least three vertices.
The RE~\ref{def:refe3} covers a neighboring vertex pair (in particular, the 2-vertex graph), and in fact can cover any connected graph.
Still, REs~\ref{def:refe1}–\ref{def:refe2} remain useful for improved robustness (Appendix~\ref{app:robref3}).
Indeed, the different REs can be in principle combined in different ways to optimize robustness for example. Robustness is treated in Appendix~\ref{app:robref3}.

For all three REs, the local measurements on the chain vertices, which deflate the inflated graph state as described in Eq.\;\eqref{eq:deflationprojector1} of Theorem~\ref{theo:main}, correspond to measurements in the Pauli $X$ basis, as depicted in Fig.~\ref{fig:graphmap}.
Some settings in REs~\ref{def:refe2}–\ref{def:refe3} also require Pauli-$Y$ on chains. 
We present the submeasurements $C$ in stabilizer form below, and their tensor-product forms are in the Appendix~\ref{app:expect}.
Recall that $g^{\prime}_{u}$ for $u \in V^{\prime}$ denotes a generator element of $\vert G' \rangle$, $f_{v}$ denotes an inflated generator element from Eq.\;\eqref{eq:infstabgenerator}.

\begin{refexp}[Odd Circular Subgraph]
\label{def:refe1}
Let $G^{\prime}$ denote the inflated graph derived from a connected graph $G =(V,E)$ containing a induced circular subgraph $G_{c} =(V_{c},E_{c})$ with an odd number of vertices, and let $\vert G^{\prime}\rangle$ denote the corresponding inflated graph state $\vert G^{\prime}\rangle$.
The reference experiment consists of the state $\vert G^{\prime}\rangle$ and the Pauli measurements defined in Tab.~\ref{tab:remeas}: $M_{V_{c}}$, and $M_{v}^{(X)}$, $M_{v}^{(Z)}$ for all vertices $v \in V_{c}$, as well as $F_{u}$ for all vertices $u \in V$.
Overall the RE consists of $\vert V \vert + 1 + 2\vert V_{c} \vert$ measurement settings.
\end{refexp}
\begin{refexp}[Star Subgraph]
\label{def:refe2}
Let $G^{\prime}$ denote the inflated graph derived from a connected graph $G = (V,E)$ that contains an induced star subgraph, i.e., a vertex $v_{c}$ with at least two adjacent vertices that do not share an edge among them.
Let $\vert G^{\prime}\rangle$ denote the corresponding inflated graph state.
If $\vert N(v_{c})\vert > 2$, we select a set of three vertices $V_{3} = \{ v_{1},v_{2},v_{3} \} \subseteq N(v_{c})$, else $N(v_{c}) = \{v_{1},v_{2} \}$, i.e., $\vert N(v_{c}) \vert = 2$.
The reference experiment consists of the state $\vert G^{\prime} \rangle$, and the Pauli measurements defined in Tab.~\ref{tab:remeas}:
The $F_{u}$ for all vertices $u \in V$ and $\Tilde{F}_{v_{c}}$.
There as two distinct cases for which the RE furthermore contains the Pauli measurements.
If $\vert N(v_{c})\vert > 2$, they are $M_{v_{i}v_{j}}$ for $(v_{i},v_{j}) \in \{ (v_{1},v_{2}),(v_{1},v_{3}),(v_{2},v_{3}) \}$.
Overall the RE consists of $\vert V \vert + 4$ measurement settings.
If $\vert N(v_{c}) \vert = 2$, the measurements are $M_{v_{c}}$, $M_{v_{1}}$, $M_{v_{2}}$, $M_{2}^{(X)}$, and $M_{2}^{(Y)}$.
Overall the RE consists of $\vert V \vert + 6$ measurement settings.
\end{refexp}
\begin{refexp}[Vertex Pair] \label{def:refe3}
Let $G^{\prime}$ denote the $d$-inflated graph derived from a connected graph $G = (V,E)$ with two neighboring vertices $v_{l}$ and $v_{r}$, and $\vert G^{\prime} \rangle$ denote the corresponding inflated graph state.
The reference experiment consists of the state $\vert G^{\prime} \rangle$, the Pauli measurements $F_{u}$ for all $u \in V$, and $\Tilde{F_{v_l}}$, $\Tilde{F_{v_r}}$, as well as the tensor product operators $M_{i}$ for $i=1,2,3,4$, with the phase gate $ \sqrt{Z}$ applied on vertex $v_m = d_{(v_{l},v_{r})}$.
Overall the RE consists of $\vert V \vert + 6$ measurement settings.
\end{refexp}
For the submeasurements, denote
\[h^{U}_{u} = \prod_{\substack{v \in N(u)\setminus U,\\s=1,\dots,d}} g^{\prime}_{2s_{(u,v)}}\]
for a subset of vertices~$U$ to describe the measurements on vertices adjacent to the induced subgraph used in the REs.
Then, the submeasurements of interest for the corresponding measurements, as defined in Tab.~\ref{tab:remeas}, include, for RE~\ref{def:refe1},
\begin{align}
C_{V_{c}} &= (-1) \prod_{v \in V_{c}} f_{v} \,,\label{eq:measrefstab12}\\
C_{u} &= h^{V_{c}}_{u} \prod_{v \in V_{c}\setminus \{ u \}} f_{v} \,, \label{eq:measrefstab13}
\end{align}
for RE~\ref{def:refe2},
\begin{align}
\label{eq:ref21}C_{v_{i}v_{j}} &= (-1)\, g^{\prime}_{v_{c}} f_{v_{i}} f_{v_{j}} h^{\{ v_{i},v_{j} \}}_{v_{c}}\,, \\[0.12cm]
C_{v_{1}} &= \phantom{-}g^{\prime}_{v_{c}} f_{v_{1}} h^{\{ v_{1} \}}_{v_{c}} \,, \label{eq:ref22-1} \\[0.13cm]
C_{v_{2}} &= \phantom{-}g^{\prime}_{v_{c}} h^{\{ v_{2} \}}_{v_{c}} f_{v_{2}} \,, \label{eq:ref22-2} \\[0.13cm]
C_{v_{c}} &= - \, g^{\prime}_{v_{c}} f_{v_{1}} f_{v_{2}} \,, \label{eq:ref22-3} \\[0.13cm]
C_{2}\, &= \phantom{-}f_{v_{1}} f_{v_{2}} h_{v_{c}} \,, \label{eq:ref22-4}
\end{align}
and, for RE~\ref{def:refe3},
\begin{align}
I_i &= \frac{1}{\sqrt{2}} \left( P_i + \sqrt{Z}^{\phantom{\dagger}}_{v_m} P_i \, \sqrt{Z}^\dagger_{v_m} \right)\,. \label{eq:ref3}
\end{align}
We denote $v_m := d_{(v_{l},v_{r})}$ and define
\begin{align*}
P_{1} &= \left(-g^{\prime}_{v_m}\right)^{d+1} f_{v_{r}}
&P_{3} &= \left(g^{\prime}_{v_m}\right)^{d\phantom{+1}} \Tilde{f}_{v_{l}}^{(\mathrm{e})} \Tilde{f}^{(\mathrm{e})}_{v_{r}} \\
P_{2} &= \left(-g^{\prime}_{v_m}\right)^{d\phantom{+1}} \, f_{v_{l}} \,,
&P_{4} &= \left(g^{\prime}_{v_m}\right)^{d+1} \Tilde{f}_{v_{l}}^{(\mathrm{o})} \Tilde{f}^{(\mathrm{o})}_{v_{r}} \,,
\end{align*}
with
\begin{align*}
\Tilde{f}^{(\mathrm{e})}_{v_a} &= g^{\prime}_{v_a} h^{\{ v_b \}}_{v_a} \hspace{-0.05cm}\prod^{\lfloor d/2 \rfloor}_{s=1} g^{\prime}_{2s_{(v_a,v_b)}} \,,\\
\Tilde{f}^{(\mathrm{o})}_{v_a} &= h^{\{ v_b \}}_{v_a} \hspace{-0.05cm} \prod^{\lfloor d/2 \rfloor}_{s=1} g^{\prime}_{(2s-1)_{(v_a,v_b)}} \,,
\end{align*}
for $a,b \in \{l,r\}$.

For every RE and a minimal induced subgraph, Figs.~\ref{fig:exref1}\,-~\ref{fig:exref4} depict examples of measurement settings and corresponding submeasurements.
In Appendix~\ref{app:expect}, we write down the measurement settings and submeasurements for each respective physical experiment. 

\begin{table}[ht!]
\begin{tabular}{Sr|Sc|Sc|Sr}
$M$ & Pauli-$Z$ & Pauli-$Y$ & \multicolumn{1}{Sl}{$C$}  \\
\cline{1-4} all   $F_{u}$ & $N(u)$ & & $f_{u}$ (Eq.\,\eqref{eq:infstabgenerator}) \\
\cline{1-4} \multicolumn{4}{Sl}{RE\,\ref{def:refe1}} \\
\cline{1-4} all & $N(V_{c})$ & &  \\
$M_{V_c}$ & & & $C_{V_{c}}$~\eqref{eq:measrefstab12} \\
$M_{u}^{(X)}$ & & $N(u) \cap V_{c}$ & $C_{u}$~\eqref{eq:measrefstab13}\\
$M_{u}^{(Z)}$ & $\{ u \}$ & $N(u) \cap V_{c}$ & $C_{u}$~\eqref{eq:measrefstab13} \\
\cline{1-4} \multicolumn{4}{Sl}{RE\,\ref{def:refe2} ($\vert N(v_{c}) \vert > 2$)}   \\
\cline{1-4} all & $N(V_{s})$ & $N^{\prime}(v_c)$ &  \\
$\Tilde{F}_{v_{c}}$ & $N(v_{c})$ & & $f_{v_{c}}$~\eqref{eq:infstabgenerator} \\
$M_{v_{i}v_{j}}$ & & & $C_{v_{i}v_{j}}$~\eqref{eq:ref21} \\
\cline{1-4} \multicolumn{4}{Sl}{RE\,\ref{def:refe2} ($\vert N(v_{c}) \vert = 2$)}  \\
\cline{1-4} all & $N(V_{s})$ & $N^{\prime}(v_c)$ &  \\
$\Tilde{F}_{v_{c}}$ & $N(v_{c})$ & & $f_{v_{c}}$~\eqref{eq:infstabgenerator} \\
$M_{v_c}$ & & & $C_{v_{c}}$~\eqref{eq:ref22-1} \\
$M_{v_1}$ & $\{v_{2}\}$ & & $C_{v_{1}}$~\eqref{eq:ref22-2} \\
$M_{v_2}$ & $\{v_{1}\}$ & & $C_{v_{2}}$~\eqref{eq:ref22-3} \\
$M_{2}^{(X)}$ & & $\{ v_{1},v_{2} \}$ & $C_{2}$~\eqref{eq:ref22-4} \\
$M_{2}^{(Y)}$ & &$\{ v_{c},v_{1},v_{2} \}$& $C_{2}$~\eqref{eq:ref22-4} \\
\cline{1-4} \multicolumn{4}{Sl}{RE\,\ref{def:refe3}} \\
\cline{1-4} all & $N(\{ v_{l},v_{r} \})$ & $N^{\prime}(v_{m})$ & \\
$\Tilde{F}_{v_{l}}$ & $\{ v_{l} \}$ & & $P_{2}$~\eqref{eq:ref3} \\
$\Tilde{F}_{v_{r}}$ & $\{ v_{r} \}$ & & $P_{1}$~\eqref{eq:ref3} \\
$(\ast)$ $M_{1}$ & $\{ v_{l} \}$ & & $I_{1}$~\eqref{eq:ref3} \\
$(\ast)$ $M_{2}$ & $\{ v_{r} \}$ & & $I_{2}$~\eqref{eq:ref3}\\
$(\ast)$ $M_{3}$ & & & $I_{3}$~\eqref{eq:ref3} \\
$(\ast)$ $M_{4}$ & $\{ v_{l},v_{r} \}$ & & $I_{4}$~\eqref{eq:ref3}
\end{tabular}
\caption{We list the measurement settings $M$ of the reference experiments (RE)~\ref{def:refe1}\,-\,\ref{def:refe3} in terms of the local measurement bases together with submeasurements $C$.
By default, all vertices measure in Pauli-$X$ basis.
For all other measurement bases (Pauli-$Y$/$Z$), the table contains the vertices for each measurement setting.
For example, in measurement setting $M_{2}^{(Y)}$, all vertices measure in the Pauli-$X$ basis, except for: the power vertices $\{ v_{1}, v_{2} \}$ (Pauli-$Z$), the two chain vertices $N^{\prime}(v_{c})$ as well as $\{ v_{c},v_{1},v_{2} \}$ (Pauli-$Y$).
($\ast$) The $M_{1}$, $M_{2}$, $M_{3}$, $M_{4}$ in RE~\ref{def:refe3} measure $R_{z}= \left( X+Y \right)/\sqrt{2}$ at vertex $v_{m}$.}
\label{tab:remeas}
\end{table}

To make statements about the ability to self-test, we define the explicit Bell experiment.
For a fixed RE, consider its number of devices, a referee issues finite list of global question that are exactly the list of named measurements (e.g., $F_u$, $\tilde{F_{v_c}}$, $M^{(X)}_v$, $M^{(Z)}_v$, $M_{v_iv_j}$, etc.). 
Before measuring, each party receives the inputs broadcast by all vertices within graph distance $d$ along the RE's graph $G$. 
For each global question, the RE’s Pauli tensor product uniquely determines the local observable at every site given the communicated inputs.
The number of pre-communication inputs for each party is the number of different operators across all the global measurements, e.g. a single Pauli-$X$ on chain vertices, Pauli-$X,Y,Z$ on power vertices.
After a round of classical communication, the local input also includes the inputs from the vertices in communication range - these vary from range, graph and RE.
When a submeasurement is tested, the referee still asks the parent global question, parties measure everywhere, and outcomes on vertices where the submeasurement operator has $\mathds{1}$ are discarded in post-processing. 
The observed correlations from this Bell procedure are then compared to the RE’s ideal values (within $\epsilon$ for robustness).

\subsection{Self-testing graph states via inflation and deflation}\label{sec:maintheorem}

Given a target graph state $\vert G \rangle$, the REs use the inflated graph state $\vert G'\rangle$ to both prepare the target graph state, via deflation (whose measurements are included in the RE) and self-test it, since the resulting correlations contradict $d$-LHV$^\ast$ models and enforce local anticommutation. 
Hence, if a PE simulates the RE on $\vert G' \rangle$, then the deflated PE is equivalent to the deflated RE, which is locally equivalent to $\vert G \rangle$.
We use calligraphic font to denote measurements and observables associated with the PE, as opposed to italic font used for the REs.

\begin{thm}\label{theo:main}
For any $d$-inflated graph $G^{\prime}$ derived from a connected graph $G$, there exists a choice of reference experiment implying the following, even when one-way classical communication is permitted up to distance $d$ along the edges of the $d$-inflated graph:
Given a physical experiment $( \vert \Psi^{\prime} \rangle, \{ \mathcal{M}_k\})$ that $\epsilon$-simulates (Def.~\ref{def:simrob}) the reference experiment $( \vert G^{\prime}\rangle, \{ M_k \} )$, then, the deflated physical experiment $( \vert \Psi^{(x)} \rangle, \{ \mathcal{M}_k \} )$ is $\delta$-equivalent (Def.~\ref{def:equirob}) to the deflated reference experiment $( \vert G^{(x)}\rangle , \{ M_k \} )$ for all outcomes $(x)$. 
Deflation is implemented by fixed projective measurements on the chain vertices,
\begin{align}
\vert G^{(x)} \rangle &= \bigotimes_{v \in V^{\prime} \setminus V} \Pi^{(x)}_{v}(M_{k_{v}}) \vert G^{\prime} \rangle\,,\label{eq:deflationprojector1}\\
\vert \Psi^{(x)} \rangle &= \bigotimes_{v \in V^{\prime} \setminus V} \Pi^{(x)}_{v}(\mathcal{M}_{k_{v}}) \vert \Psi^{\prime} \rangle \,. \label{eq:deflationprojector2}
\end{align}
Here $\Pi^{(x)}(A)$ represents the product of local measurement projectors corresponding to the observables $A$ with outcome $(x)$.
The comparison is conditional on these fixed chain measurements dictated by the reference experiment and holds for every outcome pattern~$\mathbf{x}$.
\end{thm}
In Eqs.\,\eqref{eq:deflationprojector1} and \eqref{eq:deflationprojector2}, chain vertices are always measured in the deflation bases defined by the RE (mostly Pauli-$X$, in some cases Pauli-$Y$), which reduces $\vert G' \rangle$ to $\vert G \rangle$ up to local corrections as illustrated in Fig.~\ref{fig:graphmap}.

\paragraph{Theorem~\ref{theo:main} for RE~\ref{def:refe1},~\ref{def:refe2} or~\ref{def:refe3}.}
Consider a PE with state $\vert \Psi \rangle$, measurements $\mathcal{M}_k$, and submeasurements $\mathcal{C}_k$ compatible with a RE~\ref{def:refe1},~\ref{def:refe2} or~\ref{def:refe3}.
In this case, Theorem~\ref{theo:main} is a statement on the PE's post-deflation state $\vert \Psi^{\prime} \rangle$ after measuring the chain vertices in the $\mathcal{X}$ basis with outcome $\mathbf{x} = (x_{v})_{v\in V^{\prime} \setminus V}$,
\begin{equation}\label{postmeasurepower}
\vert \Psi^{\prime} \rangle = \Pi_{\mathcal{X}}^{(\mathbf{x})} \vert \Psi \rangle = \bigotimes_{v \in V^{\prime}\setminus V} \left( \mathds{1}_{v} + (-1)^{x_{v}} \mathcal{X}_{v} \right) \vert \Psi \rangle \,.
\end{equation}

To map $\vert \Psi^{\prime} \rangle$ to $\vert G\rangle$ by a local isometry, we use Prop.~\ref{prop:selftestgraphstates}.
We need \textit{(i)} anticommuting observables $\mathcal{X}_u,\mathcal{Z}_u$ for every power vertex $u$, and \textit{(ii)} stabilizer relations up to known signs $x_e$:
\begin{equation}\label{eq:postmeasuregraphstabcond}
\mathcal{X}_{u} \bigotimes_{v\in N(u)}\mathcal{Z}_{v}\ket{\Psi'} = (-1)^{(x_e)_{u}} \ket{\Psi'} \,.
\end{equation}
\paragraph{Local anticommutation relations.}
\begin{lem} \label{lem:ref1234} If a physical experiment $( \vert \Psi \rangle, \{ \mathcal{M}_k \}_k)$ is compatible with one of the reference experiments~\ref{def:refe1},~\ref{def:refe2},~or~\ref{def:refe3}, and simulates it according to Def.~\ref{def:simrob} for $\epsilon = 0$, then $\{ \mathcal{X}_{u}, \mathcal{Z}_{u} \} \vert \Psi \rangle = 0$ for all power vertices $u$ but vertex $v_{l}$ in RE~\ref{def:refe3} where $ \{ \mathcal{X}_{v_{l}} \mathcal{X}^{(X)}_{v_m}, Z_{v_{l}} \mathcal{X}^{(Z)}_{v_m} \} \vert \Psi \rangle =0$ for $v_m = d_{(v_{l},v_{r})}$.
\end{lem}
The complete proof of Lemma~\ref{lem:ref1234} is provided in Appendix~\ref{app:anti}.
Conceptually, the simulation condition imposes constraints from stabilizer formalism on the PE's observables.
Additionally, certain measurement settings correspond to measurements that lead to a nonlocal paradox for the inflated graph state.
The constraints imposed on the observables by these settings lead to an anticommutation relation for two observables on a single power vertex.
From a single anticommutation relation, Lemma~\ref{lem:indu} in Appendix~\ref{app:anti} extends the anticommutation relation to a neighboring power vertex.
This process, applied iteratively, yields an anticommutation relation for every power vertex of the connected graph, thus fulfilling \textit{(i)} in Prop.~\ref{prop:selftestgraphstates}.

\paragraph{Stabilizer conditions after deflation.}
Aiming for condition \textit{(ii)} in Prop.~\ref{prop:selftestgraphstates}, the simulation for all three reference experiments ensures $\langle \Psi \vert \zeta_{u} \vert \Psi \rangle = 1$ and thus
\begin{equation}
\zeta_{u} \vert \Psi \rangle = \vert \Psi \rangle\,, \label{eq:infstabsimcondition}
\end{equation}
for the PE's observables $\zeta_{u}$ corresponding to inflated stabilizers $f_u$ for all $u \in V$, and the PE's state $\vert \Psi \rangle$.
These conditions ensure the replication of stabilizing conditions for the non-inflated graph states.
After the measurement, the chain vertices' observables can be absorbed in the projectors \[ \Pi^{(x)}_{\mathcal{X}} \mathcal{X}_{v} = (-1)^{x_{v}} \Pi^{(x)}_{\mathcal{X}}\,, \] for $v\in V^{\prime} \setminus V$ and traded for a sign factor.
The collection of these factors turns out to be $(x_e)_{u}$ from Eq.\;\eqref{eq:postmeasuregraphstabcond}.

\paragraph{Isometry.}
We apply the standard local isometry that mimics a SWAP gat circuit from Fig.~\ref{fig:swapcirc} to the state $\vert \Psi'\rangle$ with one auxiliary qubit per power vertex.
The full evaluation of the isometry can be found in Appendix~\ref{app:isocom}.
Here, we present that the isometry applied to the postmeasurement state~\eqref{postmeasurepower} is
\begin{equation}
\vert G^{(\mathbf{x})} \rangle \otimes \vert \mathit{junk} \rangle = Z^{\mathbf{x}_e} \vert G \rangle \otimes \vert \mathit{junk} \rangle\,,
\end{equation}
with the residual state $\vert \mathit{junk} \rangle = \Pi_{\mathcal{X}}^{(\mathbf{x})} \bigotimes_{u \in V} (\mathds{1}_{u} + \mathcal{Z}_{u} )/\sqrt{2} \vert \Psi \rangle$ and $Z^{\mathbf{x}_e} = \bigotimes_{u \in V} Z^{(x_e)_{u}}_{u}$.
For any power vertex $u \in V$,
\begin{equation}
g_{u} \Pi_{X}^{(\mathbf{x})} \vert G^{\prime} \rangle = g_{u} \Pi_{X}^{(\mathbf{x})} f_{u} \vert G^{\prime} \rangle = (-1)^{(x_e)_{u}} \Pi_{X}^{(\mathbf{x})} \vert G^{\prime} \rangle \,,
\end{equation}
so that $\Pi_{X}^{(\mathbf{x})} \vert G^{\prime} \rangle = \rangle G^{(\mathbf{x})}\rangle$.
In RE~\ref{def:refe3} for a two-vertex graph, the two neighbors of $v_m$ measure Pauli-$Y$, and the $v_l$ phase depends also on $x_{v_m}$ when $d$ is odd; for $d=1$ replace $\mathcal{X}_{v_l}$ by $\mathcal{Y}_{v_l}$ in the isometry.

\paragraph{Robustness.}
In Appendix~\ref{app:robref3}, we extend the proof to robust equivalence by trailing the proof for the ideal case and bounding all expressions with the techniques present in~\cite{mckague2011self} and~\cite{SOS2015}.
The bounds on the isometry for all the REs discussed are in Tab.~\ref{tab:deltas}.
They scale with $\sqrt{\epsilon}$ where $\epsilon$, in Def.~\ref{def:selfrob} is the deviation of the measurement results from the ideal ones.
The bounds on the self-testing procedure with inflated graph states allowing bounded classical communication are slightly worse than for the non-inflated graph states, which is mainly due to the higher number of measurements required to derive the anticommutation relations.

\subsection{Self-testing families of graph states directly}
\label{sec:full}

Self-testing graph states directly with bounded, graph-dictated communication is generally impossible.
Especially if the extent of the graph is small compared to the communication range, it is impossible to find constraints on the observables since the number of observables can vastly outnumber all possible measurement outcomes.
We show, however, that it is possible to self-test certain graph states directly, that is, without additionally inflating or deflating them, even with bounded communication.

We introduce two reference experiments (REs): one for the $d$-inflated odd circular graph and one for the honeycomb cluster.
We use the following Pauli tensor products on a graph:
\begin{equation} \label{eq:measursxx}
P^{(\sigma)}_{u} = \sigma_{u} \bigotimes_{v \in N_d(u)} X_{v}\,,
\end{equation}
for a vertex $u$, its $d$-neighborhood, and the Pauli operator $\sigma_u\in\{X,Z\}$, all other operators are identity.
We also use alternating patterns
\begin{align}\label{eq:measurxzx}
\Tilde{P}^{(X)}_{u} &= X_{u} \bigotimes_{\substack{ v \in N_{d}(u) \\ \vert u - v \vert \bmod 2 =0 }} Z_{v} \bigotimes_{\substack{ v \in N_{d}(u) \\ \vert u - v \vert \bmod 2 =1 }} X_{v} \,,\\
\Tilde{P}^{(Z)}_{u} &= Z_{u} \bigotimes_{\substack{ v \in N_{d}(u) \\ \vert u - v \vert \bmod 2 =0 }} X_{v} \bigotimes_{\substack{ v \in N_{d}(u) \\ \vert u - v \vert \bmod 2 =1 }} Z_{v} \,, \label{eq:measurzxz}
\end{align}
with $\vert u-v \vert = dist(u,v)$.
Thus, the $d$-neighborhood $N_d(u)$ carries either all $X$ or an alternating $X/Z$ pattern.
\begin{refexp} \label{def:fullcircle}
Consider an odd circular graph $G = (V,E)$ and with $6d + 3$ vertices.
Considering any vertex as a power vertex, this graph can be considered as a $d$-inflated triangle graph - there are $2d+1$ unique ones.
The reference experiment consists of the graph state $\vert G \rangle$, and $2d+1$-times the measurements of RE~\ref{def:refe1} such that every vertex serve once as a power vertex in RE~\ref{def:refe1}.
Furthermore, the reference experiment contains the Pauli measurements
\begin{align} \label{eq:ref4altx}
\Tilde{M}^{(X)}_{u} &= \Tilde{P}^{(X)}_{u} \prod_{v \in N(u)} P^{(Z)}_{v} \,,\\ \label{eq:ref4altz}
\Tilde{M}^{(Z)}_{u} &= P^{(X)}_{u} \prod_{v \in N(u)} \Tilde{P}^{(Z)}_{v} \,, \\
\label{eq:ref4alt}
M^{(alt)}_{u} &= \Tilde{P}^{(X)}_{u} \prod_{v \in N(u)} \Tilde{P}^{(Z)}_{v} \,,
\end{align} for all $u \in V$.
Overall, the experiment consists of $4\vert V \vert + 1 + 2(2d+1)$
\end{refexp}
As submeasurements of $\Tilde{M}^{(X)}_{u}$ and $\Tilde{M}^{(Z)}_{u}$, we target the inflated generators $f_u$; the submeasurements of $M^{(alt)}_{u}$ are the inflated generators $g^{\prime}_u$ .

\begin{refexp}\label{def:refhex}
Consider a honeycomb lattice $H = (V,E)$, which is a combination of two hexagonal lattices, and we divide the set of all vertices accordingly into two distinct subsets $V_{hex_{1}},V_{hex_{2}}$.

The RE consists of the honeycomb cluster state $\vert H \rangle$ and the measurements 
\begin{align}\label{eq:measurehex1}
M^{(alt)}_{1} &= \bigotimes_{u\in V_{hex_{1}}} X_{u} \bigotimes_{v \in V_{hex_{2}}} Z_{u} \,, \\ \label{eq:measurehex2}
M^{(alt)}_{2} &= \bigotimes_{u\in V_{hex_{1}}} Z_{u} \bigotimes_{v \in V_{hex_{2}}} X_{u} \,.
\end{align}
Furthermore, it contains the measurements of RE~\ref{def:refe2} ($\vert N(v_{c}) \vert = 3$) for the $d$-inflated star graph on a \emph{tripoint star} depicted in Fig.~\ref{fig:honey} around for an arbitrary vertex $v_{c}$.
To embed the star graph and measurements from RE~\ref{def:refe2} into the honeycomb cluster, vertices adjacent to the star measure in the Pauli $Z$ basis whenever a neighboring vertex on the star measures in the Pauli $X$ basis.
Additionally, the reference experiment contains the Pauli measurements $\Tilde{M}^{(X)}_{v_{0}}$ in Eq.\;\eqref{eq:ref4altx} and $\Tilde{M}^{(Z)}_{v_{0}}$ in Eq.\;\eqref{eq:ref4altz} for an arbitrary vertex $v_{0} \in N^{\prime}(v_{c})$ that is at distance $2d+1$ from $v_{c}$ on the tripoint star.
\end{refexp}

As submeasurements, we target the generator elements $g_{v}$ and the inflated generator elements $f_{v}$ on the tripoint star for $v \in V_{hex_{1}}$ and $v \in V_{hex_{2}}$ respectively; and $f_{v_{0}},f_{v_{c}}$ as a submeasurement Eq.\;\eqref{eq:ref4altx} and Eq.\;\eqref{eq:ref4altz}.

\begin{thm}
The set of measurement correlations from the reference experiments~\ref{def:fullcircle} and~\ref{def:refhex}, robustly self-test them respectively in terms of Def.~\ref{def:selfrob}, even if one-way classical communication is permitted up to distance $d$ on the edges of the graph $G$ given in the reference experiment.
\label{theo:entire}
\end{thm}

The detailed proof of Theorem~\ref{theo:entire} is in Appendix~\ref{app:selfsymmfull}; here, we offer an intuitive overview.

Consider RE~\ref{def:fullcircle}, if a physical experiment (PE) simulates $2d+1$ times the RE~\ref{def:refe1}, with each vertex alternatively serving as a power vertex or a chain vertex, Lemma~\ref{lem:ref1234} implies the existence of measurement observables $\mathcal{X}_{u}$ and $\mathcal{Z}_{u}$ for every vertex $u$ in the inflated graph such that $\{ \mathcal{X}_{u}, \mathcal{Z}_{u} \} \vert \Psi \rangle = 0$.
This establishes the anticommutation condition \textit{(i)} of Prop.~\ref{prop:selftestgraphstates}.
To complete self-testing we must also certify the stabilizer elements $X_u \bigotimes_{v\in N(u)} Z_v$, which are not directly enforced by repeating RE~\ref{def:refe1}.
We therefore use the submeasurements in \eqref{eq:ref4alt}

Communication changes which local observables are implemented under these inputs: the RE~\ref{def:refe1} settings that established anticommutation place Pauli-$X$ on neighbors of the power vertex, whereas the stabilizer elements $X_u \bigotimes_{v\in N(u)} Z_v$ require those neighbors to see Pauli-$Z$.
Let $\X_u$ (resp.\ $\Z_u$) denote the observable measured at $u$ when inputs match \eqref{eq:measurxzx} (resp.\ \eqref{eq:measurzxz}).
The vertices within communication distance measure in an alternating pattern according to Pauli operator $X$ or $Z$.
Recall that, in contrast, the local observables $\mathcal{X}_{u}$ and $\mathcal{Z}_{u}$ stem from measurement inputs that correspond to the Pauli tensor product~\eqref{eq:measursxx}.
In contrast, $\mathcal{X}_u,\mathcal{Z}_u$ arise from \eqref{eq:measursxx}, where all vertices within communication distance of $u$ measure according to Pauli operator $X$.

Using Eq.\;\eqref{eq:measrefstab12}, replace identities in the $d$-neighborhood by $Z$ to transform those submeasurements into the alternating patterns \eqref{eq:measurxzx}–\eqref{eq:measurzxz}.
This propagates (see Lemma~\ref{lem:indu}) the anticommutation relations from $\{\mathcal{X}_u,\mathcal{Z}_u\}$ to
\begin{equation}\label{eq:brandnewAC}
\{\X_{u},\Z_{u}\}\vert {\Psi} \rangle = 0\,,
\end{equation}
for $u \in V^{\prime}$.

Consequently, we acquire two pairs of anticommuting observables per vertex, denoted as $(\mathcal{X}_{u},\mathcal{Z}_{u})$ and $(\X_{u},\Z_{u})$.
Notably, unlike the former pair, the latter proves instrumental in verifying the adherence to the stabilizer relations of the inflated graph state.
For that, the reference experiment has to involve submeasurement consisting of $X_{u}\bigotimes_{v\in N(u)}Z_{v}$ and identities everywhere else.
The corresponding measurement would have $X$ and $Z$ in such a way that they mimic the tensor product~\eqref{eq:ref4alt}.
In such a case, from the inputs sent to all vertices in $d$-neighborhood of $u$ and two neighboring vertices, all of them are ensured to measure $\X$ and $\Z$.
The simulation of the stabilizing condition equips us with the relation
\begin{equation}\label{eq:brandnewSTAB}
\X_{u}\bigotimes_{v\in N(u)} \Z_{v}\vert {\Psi} \rangle = \vert{\Psi} \rangle\,,
\end{equation}
for $u \in V^{\prime}$.
Finally, the Eq.\;\eqref{eq:brandnewAC} and Eq.\;\eqref{eq:brandnewSTAB} ensure the existence of a local isometry mapping the physical state to the inflated graph state. 

For RE~\ref{def:refhex}, the argument is analogous.
Reference experiment~\ref{def:refe2} on a tripoint star supplies a initial anticommutation near $v_c$.
Then, using $f_{v_0}$ from \eqref{eq:ref4altx} and $f_{v_c}$ from \eqref{eq:ref4altz}, propagates the initial anticommutation relation (all vertices in communication range receive input corresponding to Pauli-$X$) to $\{\X_{v_c},\Z_{v_c}\}$.
Then, they propagate across the lattice via $f_u$ submeasurements from \eqref{eq:measurehex1}–\eqref{eq:measurehex2}.
The stabilizer condition \textit{(ii)} in Prop.~\ref{prop:selftestgraphstates} follow from the $g_u$ submeasurements of \eqref{eq:measurehex1} -- \eqref{eq:measurehex2}, giving \eqref{eq:brandnewAC}–\eqref{eq:brandnewSTAB} on all sites and hence a local isometry to the honeycomb cluster state.

Note that RE~\ref{def:refhex} does not depend on the communication distance $d$, which applies if the lattice is infinite.
As soon as it is finite, the stabilizer elements change at the boundaries and break the symmetry, which gives the $d$-LHV$^\ast$ model an advantage.
It is however possible to consider a honeycomb lattice of finite size with periodic boundary conditions, i.e., mapped onto a torus.
The size of the torus restricts the communication distance $d$ for which Theorem~\ref{theo:entire} holds.
Lastly, we want to emphasize that the result from Theorem~\ref{theo:main} and Theorem~\ref{theo:entire} are, in general, compatible that is, for a given communication distance $d$, they can be combined to self-test certain parts of the graph entirely while using the deflation method elsewhere.

\begin{figure}
\centering
\begin{tikzpicture}[scale=0.9]
\tikzstyle{power}=[draw=black, very thick, shape=circle]
\tikzstyle{chain}=[draw=black, shape=circle, minimum size=5pt]
\tikzstyle{empty}=[]
\tikzstyle{edge}=[draw=black, line width=1pt]
\tikzstyle{thickedge1}=[dashed, draw=myRed, line width=2pt]
\tikzstyle{thickedge}=[dashed, draw=myCyan, line width=2pt]
\tikzstyle{triskelion}=[dashed, draw=myRed, line width =1.5pt]

\node [style=empty] (141) at (-6.062,6.5) {} ;
\node [style=empty] (150) at (-6.062,7.5) {} ;
\node [style=power] (151) at (-6.928,8.0) {} ;
\node [style=empty] (160) at (-6.928,8.75) {} ;
\node [style=empty] (221) at (-6.062,3.5) {} ;
\node [style=empty] (230) at (-6.062,4.5) {} ;
\node [style=power] (231) at (-6.928,5.0) {} ;
\node [style=power] (240) at (-6.928,6.0) {} ;
\node [style=power] (241) at (-7.794,6.5) {} ;
\node [style=power] (250) at (-7.794,7.5) {} ;
\node [style=power] (251) at (-8.66,8.0) {} ;
\node [style=empty] (260) at (-8.66,8.75) {} ;
\node [style=empty] (311) at (-6.928,2.25) {} ;
\node [style=power] (320) at (-6.928,3.0) {} ;
\node [style=power] (321) at (-7.794,3.5) {} ;
\node [style=power] (330) at (-7.794,4.5) {} ;
\node [style=power] (331) at (-8.66,5.0) {} ;
\node [style=power] (340) at (-8.66,6.0) {$v$} ;
\node [style=power] (341) at (-9.526,6.5) {} ;
\node [style=power] (350) at (-9.526,7.5) {} ;
\node [style=power] (351) at (-10.392,8.0) {} ;
\node [style=empty] (360) at (-10.392,8.75) {} ;
\node [style=empty] (411) at (-8.66,2.25) {} ;
\node [style=power] (420) at (-8.66,3.0) {} ;
\node [style=power] (421) at (-9.526,3.5) {} ;
\node [style=power] (430) at (-9.526,4.5) {} ;
\node [style=power] (431) at (-10.392,5.0) {} ;
\node [style=power] (440) at (-10.392,6.0) {} ;
\node [style=power] (441) at (-11.258,6.5) {} ;
\node [style=power] (450) at (-11.258,7.5) {} ;
\node [style=power] (451) at (-12.124,8.0) {} ;
\node [style=empty] (460) at (-12.1239,8.75) {} ;
\node [style=empty] (511) at (-10.392,2.25) {} ;
\node [style=power] (520) at (-10.392,3.0) {} ;
\node [style=power] (521) at (-11.258,3.5) {} ;
\node [style=power] (530) at (-11.258,4.5) {} ;
\node [style=power] (531) at (-12.124,5.0) {$u$} ;
\node [style=power] (540) at (-12.124,6.0) {} ;
\node [style=power] (541) at (-12.99,6.5) {} ;
\node [style=power] (550) at (-12.99,7.5) {} ;
\node [style=power] (551) at (-13.856,8.0) {} ;
\node [style=empty] (560) at (-13.856,8.75) {} ;
\node [style=empty] (611) at (-12.124,2.25) {} ;
\node [style=power] (620) at (-12.1239,3.0) {} ;
\node [style=power] (621) at (-12.99,3.5) {} ;
\node [style=power] (630) at (-12.9898,4.5) {} ;
\node [style=power] (631) at (-13.856,5.0) {} ;
\node [style=power] (640) at (-13.856,6.0) {} ;
\node [style=empty] (641) at (-14.722,6.5) {} ;
\node [style=empty] (650) at (-14.722,7.5) {} ;
\node [style=empty] (711) at (-13.856,2.25) {} ;
\node [style=power] (720) at (-13.856,3.0) {} ;
\node [style=empty] (721) at (-14.722,3.5) {} ;
\node [style=empty] (730) at (-14.722,4.5) {} ;

\draw [style=thickedge1] (141) to (240);
\draw [style=edge] (150) to (151);
\draw [style=edge] (151) to (250);
\draw [style=edge] (151) to (160);
\draw [style=edge] (221) to (320);
\draw [style=thickedge] (230) to (231);
\draw [style=thickedge] (231) to (330);
\draw [style=edge] (231) to (240);
\draw [style=thickedge1] (240) to (241);
\draw [style=thickedge1] (241) to (340);
\draw [style=edge] (241) to (250);
\draw [style=edge] (250) to (251);
\draw [style=edge] (251) to (350);
\draw [style=edge] (251) to (260);
\draw [style=edge] (311) to (320);
\draw [style=edge] (320) to (321);
\draw [style=edge] (321) to (420);
\draw [style=edge] (321) to (330);
\draw [style=thickedge] (330) to (331);

\draw[myCyan, line width=2pt,dash pattern= on 3pt off 5pt] (331) to (430);
\draw[myRed, line width=2pt,dash pattern= on 3pt off 5pt,dash phase=4pt] (331) to (430);

\draw [style=thickedge1] (331) to (340);
\draw [style=thickedge1] (340) to (341);
\draw [style=edge] (341) to (440);
\draw [style=thickedge1] (341) to (350);
\draw [style=thickedge1] (350) to (351);
\draw [style=edge] (351) to (450);
\draw [style=thickedge1] (351) to (360);
\draw [style=edge] (411) to (420);
\draw [style=edge] (420) to (421);
\draw [style=thickedge1] (421) to (520);
\draw [style=thickedge1] (421) to (430);
\draw [style=thickedge] (430) to (431);
\draw [style=thickedge] (431) to (530);
\draw [style=edge] (431) to (440);
\draw [style=edge] (440) to (441);
\draw [style=edge] (441) to (540);
\draw [style=edge] (441) to (450);
\draw [style=edge] (450) to (451);
\draw [style=edge] (451) to (550);
\draw [style=edge] (451) to (460);
\draw [style=thickedge1] (511) to (520);
\draw [style=edge] (520) to (521);
\draw [style=edge] (521) to (620);
\draw [style=edge] (521) to (530);
\draw [style=thickedge] (530) to (531);
\draw [style=thickedge] (531) to (630);
\draw [style=thickedge] (531) to (540);
\draw [style=thickedge] (540) to (541);
\draw [style=edge] (541) to (640);
\draw [style=thickedge] (541) to (550);
\draw [style=thickedge] (550) to (551);
\draw [style=edge] (551) to (650);
\draw [style=thickedge] (551) to (560);
\draw [style=edge] (611) to (620);
\draw [style=edge] (620) to (621);
\draw [style=thickedge] (621) to (720);
\draw [style=thickedge] (621) to (630);
\draw [style=edge] (630) to (631);
\draw [style=edge] (631) to (730);
\draw [style=edge] (631) to (640);
\draw [style=edge] (640) to (641);
\draw [style=thickedge] (711) to (720);
\draw [style=edge] (720) to (721);

\end{tikzpicture}
\caption{We illustrate an excerpt of the honeycomb cluster. The dashed lines sketch the pattern that defines the $d$-inflated star graph, which we call \emph{tripoint star}. The blue tripoint star is defined around vertex $u \in V_{hex_{1}}$ and the red tripoint star around vertex $v \in V_{hex_{2}}$, where $V_{hex_{1}}$ and $V_{hex_{2}}$ are the two hexagonal lattices that define the honeycomb lattice.}
\label{fig:honey}
\end{figure}

\section{Conclusion}
\label{sec:conc}

In this work we show that it is possible to self-test graph states, even when allowing for communication amongst parties.
Surprisingly, we show that for some graph states this can be done directly, that is, in the standard setting of self-testing, but allowing communication between the devices. 
For our communication scenario this is not possible for all graph states, but we are able to self-test them in a round about way, by preparing larger graph states and measuring them, such that they prepare the target graph state at the same time as self-testing it.

In this work, the allowed communication is defined by the same graph as the state itself, this follows the original works where communication is considered for showing non-classicality \cite{barrett2007modeling,Meyer2023}, which is motivated by the fact that it is natural that neighbors in the graph communicate in order to establish the graph state itself. One could, of course, consider different communication graphs. And indeed states which are not graph states. This would be an interesting open set of problems.

Finally, we note that another motivation for our communication scenario is that it has already found application in proving quantum advantage in computation~\cite{bravyi2018quantum,bravyi2020quantum}, communication~\cite{gall2018quantum} and randomness expansion \cite{coudron2018trading}. Given these, and the ever-increasing use of self testing to device independence in quantum technologies \cite{vsupic2018self}, our results offer perspectives for pushing the applications of self testing to new situations and advantage.

\begin{acknowledgments}
We acknowledge funding from the ANR through the ANR-17-CE24-0035 VanQuTe project and the PEPR integrated project EPiQ ANR-22-PETQ-0007 as well as the HQI initiative (www.hqi.fr) award number ANR-22-PNCQ-0002 as part of Plan France 2030.
\end{acknowledgments}

\bibliographystyle{quantum}
\bibliography{main}

\onecolumngrid
\pagebreak
\appendix

\section{Novel reference experiment for graph states}
\label{app:line}

We present a new reference experiment (RE) using Pauli measurements on graph states whose corresponding graph contains an induced subgraph with three vertices in a line.
Here, we give a proof for an ideal self-test, $\epsilon,\delta\rightarrow 0$ in Def.~\ref{def:selfrob}, and show its robustness, $\delta \sim \mathcal{O}(\sqrt{\epsilon})$ in Appendix~\ref{app:robref3}.
For a set $U$ of vertices, we denote $N(U) = \bigcap_{u \in U} N(u)$ as the intersection of all nearest neighborhoods of the vertices in the subset.

\begin{refexp}[Line of three vertices]
Let $G=(V,E)$ be a graph with an induced subgraph of three vertices $\{ v_{l},v_{c},v_{r}\}$ in a line $(v_{l},v_{c}),(v_{r},v_{c}) \in E$ and $(v_{l},v_{r})\notin E$.
The reference experiment consists of the corresponding graph state $\vert G \rangle$, the Pauli measurements corresponding to the generator elements $g_{u} = X_{u} Z^{N(u)}$ for every vertex $u \in V$, $M_{v} = Y_{v_{c}}\, Y_{v} Z^{N(\{ v_{c},v\})}$ for all $v \in N(v_{c})$, and $M_{v_{c}} = Y_{v_{l}} X_{v_{c}} Y_{v_{r}} Z^{N(\{ v_{l}, v_{c}, v_{r}\})}$.
\label{def:refe0}
\end{refexp}
\begin{thm}
If a physical experiment is \emph{compatible} and $\epsilon$-\emph{simulates} in terms of Def.~\ref{def:simrob} the reference experiment~\ref{def:refe0}, they are $\delta$-\emph{equivalent} in terms of Def.~\ref{def:selfrob}.
\label{theo:ref0}
\end{thm}
\begin{proof}[Proof for $\epsilon,\delta = 0$]
The first step is to exploit all the implications from the conditions that a PE is \emph{compatible} and $\epsilon$-\emph{simulates} the reference experiment~\ref{def:refe0}.

We determine the expectation values of the RE's measurements by writing them in terms of the graph state's generator elements.
Given that $M_{v_{c}} =- g_{v_{l}} g_{v_{c}} g_{v_{r}}$ and $M_{v} = g_{v_{c}} g_{v}$, the expectation values are
\(
(-1) \, \mathbb{E}_{\vert \psi \rangle} (M_{v_{c}}) = \mathbb{E}_{\vert \psi \rangle} (M_{v}) = \mathbb{E}_{\vert \psi \rangle} (g_{u} ) = 1 \,,
\)
for all $u \in V$ and $v \in N(v_{c})$.
Since the expectation values are equal to $\pm 1$, the measurements have deterministic outcomes, and we can deduce that $(-1)M_{v_{c}} \vert \psi \rangle = M_{v} \vert \psi \rangle = g_{u} \vert \psi \rangle = \vert \psi \rangle$.

Consider a PE that simulates the RE, with the state $\vert \psi \rangle$ and measurements $\xi_{u}$ for $u \in V$, $\mathcal{M}_{v}$ for $v \in N(v_{c})$, and $\mathcal{M}_{v_{c}}$, where $\xi_{u}$ represents the PE's measurement correlations corresponding to the generator element $g_{u}$ in the RE.
We also use calligraphic font for the observables of the PE in contrast to italic font for the RE.
Then, $\xi_{u} = \mathcal{X}_{u} \mathcal{Z}^{N(u)}$, $\mathcal{M}_{v} = \mathcal{Y}_{v_{c}}\, \mathcal{Y}_{v} \mathcal{Z}^{N(\{ v_{c},v\})}$, and $\mathcal{M}_{v_{c}} = \mathcal{Y}_{v_{l}} \mathcal{X}_{v_{c}} \mathcal{Y}_{v_{r}} \mathcal{Z}^{N(\{ v_{l}, v_{c}, v_{r}\})}$.
Recall that we do not assume anything about the observables of the PE but that they are defined over the same Hilbert space as the ones from the RE.
Since the PE simulates the RE, the measurement outcome is also deterministically equal to $\pm1$, and $O^2= \mathds{1}$ for all occurring local operators $O$.
Furthermore,
\( (-1) \, \mathcal{M}_{v_{c}} \vert \psi \rangle = \mathcal{M}_{u} \vert \psi \rangle = \xi_w \vert \psi \rangle = \vert \psi \rangle \,.\)
so that $\mathcal{X}_{u} \mathcal{Z}^{N(u)} \vert \psi \rangle = \vert \psi \rangle$ and $\mathcal{Y}_{v_{c}}\, \mathcal{Y}_{v} \mathcal{Z}^{N(\{ v_{c},u\})} \vert \psi \rangle = \vert \psi \rangle$, and in particular,
\begin{equation} \begin{array}{lllllll}
& \mathcal{Y}_{v_{l}} & \mathcal{X}_{v_{c}} & \mathcal{Y}_{v_{r}} & \mathcal{Z}^{N(\{ v_{l},v_{c}, v_{r}\})} &\vert \psi \rangle= & -\vert \psi \rangle \,,\\
& \mathcal{Y}_{v_{l}} & \mathcal{Y}_{v_{c}} & \mathcal{Z}_{v_{r}} & \mathcal{Z}^{N(\{ v_{l},v_{c}\})\setminus \{ v_{r} \}} & \vert \psi \rangle= & \phantom{-} \vert \psi \rangle \,,\\
& \mathcal{Z}_{v_{l}} & \mathcal{X}_{v_{c}} & \mathcal{Z}_{v_{r}} & \mathcal{Z}^{N( v_{c}) \setminus \{ v_{l},v_{r} \}} &\vert \psi \rangle= & \phantom{-} \vert \psi \rangle \,,\\
& \mathcal{Z}_{v_{l}} & \mathcal{Y}_{v_{c}} & \mathcal{Y}_{v_{r}} & \mathcal{Z}^{N(\{ v_{c}, v_{r}\}) \setminus \{ v_{l} \}} &\vert \psi \rangle= &\phantom{-} \vert \psi \rangle \,.
\end{array} \nonumber \end{equation}
By applying these relations in sequence, we now prove anticommutation relations in terms of their action on the state $\vert \psi \rangle$ for two observables of each vertex.
For vertex $v_{c}$, \( (-1)\vert \psi \rangle = \mathcal{M}_{v_{c}} \mathcal{M}_{v_{l}} \xi_{v_{c}} \mathcal{M}_{v_{r}} \vert \psi \rangle = \mathcal{X}_{v_{c}} \mathcal{Y}_{v_{c}} \mathcal{X}_{v_{c}} \mathcal{Y}_{v_{c}} \vert \psi \rangle \) implies $\{\mathcal{X}_{v_{c}}, \mathcal{Y}_{v_{c}}\} \vert \psi \rangle =0$.
The calculus can be visualized with the above array of relations by multiplying the observables from top to bottom for every vertex.

With $\mathcal{M}_{v} \xi_{v_{c}} \mathcal{M}_{v} \xi_{v_{c}} \vert \psi \rangle = \vert \psi \rangle$ this anticommutation relation propagates to one for every $v \in N(v_{c})$.
As above, it is easy to evaluate by writing the observables row-wise and multiplying them column-wise, namely

\begin{equation} \begin{array}{lllllll}
& \mathcal{Z}_{v} & \mathcal{X}_{v_{c}} & \mathcal{Z}^{N(v_{c}) \setminus \{ v \}} &\vert \psi \rangle= & \vert \psi \rangle \,,\\
& \mathcal{Y}_{v} & \mathcal{Y}_{v_{c}} & \mathcal{Z}^{N(\{v,v_{c}\})} & \vert \psi \rangle= & \vert \psi \rangle \,,\\
& \mathcal{Z}_{v} & \mathcal{X}_{v_{c}} & \mathcal{Z}^{N(v_{c}) \setminus \{ v \}} &\vert \psi \rangle= & \vert \psi \rangle \,,\\
& \mathcal{Y}_{v} & \mathcal{Y}_{v_{c}} & \mathcal{Z}^{N(\{v,v_{c}\})} & \vert \psi \rangle= & \vert \psi \rangle \,.
\end{array} \nonumber \end{equation}
It follows from $\{\mathcal{X}_{v_{c}}, \mathcal{Y}_{v_{c}}\} \vert \psi \rangle =0$ that $\{\mathcal{Y}_{v}, \mathcal{Z}_{v}\} \vert \psi \rangle =0$.
We write \( \vert \psi \rangle = \xi_w \mathcal{M}_{v} \xi_w \mathcal{M}_{v} \vert \psi \rangle \) in an array,
\begin{equation} \begin{array}{lllllll}
& \mathcal{X}_{w} & \mathcal{Z}_{v} & \mathcal{Z}^{N(w) \setminus \{ v \}} &\vert \psi \rangle= & \vert \psi \rangle \,,\\
& \mathcal{Z}_{w} & \mathcal{Y}_{v} & \mathcal{Y}_{v_{c}} \mathcal{Z}^{N(\{v,v_{c}\})\setminus \{ w \}} & \vert \psi \rangle= & \vert \psi \rangle \,,\\
& \mathcal{X}_{w} & \mathcal{Z}_{v} & \mathcal{Z}^{N(w) \setminus \{ v \}} &\vert \psi \rangle= & \vert \psi \rangle \,,\\
& \mathcal{Z}_{w} & \mathcal{Y}_{v} & \mathcal{Y}_{v_{c}} \mathcal{Z}^{N(\{v,v_{c}\})\setminus \{ w \}} & \vert \psi \rangle= & \vert \psi \rangle \,.
\end{array} \nonumber \end{equation}
to visualize that it implies $\{ \mathcal{X}_w,\mathcal{Z}_w\} \vert \psi \rangle = 0$ for every vertex $w \in N(v)\setminus \{v_{c}\}$ from the previous anticommutation relation.

For all other vertices, one sequentially applies Lemma 2 from~\cite{mckague2011self} that propagates the anticommutation relation from a vertex $w \in N(u)$ with $\{ \mathcal{X}_w,\mathcal{Z}_w \} \vert \psi \rangle = 0$ to any other vertex $u$ such that $\{ \mathcal{X}_{u},\mathcal{Z}_{u} \} \vert \psi \rangle = 0$ using $\xi_w \xi_{u} \xi_w \xi_{u} \vert \psi \rangle = \vert \psi \rangle$.
Again, writing the measurements in an array can help visually verify this.

Note that one can obtain an anticommutation relation $\{\mathcal{X},\mathcal{Z}\} \vert \psi \rangle = 0$ for all vertices in particular for $v \in N(v_{c})$ and $v_{c}$ in exchange for measuring on these vertices in a third Pauli basis if the graph consists of more than three vertices.
Specifically, if a vertex $u \in N(N(v_{c}))$ exists, then we have shown that $\{\mathcal{X}_{u},\mathcal{Z}_{u}\} \vert \psi \rangle = 0$, and using the Lemma 2 in~\cite{mckague2011self}, we can back-propagate the anticommutation relation to these vertices.
In this case, the following isometry is slightly easier to evaluate because it is equivalent to the one used in~\cite{mckague2011self}.

We demonstrate the equivalence between the PE and the RE using all the implications arising from the PE simulating the RE.

The local isometry for Def.~\ref{def:equirob} is $\Phi = \prod_{u\in V} \phi_{u}$, with $\phi_{u}$ defined by one of the circuits in Fig.~\ref{fig:swapcirc} and Fig.~\ref{fig:swapcirc2}.
\begin{figure}
\centering
\begin{quantikz}
\lstick{$\ket{ 0}$}& \gate{H} & \ctrl{1}&\qw & \qw & \gate{H} & \ctrl{1} & \gate{S} &\qw \\
\lstick{$\ket{\psi}$}& \qw & \gate{Z} & \qw & \qw & \qw & \gate{Y} &\qw & \qw\\ 
\lstick{$\ket{ 0}$}& \gate{H} & \ctrl{1} & \ctrl{1}& \gate{S} & \gate{H} & \ctrl{1} & \qw & \qw \\
\lstick{$\ket{\psi}$}& \qw & \gate{Y} & \gate{X} & \qw & \qw & \gate{X} &\qw & \qw \end{quantikz}
\caption{As an alternative to the quantum circuit in Fig.~\ref{fig:swapcirc}, each quantum circuit implements a SWAP gate between an input state $\vert \psi \rangle$ and the state $\vert 0 \rangle$ with the Hadamard gate $H$ and Pauli operators $X,Y,Z$.}
\label{fig:swapcirc2}
\end{figure}
The choice of circuit (i.e., operators) depends on the observables in the reference experiment.
Including the respective ancilla qubits, we apply the third circuit to the vertex $v_{c}$, the second circuit to all neighbor vertices $u \in N(v_{c})$ of the vertex $v_{c}$, and the first circuit to all other vertices.

For brevity, we define the corresponding operators $\mathcal{O}_{u} = \mathcal{X}_{u}$ for $u \in V\setminus N(v_{c})$, $\mathcal{O}_{u} = \mathrm{i} \mathcal{Y}_{u}$ for $u \in N(v_{c})$, and $\Bar{\mathcal{Z}}_{u} := (\mathds{1}_{u} + \mathcal{Z}_{u})/2$ for $u \in V \setminus \{ v_{c}\}$ and $\Bar{\mathcal{Z}}_{v_{c}}: = (\mathds{1}_{v_{c}} + \mathrm{i} \mathcal{Y}_{v_{c}} \mathcal{X}_{v_{c}})/2$.
Using the anticommutation relations, we obtain $\mathcal{O}^{a_{u}}_{u} (\Bar{\mathcal{Z}}_{a_{u}})_{u} \vert \psi \rangle = \Bar{\mathcal{Z}}_{u} \mathcal{O}^{a_{u}}_{u} \vert \psi \rangle$.
As a result, the local isometries' action is
\begin{equation}
\Phi (\vert \mathbf{0} \rangle \otimes \vert \psi \rangle ) = \bigotimes_{u\in V} \sum_{a_{u} = 0,1 } \vert a_{u} \rangle \otimes \mathcal{O}^{a_{u}}_{u} \, (\Bar{\mathcal{Z}}_{a_{u}})_{u} \, \vert \psi \rangle = \sum_{\mathbf{a} \in \mathbb{F}^{\vert V \vert}_{2} } \vert \mathbf{a} \rangle \bigotimes_{u\in V} \mathcal{O}^{a_{u}}_{u} \, (\Bar{\mathcal{Z}}_{a_{u}})_{u} \, \vert \psi \rangle = \sum_{\mathbf{a} \in \mathbb{F}^{\vert V \vert}_{2} } \vert \mathbf{a} \rangle \bigotimes_{u\in V} \Bar{\mathcal{Z}}_{u} \, \mathcal{O}^{a_{u}}_{u} \, \vert \psi \rangle \,. \label{eq:isocirceval}
\end{equation}

For further evaluation, we denote $V^{\prime} = V \setminus ( N (v_{c}) \cup \{ v_{c} \})$.
Additionally, consider some arbitrary but fixed order of the vertices $N(v_{c})$ and $V^{\prime}$ which justify the notation $u < v$ for two ordered vertices $u,v$.
Then,
\begin{align}
\Phi (\vert \mathbf{0} \rangle \otimes \vert \psi \rangle ) &= \sum_{\mathbf{a} \in \mathbb{F}^{\vert V \vert}_{2} } \vert \mathbf{a} \rangle \bigotimes_{u\in V} \Bar{\mathcal{Z}}_{u} \, \mathcal{O}^{a_{u}}_{u} \, \vert \psi \rangle \label{ali:equiref0} \\
&= \frac{1}{2^{\vert V \vert}}\sum_{\mathbf{a} \in \mathbb{F}^{\vert V \vert}_{2} } \vert \mathbf{a} \rangle \otimes \Big[ ( \mathds{1}_{v_{c}} + \mathrm{i} \mathcal{Y}_{v_{c}} \mathcal{X}_{v_{c}} ) \, \mathcal{X}^{a_{v_{c}}}_{v_{c}} \bigotimes_{u\in N(v_{c})} (\mathds{1}_{u} + \mathcal{Z}_{u}) \, (\mathrm{i} \mathcal{Y}_{u})^{a_{u}} \bigotimes_{v\in V^{\prime}} (\mathds{1}_{v} + \mathcal{Z}_{v}) \, \mathcal{X}^{a_{v}}_{v} \, \vert \psi \rangle \Big] \nonumber \\
&= \sum_{\mathbf{a} \in \mathbb{F}^{\vert V \vert}_{2} } \vert \mathbf{a} \rangle \bigotimes_{\substack{u\in N(v_{c})\\ v \in V^{\prime}}} (-1)^{a_{u} a_{v} \delta(v \in N(u))} \hspace{-0.25cm} \prod_{u < w \in N(v_{c})} \hspace{-0.25cm}(-1)^{ a_{u} a_w} \Big[ \Bar{\mathcal{Z}}_{v_{c}} \Bar{\mathcal{Z}}_{u} \Bar{\mathcal{Z}}_{v} \, \mathcal{X}^{a_{v_{c}}}_{v_{c}} \, (\mathrm{i} \mathcal{Y}_{v_{c}})^{a_{u}} \mathcal{X}^{a_{v}}_{v} \, \vert \psi \rangle \Big] \label{ali:appref0-2}\\
&= \sum_{\mathbf{a} \in \mathbb{F}^{\vert V \vert}_{2} } \vert \mathbf{a} \rangle \bigotimes_{\substack{u\in N(v_{c})\\ v \in V^{\prime}}} (-1)^{a_{u} a_{v_{c}} + a_{u} a_{v} \delta(v \in N(u))} \hspace{-0.25cm} \prod_{u < w \in N(v_{c})} \hspace{-0.25cm} (-1)^{ a_{u} a_w} \Big[ \Bar{\mathcal{Z}}_{v_{c}} \Bar{\mathcal{Z}}_{u} \Bar{\mathcal{Z}}_{v} \, \mathcal{X}^{a_{v_{c}}}_{v_{c}} \mathcal{X}_{v_{c}}^{a_{u}} \mathcal{X}^{a_{v}}_{v} \vert \psi \rangle \Big] \label{ali:appref0-3} \\
&= \sum_{\mathbf{a} \in \mathbb{F}^{\vert V \vert}_{2} } \vert \mathbf{a} \rangle \bigotimes_{\substack{u\in N(v_{c})\\ v \in V^{\prime}}} (-1)^{a_{u} a_{v_{c}}} \hspace{-0.25cm}\prod_{\substack{u < w \in N(v_{c})\\v < w}} (-1)^{a_{u} a_w + a_{v} a_w \delta(w \in N(v))} \Big[ \Bar{\mathcal{Z}}_{v_{c}} \Bar{\mathcal{Z}}_{u} \Bar{\mathcal{Z}}_{v} \vert \psi \rangle \Big] \label{ali:appref0-4}\\
&= \sum_{\mathbf{a} \in \mathbb{F}^{\vert V \vert}_{2} } (-1)^{g^{\prime}(\mathbf{a})} \vert \mathbf{a} \rangle \bigotimes_{\substack{u\in N(v_{c})\\ v \in V^{\prime}}} \Big[ \Bar{\mathcal{Z}}_{v_{c}} \Bar{\mathcal{Z}}_{u} \Bar{\mathcal{Z}}_{v} \vert \psi \rangle \Big] \,, \label{ali:appref0-5}
\end{align}
with \(g^{\prime}(\mathbf{a}) = \sum_{u \in N(v_{c})} a_{u} a_{v_{c}} + \sum_{\substack{v\in V^{\prime}\\v < w \in V}} a_{v} a_w \, \delta(w \in N(v)) + \sum_{\substack{u, w \in N(v_{c})\\ u<w}} a_{u} a_w \,.\)

For Eq.\;\eqref{ali:appref0-2} choose an arbitrary but fixed order of the vertices $u \in N(v_{c})$ and apply the following three steps sequentially for every vertex.
First, we use $\mathcal{M}_{u} \vert \psi \rangle = \vert \psi \rangle $ as $\mathcal{Y}_{u}^{a_{u}} \vert \psi \rangle = (\mathcal{Z}^{N(\{ v_{c},u\})} \mathcal{Y}_{v_{c}})^{a_{u}} \vert \psi \rangle$, then the anticommutation relation $\mathcal{Z}^{a_{u}}_{v} \mathcal{Y}^{a_{v}}_{v} \vert \psi \rangle = (-1)^{a_{u} a_{v}} \mathcal{Y}^{a_{v}}_{v} \mathcal{Z}^{a_{u}}_{v} \vert \psi \rangle$ for all remaining $u \in N( v_{c})$ and, third, $(\mathds{1}_{v} + \mathcal{Z}_{v}) \mathcal{Z}^{a_{u}}_{v} = (\mathds{1}_{v} + \mathcal{Z}_{v})$.

For Eq.\;\eqref{ali:appref0-3}, we anticommute $ \mathcal{X}^{a_{v_{c}}}_{v_{c}} $, $\mathcal{Y}_{v_{c}}^{ a_{u}}$, and absorb $(\mathds{1}_{v_{c}} + \mathrm{i} \mathcal{Y}_{v_{c}} \mathcal{X}_{v_{c}}) (\mathrm{i}\mathcal{Y}_{v_{c}})^{a_{u}} = (\mathds{1}_{v_{c}} + \mathrm{i} \mathcal{Y}_{v_{c}} \mathcal{X}_{v_{c}}) \mathcal{X}_{v_{c}}^{ a_{u}}$ for every $u \in N(v_{c})$.

For Eq.\;\eqref{ali:appref0-4}, we apply $\mathcal{X}_{v_{c}}^{a_{u}} \vert \psi \rangle = (\mathcal{Z}^{N(v_{c})})^{a_{u}} \vert \psi \rangle $ and $\mathcal{X}_{v}^{a_{v}} \vert \psi \rangle = (\mathcal{Z}^{N(v)})^{a_{v}} \vert \psi \rangle $ for an arbitrarily fixed order of $v \in V \setminus N(v_{c}) $ from $\xi_{v} \vert \psi \rangle = \vert \psi \rangle$.

As a result, we obtain \[ \Phi (\vert \mathbf{0} \rangle \otimes \vert \psi \rangle )= \vert G^{\prime} \rangle \otimes \vert \mathit{junk} \rangle \,, \hspace{0.25cm} \vert \mathit{junk} \rangle = 2^{\vert V \vert/2 } \bigotimes_{u \in V} \Bar{\mathcal{Z}}_{u} \vert \psi \rangle \,,\hspace{0.25cm}\vert G^{\prime} \rangle = \sum_{\mathbf{a} \in \mathbb{F}^{\vert V \vert}_{2} } (-1)^{g^{\prime}(\mathbf{a})} \vert \mathbf{a} \rangle /2^{\vert V \vert/2} \,.\]
After local complementation around vertex $v_{c}$, the graph state is the one in RE~\ref{def:refe0}, which can be attained by local unitary operations~\cite{hein2004multiparty}, such that
\begin{equation} \vert G \rangle = \frac{1}{2^{\vert V \vert/2 }} \sum_{\mathbf{a} \in \mathbb{F}^{\vert V \vert}_{2} } (-1)^{g(\mathbf{a})} \vert \mathbf{a} \rangle \,, ~~g(\mathbf{a}) = \sum_{u \in N(v_{c})} a_{u} a_{v_{c}} + \sum_{\substack{v\in V^{\prime}\\v < w \in V}} a_{v} a_w \, \delta(w \in N(v)) \,.\label{eq:appref0sign}\end{equation}

It remains to show that the isometry maps the measurements from the PE to the local measurements in the RE.
For any vertex $u$, the observables are proportional to $\mathcal{O}_{u}$, $2 (\Bar{\mathcal{Z}}_{0})_{u} - \mathds{1}_{u}$, or a product of the two.

Applying the isometry to $\mathcal{O}_{u} \vert \psi \rangle$ changes $(\mathcal{O}_{u})^{a_{u}} \rightarrow (\mathcal{O}_{u})^{a_{u} +1} $ in Eq.\;\eqref{ali:equiref0}.
Then, one substitutes $a_{u} \rightarrow a^{\prime}_{u} = a_{u} +1 $ so that $(\mathcal{O}_{u})^{a_{u}} \rightarrow (\mathcal{O}_{u})^{a^{\prime}_{u}}$ and $\vert a_{u} \rangle \rightarrow \vert a^{\prime}_{u} + 1 \rangle$, which is the action of the Pauli operator $X_{u}$.

The isometry acts on $\left(2 \Bar{\mathcal{Z}}_{u} - \mathds{1}_{u} \right) \vert \psi \rangle$ by changing $(\mathcal{O}_{u})^{a_{u}} \rightarrow (\mathcal{O}_{u})^{a_{u}} \left(2 \Bar{\mathcal{Z}}_{u} - \mathds{1}_{u} \right) $ in Eq.\;\eqref{ali:equiref0}.
We anticommute $(\mathcal{O}_{u})^{a_{u}} \left(2 \Bar{\mathcal{Z}}_{u} - \mathds{1}_{u} \right) = \left(2 (\Bar{\mathcal{Z}}_{a_{u}})_{u} - \mathds{1}_{u} \right)(\mathcal{O}_{u})^{a_{u}} $.
Absorbing $\Bar{\mathcal{Z}}_{u} (\Bar{\mathcal{Z}}_{a_{u}})_{u} = \delta_{a_{u},1} \Bar{\mathcal{Z}}_{u}$ leaves a phase $\left(2 \delta_{a_{u},1} - \mathds{1}_{u} \right) = (-1)^{a_{u}}$ such that $\vert a_{u} \rangle \rightarrow (-1)^{a_{u}} \vert a_{u} \rangle$, the action of Pauli operator $Z_{u}$.
The composition of the observables in $\mathcal{O}_{u}$ and $2 \Bar{\mathcal{Z}}_{u} - \mathds{1}_{u} $ matches the corresponding Pauli operators in the reference experiment.

We conclude that the physical experiment and the reference experiment are equivalent.
\end{proof}

\section{Tensor-product representation of Pauli correlations}
\label{app:expect}

We start by introducing some useful notation.
Given a set of vertices $V$, we write $O^{V} = \bigotimes_{v \in V} O_{v}$ for local operators $O_{v}$.
In particular, if $O_{v}$ are Pauli operators, $O^V$ describes a Pauli measurement whose outcome is a string of binary outputs from each vertex $v$.
Furthermore, we label the chain vertices $s_{(u,v)}$ as odd (`o') if $s$ is odd and even (`e') if $s$ is even.
If $s \leq d$, the vertex is on the left (`L'), and if $s > d$, on the right (`R') side of the chain.
Figure~\ref{fig:nota} illustrates the notation.
We denote the tensor products of Pauli $X$ operators on the set of vertices for these labels as
\begin{equation*}
\underset{(u,v)}{ {X_{\text{e},\text{L}}}} = \bigotimes^{\lfloor d/2 \rfloor}_{s=1} X_{2s_{(u,v)}} \,,~~ \underset{(u,v)}{ {X_{\text{e},\text{R}}}} = \bigotimes^{d}_{s= \lceil d/2 \rceil} X_{(2s)_{(u,v)}} \,,~~
\underset{(u,v)}{ {X_{\text{o},\text{L}}}} = \bigotimes^{\lceil d/2 \rceil}_{s=1} X_{(2s-1)_{(u,v)}}\,,~~ \underset{(u,v)}{ {X_{\text{o},\text{R}}}} = \bigotimes^{d}_{s= \lfloor d/2 \rfloor} X_{(2s-1)_{(u,v)}}\,,
\end{equation*}
and \( \underset{(u,v)}{ {X_{\text{o}}}} = \underset{(u,v)}{ {X_{\text{o},\text{L}}}} \underset{(u,v)}{ {X_{\text{o},\text{R}}}}\,,\,\underset{(u,v)}{ {X_{\text{e}}}} = \underset{(u,v)}{ {X_{\text{e},\text{L}}}} \underset{(u,v)}{ {X_{\text{e},\text{R}}}}\,. \)
A tensor product of multiple chains, given a Pauli product $P$ over a single one, is \(\underset{(u,N(u))}{P} = \bigotimes_{v \in N(u)} \underset{(u,v)}{P}\,, \) for a set of $V$.

In the following, we use this notation to write all correlation operators from the reference experiments (REs)~\ref{def:refe1},~\ref{def:refe2},~\ref{def:refe3} in terms of tensor products.
Next to them, we also write the observables that a compatible physical experiment (PE) assigns.
We use calligraphic font for the measurements and observables of the PE, in contrast to italic font for the REs.
We denote the measurements of the PE corresponding to the inflated generator element $f_{u}$ by $\zeta_{u}$.
Recall that we label all measurement observables according to the local measurement setting.
The information gained from the classical communication is accounted for in the superscripts.
To keep these minimal, we only label the chain vertices according to the measurement setting of its nearest power vertex.
This is possible since the measurement settings of the chain vertices are mostly the same, and we specifically outline if they change.
The notation for the tensor product operator is in Tab.~\ref{tab:penotation}

The $d$-inflated generator element is \begin{align}f_{u} &=X_{u} Z^{N(u)} \underset{(u,N(u))}{X_{\text{e}}} \,, \label{eq:infstabpauli} \\
\zeta_{u} &= \mathcal{X}_{u} \mathcal{Z}^{N(u)} \underset{(u,N(u))}{\mathcal{X}^{(X)}_{\text{e},L} \mathcal{X}^{(Z)}_{\text{e},R}}
\label{eq:pe-infstabpauli} \,, \end{align}
which also brings out its resemblance with the generator element of the original graph state $g_{u} = X_{u} Z^{N(u)}$.
\begin{figure}
\centering
\begin{tikzpicture}[scale=1.1]
\begin{scope}

\node[style={circle,draw},minimum size=0.85cm] (u) at (0,0) {};
\node[style={circle,draw},minimum size=0.75cm] (121) at (1,0) {o};
\node[style={circle,draw},minimum size=0.75cm] (122) at (2,0) {e};
\node[style={circle,draw},minimum size=0.75cm] (123) at (3,0) {o};
\node[style={circle,draw},minimum size=0.75cm] (124) at (4,0) {e};
\draw (4,1) node [anchor=north east][inner sep=0.75pt] [align=left] {L};
\draw[dashed] (4.6,-1) -- (4.6,1) ;
\draw (5.2,1) node [anchor=north west][inner sep=0.75pt] [align=left] {R};
\node[style={circle,draw},minimum size=0.75cm] (125) at (5.2,0) {o};
\node[style={circle,draw},minimum size=0.75cm] (126) at (6.2,0) {e};
\node[style={circle,draw},minimum size=0.75cm] (127) at (7.2,0) {o};
\node[style={circle,draw},minimum size=0.75cm] (128) at (8.2,0) {e};
\node[style={circle,draw},minimum size=0.85cm] (v) at (9.2,0) {};

\draw (0,-0.6) node [anchor=north][inner sep=0.75pt] [align=left] {$u$};
\draw (1,-0.5) node [anchor=north][inner sep=0.75pt] [align=left] {$1_{(u,v)}$};
\draw (2,-0.5) node [anchor=north][inner sep=0.75pt] [align=left] {$2_{(u,v)}$};
\draw (3,-0.5) node [anchor=north][inner sep=0.75pt] [align=left] {$3_{(u,v)}$};
\draw (4,-0.5) node [anchor=north][inner sep=0.75pt] [align=left] {$4_{(u,v)}$};
\draw (5.2,-0.5) node [anchor=north][inner sep=0.75pt] [align=left] {$5_{(u,v)}$};
\draw (6.2,-0.5) node [anchor=north][inner sep=0.75pt] [align=left] {$6_{(u,v)}$};
\draw (7.2,-0.5) node [anchor=north][inner sep=0.75pt] [align=left] {$7_{(u,v)}$};
\draw (8.2,-0.5) node [anchor=north][inner sep=0.75pt] [align=left] {$8_{(u,v)}$};
\draw (9.2,-0.6) node [anchor=north][inner sep=0.75pt] [align=left] {$v$};

\draw[-] (u) -- (121); 
\draw[-] (121) -- (122);
\draw[-] (122) -- (123);
\draw[-] (123) -- (124);
\draw[-] (124) -- (125);
\draw[-] (125) -- (126);
\draw[-] (126) -- (127);
\draw[-] (127) -- (128);
\draw[-] (128) -- (v);
\end{scope}
\end{tikzpicture}
\caption{For a $d=4$-inflated graph from two connected vertices $u$ and $v$, the vertices' labels are below, and their characterization as odd (o) or even (e) inside the nodes.
The chain is divided into a left (L) and right (R).}
\label{fig:nota}
\end{figure}
Given a graph $G = (V,E)$ and an induced subgraph $G_U = (U \subset V, E_U = (U \times U) \cap E)$, we define
\begin{align}
S_{u} &= Z^{N(u) \setminus U} \, {X_{\text{e}}}_{(u,N(u) \setminus U)} \,, \\
\mathcal{S}_{u} &= \mathcal{Z}^{N(u) \setminus U} \, \underset{(u,N(u) \setminus U)}{\mathcal{X}^{(X)}_{\text{e,L}}\mathcal{X}^{(Z)}_{\text{e,R}}} \,.
\end{align}
For example, if the induced subgraph consisting of a single vertex $u$, the inflated generator element can be written as $f_{u} = X_{u} S_{u}$.

\begin{table}[]
\centering
\begin{tabular}{Sr|Sr|Sl|Sl}
& operators & measurement on power vertices & measurement on chain vertices \\
\cline{1-4}
\multirow{2}{*}{power vertices}& $\mathcal{Q}_{u}$\, & & \\
\cline{2-4}
& $\Tilde{\mathcal{Q}}_{u}$\, & & chain vertices $N^{\prime}(u)$: Pauli-$Y$ \\
\cline{1-4}
\multirow{6}{*}{chain vertices}& $\mathcal{X}^{(\sigma)}_{u}$ & power vertex $u$: Pauli-$\sigma$ \\
\cline{2-4}
& $\Tilde{\mathcal{X}}^{(\sigma)}_{u}$ & power vertex $u$: Pauli-$\sigma$ & chain vertices $N^{\prime}(u)$: Pauli-$Y$ \\
\cline{2-4}
& $\underset{(u,v)}{\mathcal{X}^{(\sigma_{1})}_{\text{a,L}} \mathcal{X}^{(\sigma_{2})}_{\text{a,R}}}$ & \begin{tabular}[t]{l} power vertex $u$: Pauli-$\sigma_{1}$;\\ power vertex $u$: Pauli-$\sigma_{2}$\end{tabular} & \\
\cline{2-4}
& $\underset{(u,v)}{\Tilde{\mathcal{X}}^{(\sigma_{1})}_{\text{a,L}}\Tilde{\mathcal{X}}^{(\sigma_{2})}_{\text{a,R}}}$ & \begin{tabular}[t]{l} power vertex $u$: Pauli-$\sigma_{1}$; \\ power vertex $v$: Pauli-$\sigma_{2}$ \end{tabular} & chain vertices $N^{\prime}(u)$: Pauli-$Y$ \\
\cline{2-4}
& $\underset{(v_{l},v_{r})}{\Tilde{\mathcal{X}}^{(\sigma,R_{z})}_{\text{a},\text{L}}}$ & power vertex $v_{l}$: Pauli-$\sigma$ & \begin{tabular}[t]{l} chain vertex $v_{m}$: $R = (X+Y)/\sqrt{2}$; \\ chain vertices $N^{\prime}(v_{m})$: Pauli-$Y$ \end{tabular} \\
\cline{2-4}
& $\underset{(v_{l},v_{r})}{\Tilde{\mathcal{X}}^{(\sigma,R_{z})}_{\text{a},\text{R}}}$ & power vertex $v_{r}$: Pauli-$\sigma$ & \begin{tabular}[t]{l} chain vertex $v_{m}$: $R_{z}=(X+Y)/\sqrt{2}$; \\ chain vertices $N^{\prime}(v_{m})$: Pauli-$Y$ \end{tabular}
\end{tabular}
\caption{Overview of operator notation for the physical experiment with bounded communication.
By default, all vertices within communication range measure in the Pauli-$X$ basis.
The table indicates all the cases, where vertices within communication range do not measure in the Pauli-$X$ basis.
Note that $\text{a} = \text{o},\text{e}$.}
\label{tab:penotation}
\end{table}

For the reference experiment~\ref{def:refe1},
\begin{align}
C_{V_{c}} &= X^{V_{c}} \left( \bigotimes_{(u,v) \in E_{c}} \underset{(u,v)}{X_{\text{o}}} \, \underset{(u,v)}{X_{\text{e}}} \right) \, S^{V_{c}} \,, \label{eq:measref12} \\ 
C_{ w_{\phantom{0}}} &= X^{V_{c}\setminus (N(w) \cup \{ w \})} \left( \bigotimes_{(u,v) \in E_{c} \setminus (w,N(w))} \underset{(u,v)}{X_{\text{o}}} \, \underset{(u,v)}{X_{\text{e}}}\, \right) Y^{N(w) \cap V_{c}} \underset{ (w,N(w)\cap V_{c} )}{ {X_{\text{o}}}} \, S^{V_{c}\setminus \{ w \}} \,, \label{eq:measref13}
\end{align}
and, for the PE, taking into account the round of classical communication of the measurement settings,
\begin{align}
\mathcal{C}_{V_{c}} &= \mathcal{X}^{V_{c}} \left( \bigotimes_{(u,v) \in E_{c}} \underset{ (u,v) }{ \mathcal{X}^{(X)}_{\text{o,L}} \mathcal{X}^{(X)}_{\text{o,R}} } \,\underset{ (u,v) }{ \mathcal{X}^{(X)}_{\text{e,L}} \mathcal{X}^{(X)}_{\text{e,R}} } \right) \, \mathcal{S}^{V_{c}} \,, \label{eq:pe-measref12} \\ 
\mathcal{C}^{(X)}_{ w_{\phantom{0}}} &= \mathcal{X}^{V_{c}\setminus (N(v) \cup \{ v \})} \left( \bigotimes_{(u,v) \in E_{c} \setminus (N(N(w)),N(w))} \underset{ (u,v) }{ \mathcal{X}^{(X)}_{\text{o,L}} \mathcal{X}^{(X)}_{\text{o,R}} } \,\underset{ (u,v) }{ \mathcal{X}^{(X)}_{\text{e,L}} \mathcal{X}^{(X)}_{\text{e,R}} }\, \right)\mathcal{Y}^{N(w) \cap V_{c}} \,\underset{ (w,N(w)\cap V_{c} )}{ \mathcal{X}^{(X)}_{\text{o,L}} \, \mathcal{X}^{(Y)}_{\text{o,R}}} \, \mathcal{S}^{V_{c}\setminus \{ w \}} \,, \label{eq:pe1-measref13} \\
\mathcal{C}^{(Z)}_{ w_{\phantom{0}}} &= \mathcal{X}^{V_{c}\setminus (N(v) \cup \{ v \})} \left( \bigotimes_{(u,v) \in E_{c} \setminus (N(N(w)),N(w))} \underset{ (u,v) }{ \mathcal{X}^{(X)}_{\text{o,L}} \mathcal{X}^{(X)}_{\text{o,R}} } \,\underset{ (u,v) }{ \mathcal{X}^{(X)}_{\text{e,L}} \mathcal{X}^{(X)}_{\text{e,R}} }\, \right)\mathcal{Y}^{N(w) \cap V_{c}} \underset{ (w,N(w)\cap V_{c} )}{ \mathcal{X}^{(Z)}_{\text{o,L}} \mathcal{X}^{(Y)}_{\text{o,R}}} \, \mathcal{S}^{V_{c}\setminus \{ w \}} \,, \label{eq:pe2-measref13}
\end{align}
for all $w \in V_{c}$.
The operators are illustrated in Fig.~\ref{fig:exref1} for an exemplary graph.

For the reference experiment~\ref{def:refe2},
\begin{equation}
\Tilde{\zeta}_{v_{c}} = \Tilde{\mathcal{X}}_{v_{c}} \mathcal{Z}^{N(v_{c})} \,\underset{(v_{c},N(v_{c}))}{\Tilde{\mathcal{X}}^{(X)}_{\text{e},L}\Tilde{\mathcal{X}}^{(Z)}_{\text{e},R}} \,,\label{eq:pe-refe25}
\end{equation}
and, if $\vert N(v_{c})\vert >2$,
\begin{align}
C_{v_{i}v_{j}} &= X_{v_{c}}\, X_{v_{i}} X_{v_{j}} \,Z^{N(v_{c}) \setminus \{ v_{i},v_{j} \}} \underset{ (v_{c},v_{i}) }{ {\Tilde{X}_{\text{o}}}} \underset{ (v_{c},v_{j}) }{ {\Tilde{X}_{\text{o}}}} \underset{ (v_{c},N(v_{c})\setminus \{v_{i},v_{j}\}) }{ \Tilde{X}_{\text{e}}} \,S_{v_{i}} S_{v_{j}} \,,\label{eq:refe21} \\
\mathcal{C}_{v_{i}v_{j}} &= \Tilde{\mathcal{X}}_{v_{c}}\, \mathcal{X}_{v_{i}} \mathcal{X}_{v_{j}} \,\mathcal{Z}^{N(v_{c}) \setminus \{ v_{i},v_{j} \}} \, \underset{ (v_{c},v_{i}) }{ \Tilde{\mathcal{X}}^{(X)}_{\text{o,L}} \Tilde{\mathcal{X}}^{(X)}_{\text{o,R}} } \,\underset{ (v_{c},v_{j}) }{ \Tilde{\mathcal{X}}^{(X)}_{\text{o,L}} \Tilde{\mathcal{X}}^{(X)}_{\text{o,R}} } \, \underset{ (v_{c},N(v_{c})\setminus \{v_{i},v_{j}\}) }{ \Tilde{\mathcal{X}}^{(X)}_{\text{e,L}}\Tilde{\mathcal{X}}^{(Z)}_{\text{e,R}}} \,\mathcal{S}_{v_{i}} \mathcal{S}_{v_{j}} \,,\label{eq:pe-refe21}
\end{align}
for $V_{3} = \{ v_{1},v_{2},v_{3} \} \subseteq N(v_{c})$ and $v_{i} \neq v_{j}$.
If $N(v_{c}) = \{ v_{1},v_{2} \}$ such that $\vert N(v_{c}) \vert =2$, for the RE
\begin{align}
\Tilde{f}_{v_{c}} &= \,X_{v_{c}}\, Z_{v_{1}} \,Z_{v_{2}} \phantom{ \underset{ (v_{c},v_{1}) }{ \Tilde{X}_{\text{o}}}} \underset{ (v_{c},v_{1}) }{ {\Tilde{X}_{\text{e}}}} \phantom{\underset{ (v_{c},v_{2}) }{ \Tilde{X}_{\text{o}}}} \underset{ (v_{c},v_{2}) }{ \Tilde{X}_{\text{e}}} \,\phantom{S_{v_{1}}} \phantom{S_{v_{2}}} \,,\label{eq:refe210} \\
C_{v_{c}}^{\phantom{X}} &= X_{v_{c}} X_{v_{1}} X_{v_{2}} \, \underset{(v_{c},v_{1})}{ {\Tilde{X}_{\text{o}}}} \phantom{ \underset{ (v_{c},v_{1}) }{ \Tilde{X}_{\text{e}}}} \underset{(v_{c},v_{2})}{ {\Tilde{X}_{\text{o}}}} \phantom{ \underset{ (v_{c},v_{2}) }{ \Tilde{X}_{\text{e}}}}\, S_{v_{1}} S_{v_{2}}\label{eq:refe23} \,, \\
C_{v_{1}} &= \,Y_{v_{c}}\, X_{v_{1}} \,Z_{v_{2}} \underset{ (v_{c},v_{1}) }{ {\Tilde{X}_{\text{o}}}} \phantom{ \underset{ (v_{c},v_{1}) }{ \Tilde{X}_{\text{o}}}}\phantom{\underset{ (v_{c},v_{2}) }{ \Tilde{X}_{\text{o}}}} \underset{ (v_{c},v_{2}) }{ \Tilde{X}_{\text{e}}} \,S_{v_{1}} \,,\label{eq:refe211} \\
C_{v_{1}} &=\, Y_{v_{c}}\, Z_{v_{1}}\, X_{v_{2}} \phantom{\underset{(v_{c},v_{1})}{\Tilde{X}_{\text{o}}}} \underset{ (v_{c},v_{1}) }{ {\Tilde{X}_{\text{e}}}} \underset{ (v_{c},v_{2}) }{ \Tilde{X}_{\text{o}}} \phantom{\underset{(v_{c},v_{2})}{\Tilde{X}_{\text{e}}}} \,\phantom{S_{v_{1}}} S_{v_{2}} \,,\label{eq:refe212} \\
C_{2} \,&= \,\phantom{Y_{v_{c}}}\, Y_{v_{1}} \, Y_{v_{2}}\,\underset{(v_{c},v_{1})}{\Tilde{X}_{\text{o}}} \underset{(v_{c},v_{1})}{\Tilde{X}_{\text{e}}} \, \underset{(v_{c},v_{2})}{\Tilde{X}_{\text{o}} } \underset{(v_{c},v_{2})}{\Tilde{X}_{\text{e}}} \,S_{v_{1}} S_{v_{2}} \,,\label{eq:refe24}
\end{align}
and, for the PE, taking into account the round of classical communication of the measurement settings,
\begin{align}
\Tilde{\zeta}_{v_{c}} &= \Tilde{\mathcal{X}}_{v_{c}}\, \mathcal{X}_{v_{1}} \mathcal{Z}_{v_{2}} \,
~\phantom{\underset{ (v_{c},v_{1}) }{ \Tilde{\mathcal{X}}^{(X)}_{\text{o,L}} \Tilde{\mathcal{X}}^{(Z)}_{\text{o,R}} }} ~\underset{ (v_{c},v_{1}) }{ \Tilde{\mathcal{X}}^{(X)}_{\text{e,L}} \Tilde{\mathcal{X}}^{(Z)}_{\text{e,R}} } \,
~\phantom{\underset{(v_{c},v_{1})}{ \Tilde{\mathcal{X}}^{(X)}_{\text{o,L}} \Tilde{\mathcal{X}}^{(Z)}_{\text{o,R}}}} ~\underset{ (v_{c},v_{2}) }{ \Tilde{\mathcal{X}}^{(X)}_{\text{e,L}}\Tilde{\mathcal{X}}^{(Z)}_{\text{e,R}}} \,
~\mathcal{S}_{v_{1}} \phantom{\mathcal{S}_{v_{2}}}\,,\label{eq:pe-refe210}\\
\mathcal{C}_{v_{c}}^{\phantom{X}} &= \Tilde{\mathcal{X}}_{v_{c}}\, \mathcal{X}_{v_{1}} \mathcal{X}_{v_{2}} \,
~\underset{(v_{c},v_{1})}{ \Tilde{\mathcal{X}}^{(X)}_{\text{o,L}} \Tilde{\mathcal{X}}^{(X)}_{\text{o,R}}} ~\phantom{\underset{ (v_{c},v_{1}) }{ \Tilde{\mathcal{X}}^{(X)}_{\text{e,L}} \Tilde{\mathcal{X}}^{(X)}_{\text{e,R}} }}\, 
~\underset{(v_{c},v_{2})}{ \Tilde{\mathcal{X}}^{(X)}_{\text{o,L}} \Tilde{\mathcal{X}}^{(X)}_{\text{o,R}}} ~\phantom{\underset{ (v_{c},v_{2}) }{ \Tilde{\mathcal{X}}^{(X)}_{\text{e,L}} \Tilde{\mathcal{X}}^{(X)}_{\text{e,R}} }}\,
~\mathcal{S}_{v_{1}} \mathcal{S}_{v_{2}} \label{eq:pe-refe23} \,, \\
\mathcal{C}_{v_{1}} &= \Tilde{\mathcal{Y}}_{v_{c}}\, \mathcal{X}_{v_{1}} \mathcal{Z}_{v_{2}} \, 
~\underset{ (v_{c},v_{1}) }{ \Tilde{\mathcal{X}}^{(Y)}_{\text{o,L}} \Tilde{\mathcal{X}}^{(X)}_{\text{o,R}} }~ \phantom{\underset{ (v_{c},v_{1}) }{ \Tilde{\mathcal{X}}^{(Y)}_{\text{e,L}} \Tilde{\mathcal{X}}^{(X)}_{\text{e,R}} }} \, 
~\phantom{\underset{ (v_{c},v_{1}) }{ \Tilde{\mathcal{X}}^{(Y)}_{\text{o,L}} \Tilde{\mathcal{X}}^{(Z)}_{\text{o,R}} }} ~\underset{ (v_{c},v_{2}) }{ \Tilde{\mathcal{X}}^{(Y)}_{\text{e,L}}\Tilde{\mathcal{X}}^{(Z)}_{\text{e,R}}} \,
~\mathcal{S}_{v_{1}} \phantom{\mathcal{S}_{v_{2}}}\,,\label{eq:pe-refe211}\\
\mathcal{C}_{v_{2}} &= \Tilde{\mathcal{Y}}_{v_{c}}\, \mathcal{Z}_{v_{1}} \mathcal{X}_{v_{2}} \,
~\phantom{\underset{ (v_{c},v_{1}) }{ \Tilde{\mathcal{X}}^{(Y)}_{\text{o,L}}\Tilde{\mathcal{X}}^{(Z)}_{\text{o,R}}}} ~ \underset{ (v_{c},v_{1}) }{ \Tilde{\mathcal{X}}^{(Y)}_{\text{e,L}}\Tilde{\mathcal{X}}^{(Z)}_{\text{e,R}}} \,
\underset{ (v_{c},v_{2}) }{ \Tilde{\mathcal{X}}^{(Y)}_{\text{o,L}} \Tilde{\mathcal{X}}^{(X)}_{\text{o,R}} } ~\phantom{\underset{ (v_{c},v_{2}) }{ \Tilde{\mathcal{X}}^{(Y)}_{\text{e,L}} \Tilde{\mathcal{X}}^{(X)}_{\text{e,R}} }} \,
~\phantom{\mathcal{S}_{v_{1}}} \mathcal{S}_{v_{2}} \,,\label{eq:pe-refe212}\\
\mathcal{C}^{(X)}_{2} &= \phantom{ \Tilde{\mathcal{Y}}_{v_{c}}}\, \mathcal{Y}_{v_{1}} \mathcal{Y}_{v_{2}} \,
~\underset{ (v_{c},v_{1}) }{ \Tilde{\mathcal{X}}^{(X)}_{\text{o,L}}\Tilde{\mathcal{X}}^{(Y)}_{\text{o,R}}} ~ \underset{ (v_{c},v_{1}) }{ \Tilde{\mathcal{X}}^{(X)}_{\text{e,L}} \Tilde{\mathcal{X}}^{(Y)}_{\text{e,R}}} \,
~ \underset{ (v_{c},v_{2}) }{ \Tilde{\mathcal{X}}^{(X)}_{\text{o,L}} \Tilde{\mathcal{X}}^{(Y)}_{\text{o,R}} } ~ \underset{ (v_{c},v_{2}) }{ \Tilde{\mathcal{X}}^{(X)}_{\text{e,L}} \Tilde{\mathcal{X}}^{(Y)}_{\text{e,R}} } \,
~\mathcal{S}_{v_{1}} \mathcal{S}_{v_{2}} \,,\label{eq:pe-refe24-1}\\
\mathcal{C}^{(Y)}_{2} &= \phantom{ \Tilde{\mathcal{Y}}_{v_{c}}}\, \mathcal{Y}_{v_{1}} \mathcal{Y}_{v_{2}} \,
~\underset{ (v_{c},v_{1}) }{ \Tilde{\mathcal{X}}^{(Y)}_{\text{o,L}}\Tilde{\mathcal{X}}^{(Y)}_{\text{o,R}}} ~ \,\underset{ (v_{c},v_{1}) }{ \Tilde{\mathcal{X}}^{(Y)}_{\text{e,L}} \Tilde{\mathcal{Y}}^{(Y)}_{\text{e,R}}} \,
~\underset{ (v_{c},v_{2}) }{ \Tilde{\mathcal{X}}^{(Y)}_{\text{o,L}} \Tilde{\mathcal{X}}^{(Y)}_{\text{o,R}} } ~ \,\underset{ (v_{c},v_{2}) }{ \Tilde{\mathcal{X}}^{(Y)}_{\text{e,L}} \Tilde{\mathcal{X}}^{(Y)}_{\text{e,R}} } \,
~\mathcal{S}_{v_{1}} \mathcal{S}_{v_{2}} \,.\label{eq:pe-refe24-2}
\end{align}
The operator $\Tilde{X}_{(v_{c},v)}$ for every edge $(v_{c},v) \in E$ is, similar to $X_{(v_{c},v)}$, a tensor product of Pauli $X$ operators but on the nearest neighbors of $v_{c}$ where it has the Pauli $Y$ operator.
The operator $\Tilde{\mathcal{X}}_{v_{c}}$ is the same as $\mathcal{X}_{v_{c}}$ but accounts for the nearest neighbors of $v_{c}$ measuring according to the Pauli $Y$ operator instead of the Pauli $X$ operator.
The Pauli correlations are illustrated in Fig.~\ref{fig:exref3} and Fig.~\ref{fig:exref3p} for exemplary graphs.

For the RE~\ref{def:refe3}, the representation of measurement correlations from Eq.\;\eqref{eq:ref3} in tensor product operators is
\begin{align}
I_{1} &= Z_{v_{l}} \underset{(v_{l},v_{r})}{\Tilde{X}_{\text{o},\text{L}}} (R_z)_{v_m} \underset{(v_{l},v_{r})}{\Tilde{X}_{\text{o},\text{R}}} \, X_{v_{r}}\, S_{v_{r}}\,,\label{ali:pcorr41}\\
I_{2} &= O_{v_{l}} \underset{(v_{l},v_{r})}{\Tilde{X}_{\text{e},\text{L}}} (R_z)_{v_m} \underset{(v_{l},v_{r})}{\Tilde{X}_{\text{e},\text{R}}} \, Z_{v_{r}} \, S_{v_{l}} \,, \label{ali:pcorr42}\\
I_{3} &= O_{v_{l}} \underset{(v_{l},v_{r})}{\Tilde{X}_{\text{e},\text{L}}} (R_z)_{v_m} \underset{(v_{l},v_{r})}{\Tilde{X}_{\text{o},\text{R}}} \, X_{v_{r}}\, S_{v_{l}} S_{v_{r}}\,,\label{ali:pcorr43}\\
I_{4} &= Z_{v_{l}} \underset{(v_{l},v_{r})}{\Tilde{X}_{\text{o},\text{L}}} (R_z)_{v_m} \underset{(v_{l},v_{r})}{\Tilde{X}_{\text{e},\text{R}}} \, Z_{v_{r}} \,,\label{ali:pcorr44}\\
P_{1} &= Z_{v_{l}} \underset{(v_{l},v_{r})}{\Tilde{X}_{\text{o},\text{L}}} ~X_{v_m} ~ ~\underset{(v_{l},v_{r})}{\Tilde{X}_{\text{o},\text{R}}} ~ X_{v_{r}}\, S_{v_{r}}\,,\label{ali:pcorr45}\\
P_{2} &= O_{v_{l}} \underset{(v_{l},v_{r})}{\Tilde{X}_{\text{e},\text{L}}} ~X_{v_m} ~ ~\underset{(v_{l},v_{r})}{\Tilde{X}_{\text{e},\text{R}}} \, Z_{v_{r}} \,S_{v_{l}} \,,\label{ali:pcorr46}
\end{align}
and, for the PE taking into account the round of classical communication of the measurement settings for the measurements $\mathcal{M}$,
\begin{align}
\mathcal{I}_{1} &= \Tilde{\mathcal{Z}}^{(R_{z})}_{v_{l}} \underset{(v_{l},v_{r})}{\Tilde{\mathcal{X}}^{(Z,R_{z})}_{\text{o},\text{L}}}~ (\Tilde{\mathcal{R}}^{(Z)}_z)_{v_m}~ \underset{(v_{l},v_{r})}{\Tilde{\mathcal{X}}^{(X,R_{z})}_{\text{o},\text{R}}} \, \Tilde{\mathcal{X}}_{v_{r}}\, \mathcal{S}_{v_{r}}\,,\label{ali:pcorr41-pe}\\
\mathcal{I}_{2} &= \Tilde{\mathcal{O}}^{(R_{z})}_{v_{l}} \underset{(v_{l},v_{r})}{\Tilde{\mathcal{X}}^{(O,R_{z})}_{\text{e},\text{L}}}~ (\Tilde{\mathcal{R}}^{(O)}_z)_{v_m}~ \underset{(v_{l},v_{r})}{\Tilde{\mathcal{X}}^{(Z,R_{z})}_{\text{e},\text{R}}} \, \Tilde{\mathcal{Z}}_{v_{r}} \, \mathcal{S}_{v_{l}} \,,\label{ali:pcorr42-pe}\\
\mathcal{I}_{3} &= \Tilde{\mathcal{O}}^{(R_{z})}_{v_{l}} \underset{(v_{l},v_{r})}{\Tilde{\mathcal{X}}^{(O,R_{z})}_{\text{e},\text{L}}} ~(\Tilde{\mathcal{R}}^{(O)}_z)_{v_m}~ \underset{(v_{l},v_{r})}{\Tilde{\mathcal{X}}^{(X,R_{z})}_{\text{o},\text{R}}} \, \Tilde{\mathcal{X}}_{v_{r}}\, \mathcal{S}_{v_{l}} \mathcal{S}_{v_{r}}\,,\label{ali:pcorr43-pe}\\
\mathcal{I}_{4} &= \Tilde{\mathcal{Z}}^{(R_{z})}_{v_{l}} \underset{(v_{l},v_{r})}{\Tilde{\mathcal{X}}^{(Z,R_{z})}_{\text{o},\text{L}}}~ (\Tilde{\mathcal{R}}^{(Z)}_z)_{v_m} ~\underset{(v_{l},v_{r})}{\Tilde{\mathcal{X}}^{(Z,R_{z})}_{\text{e},\text{R}}} \, \Tilde{\mathcal{Z}}_{v_{r}} \,,\,,\label{ali:pcorr44-pe}\\
\mathcal{P}_{1} &= \Tilde{\mathcal{Z}}_{v_{l}} \underset{(v_{l},v_{r})}{\Tilde{\mathcal{X}}^{(Z)}_{\text{o},\text{L}}} ~\mathcal{X}^{(Z)}_{v_m} ~ ~\underset{(v_{l},v_{r})}{\Tilde{\mathcal{X}}^{(X)}_{\text{o},\text{R}}} ~ \Tilde{\mathcal{X}}_{v_{r}}\, \mathcal{S}_{v_{r}}\,,\label{ali:pcorr45-pe}\\
\mathcal{P}_{2} &= \Tilde{\mathcal{O}}_{v_{l}} \underset{(v_{l},v_{r})}{\Tilde{\mathcal{X}}^{(O)}_{\text{e},\text{L}}} ~\mathcal{X}^{(O)}_{v_m} ~ ~\underset{(v_{l},v_{r})}{\Tilde{\mathcal{X}}^{(Z)}_{\text{e},\text{R}}} \, \Tilde{\mathcal{Z}}_{v_{r}} \,\mathcal{S}_{v_{l}} \,,\label{ali:pcorr46-pe}
\end{align}
with $v_m := d_{(v_{l},v_{r})}$ and \(R_z = \big( X + Y \big)/\sqrt{2} \,,\)
It is $O_{v_{l}} = Y_{v_{l}}$ if $d=1$, and else $O_{v_{l}} = X_{v_{l}}$.
The operator $\Tilde{X}_{(v_{l},v_{r})}$ is, similar to $X_{(v_{l},v_{r})}$, a tensor product of Pauli $X$ operators but on the chain vertices $(d\pm1)_{(v_{l},v_{r})}$, where it has the Pauli $Y$ operator, and on the chain vertex $v_m$, where it has the Pauli identity.
Any observable $\Tilde{\mathcal{Q}}$, is the same as a corresponding observable $\mathcal{Q}$ but accounts for the nearest neighbors of the vertex $v_m$ measuring according to the Pauli $Y$ operator instead of the Pauli $X$ operator.
The operators $I_{1},\dots,I_{4}$ are illustrated in Fig.~\ref{fig:exref4} for an exemplary graph.

\section{Proof of Lemma~\ref{lem:ref1234} \& Lemma~\ref{lem:indu} }
\label{app:anti}
\begin{proof}[Proof [Lemma~\ref{lem:ref1234}, Reference experiment~\ref{def:refe1}]]
If the physical experiment (PE) simulates the reference experiment (RE)~\ref{def:refe1}, the specified correlations from the measurements fulfill
\begin{align}\label{pe1.1}
\langle \Psi\vert \, \zeta_{u} \, \vert\Psi\rangle &= \phantom{-} 1\,, \qquad \forall u \in V\,,\\ \label{pe1.2}
\langle \Psi\vert \, \mathcal{C}_{V} \,\vert\Psi\rangle &= -1\,, \\ \label{pe1.3}
\langle \Psi\vert\mathcal{C}^{(X)}_{v}\vert\Psi\rangle &= \phantom{-} 1\,, \qquad \forall v \in V_{c} \,, \\ \label{pe1.4}
\langle \Psi\vert\mathcal{C}^{(Z)}_{v}\vert\Psi\rangle &= \phantom{-} 1\,, \qquad \forall v \in V_{c} \,.
\end{align}
Since the norm of all physical measurement observables is at most $1$, Eqs.~\eqref{pe1.1} -- \eqref{pe1.4} imply
\begin{align}\label{pe1.5}
\zeta_{u} \,\vert \Psi \rangle &= \vert\Psi\rangle \,, \qquad \forall u \in V \,, \\ \label{pe1.6}
(-1)\,\mathcal{C}_{V} \, \vert \Psi \rangle &= \vert\Psi\rangle \,,\\ \label{pe1.7}
\mathcal{C}^{(X)}_{v}\, \vert \Psi \rangle &= \vert \Psi \rangle \,, \qquad \forall v \in V_{c} \,,\\ \label{pe1.8}
\mathcal{C}^{(Z)}_{v}\, \vert \Psi \rangle &= \vert \Psi \rangle \,, \qquad \forall v \in V_{c} \,.
\end{align}
Recall that any vertex's observable depends on the measurement observables from all vertices $v \in N_d(u)$. In the RE one power vertex from the odd cycle does not measure any observable when the reference measurement is $C_{v}$. However, as explained in the main text, whenever a vertex is not asked to measure anything in the RE, in the PE it is given an input, but the results are marginalized. Thus, physical measurement $\mathcal{C}_{v}$ has two non-equivalent implementations, one in which the vertex which does not measure anything in the RE is asked to perform the measurement corresponding to $X$ and the other in which it is asked to measure $Z$. We write $\mathcal{C}^{(X)}_{v}$ and $\mathcal{C}^{(Z)}_{v}$ for these two physical implementations, respectively.
Both operators have congruent observables except on the chain vertices with the nearest power vertex $v$, thus
\begin{equation} \label{overlap}
\mathcal{C}^{(Z)}_{v} \, \mathcal{C}^{(X)}_{v} = \underset{(v,N(v)\cap V)}{\mathcal{X}^{(Z)}_{\text{o},\text{L}}} \underset{(v,N(v)\cap V)}{\mathcal{X}^{(X)}_{\text{o},\text{L}}}.
\end{equation}
We denote the vertices of the circular induced subgraph in some order along the cycle $V = \{ v_{1}, v_{2},\dots, v_{\vert V \vert} \}$ with $(v_{i},v_{i \pm 1}) \in E$ and $v_{\vert V \vert +1 } = v_{1}$, and in general $v_{i} = v_{i \% \vert V \vert}$
We get the following chain of relations
\begin{align} \label{ali:circleantieval-one}
\vert \Psi \rangle &= (-1) \, \mathcal{C}_{V} \prod^{\vert V \vert }_{k = 1} \zeta_{v_{2k}} \prod^{\vert V \vert }_{k = 1} \mathcal{C}^Z_{v_{2k}} \mathcal{C}^X_{v_{2k}} \vert \Psi \rangle \\ \nonumber
&= (-1) \left( \bigotimes^{\vert V \vert}_{k=1} \mathcal{X}_{v_k} \underset{(v_{k},v_{k+1})}{\mathcal{X}^{(X)}_{\text{o},\text{L}}} \, \underset{(v_{k},v_{k+1})}{\mathcal{X}^{(X)}_{\text{e},\text{L}}} \underset{(v_{k},v_{k+1})}{\mathcal{X}^{(X)}_{\text{o},\text{R}}} \, \underset{(v_{k},v_{k+1})}{\mathcal{X}^{(X)}_{\text{e},\text{R}}} \mathcal{S}_{v_k} \right)\\ \nonumber
&\phantom{= (-1)}\cdot \left( \bigotimes^{\vert V \vert }_{k = 1} \mathcal{Z}_{v_{2k-1}} \mathcal{X}_{v_{2k}} \mathcal{Z}_{v_{2k+1}} \underset{(v_{2k-1},v_{2k})}{\mathcal{X}^{(Z)}_{\text{o},\text{L}}} \underset{(v_{2k-1},v_{2k})}{\mathcal{X}^{(X)}_{\text{o},\text{R}}} \underset{(v_{2k},v_{2k+1})}{\mathcal{X}^{(X)}_{\text{e},\text{L}}} \underset{(v_{2k},v_{2k+1})}{\mathcal{X}^{(Z)}_{\text{e},\text{R}}} S_{v_{2k}} \right)\\ \label{ali:circleantieval-two}
&\phantom{= (-1)}\cdot \left( \bigotimes^{\vert V \vert}_{k=1} \underset{(v_{2k-1},v_{2k})}{\mathcal{X}^{(Z)}_{\text{e},\text{R}}} \underset{(v_{2k-1},v_{2k})}{\mathcal{X}^{(X)}_{\text{e},\text{R}}} \underset{(v_{2k},v_{2k+1})}{\mathcal{X}^{(Z)}_{\text{o},\text{L}}} \underset{(v_{2k},v_{2k+1})}{\mathcal{X}^{(X)}_{\text{o},\text{L}}} \right) \vert \Psi \rangle \\ \nonumber
&= (-1) \left( \bigotimes^{\vert V \vert}_{k=1} \mathcal{X}_{v_k} \, \mathcal{S}_{v_k} \right) \left( \bigotimes^{\vert V \vert }_{k = 1} \mathcal{Z}_{v_{2k-1}} \mathcal{X}_{v_{2k}} \mathcal{Z}_{v_{2k+1}} \, S_{v_{2k}} \right) \left( \bigotimes^{\vert V \vert}_{k = 1} \underset{(v_{k},v_{k+1})}{\mathcal{X}^{(X)}_{\text{o},\text{L}}} \, \underset{(v_{k},v_{k+1})}{\mathcal{X}^{(X)}_{\text{e},\text{L}}} \underset{(v_{k},v_{k+1})}{\mathcal{X}^{(X)}_{\text{o},\text{R}}} \, \underset{(v_{k},v_{k+1})}{\mathcal{X}^{(X)}_{\text{e},\text{R}}} \right)\\ \label{ali:circleantieval-three}
&\phantom{= (-1)}\cdot \left( \bigotimes^{\vert V \vert }_{k = 1} \underset{(v_{k},v_{k+1})}{\mathcal{X}^{(Z)}_{\text{o},\text{L}}} \underset{(v_{k},v_{k+1})}{\mathcal{X}^{(X)}_{\text{o},\text{R}}} \underset{(v_{k},v_{k+1})}{\mathcal{X}^{(X)}_{\text{e},\text{L}}} \underset{(v_{k},v_{k+1})}{\mathcal{X}^{(Z)}_{\text{e},\text{R}}} \right)\left( \bigotimes^{\vert V \vert}_{k=1} \underset{(v_{k},v_{k+1})}{\mathcal{X}^{(Z)}_{\text{e},\text{R}}} \underset{(v_{k},v_{k+1})}{\mathcal{X}^{(X)}_{\text{e},\text{R}}} \underset{(v_{k},v_{k+1})}{\mathcal{X}^{(Z)}_{\text{o},\text{L}}} \underset{(v_{k},v_{k+1})}{\mathcal{X}^{(X)}_{\text{o},\text{L}}} \right) \vert \Psi \rangle \\
\label{ali:circleantieval-four}&= (-1) \left( \bigotimes^{\vert V \vert }_{k = 1} \mathcal{X}_{v_k} \right) \left( \bigotimes^{\vert V \vert }_{k = 1} \mathcal{Z}_{v_{2k-1}} \mathcal{X}_{v_{2k}} \mathcal{Z}_{v_{2k+1}} \right) \vert \Psi \rangle\\
\label{ali:circleantieval-seven}&= (-1) \mathcal{X}_{v_{1}} \mathcal{Z}_{v_{1}} \mathcal{X}_{v_{1}} \mathcal{Z}_{v_{1}} \vert \Psi \rangle \,.
\end{align}
Line~\eqref{ali:circleantieval-one} was obtained by using Eqs.\,\eqref{pe1.5} -- \eqref{pe1.8} for all $v\in V_{c}$. The following three lines, constituting Eq.\;\eqref{ali:circleantieval-two} are obtained by using Eqs.\,\eqref{eq:pe-infstabpauli},~\eqref{eq:pe-measref12} -- \eqref{eq:pe2-measref13}, and Eq.\;\eqref{overlap}. Line~\eqref{ali:circleantieval-three} was obtained by simply rearranging the terms in the previous equation, taking care about commutation relations. Taking that the physical measurement observables square to the identity gives line~\eqref{ali:circleantieval-four}.
There, we are only left with observables on the power vertices.
We can then use the same arguments present in~\cite{mckague2011self}, precisely Eq.\;(15), on the non-inflated circle to arrive at Eq.\;\eqref{ali:circleantieval-seven}.
We obtained the anticommutation relation between observables on the power vertex $v_{1}$.

A simple visualization of the calculus a circle with three vertices $v_{1},v_{2},v_{3}$ is an array whose rows contain the measurements and the deterministic outcome $\chi = \pm 1$ from $ \mathcal{M} \vert \Psi\rangle = \chi \vert \Psi\rangle$.
Multiplying the local observables column-wise, in reverse order of their application to the state, yields their action on the state.
The Pauli measurements are products of all observables shown (in black and light-blue font).
The stabilizer correlations are the products of the observables in black font with their constraints from the simulation conditions on the left.
The observables in light-blue font are not considered in the calculus of the anticommutation in Eqs.\,\eqref{ali:circleantieval-one} -- \eqref{ali:circleantieval-seven}, but they are considered in the round of classical communication and thus influence the observables within communication distance.
\begin{center}
\resizebox{0.8\textwidth}{!}{%
$\hspace{-0.5cm}\begin{array}{r|clllllllllllllll}
& & \multicolumn{2}{c|}{$($v_{3}$,$v_{1}$)$} & \multicolumn{1}{c|}{v_{1}} & \multicolumn{4}{c|}{$($v_{1}$,$v_{2}$)$} & \multicolumn{1}{c|}{v_{2}} & \multicolumn{4}{c|}{$($v_{2}$,$v_{3}$)$} & \multicolumn{1}{c|}{v_{3}} & \multicolumn{2}{c}{$($v_{3}$,$v_{1}$)$} \\[0.05cm]
\mathcal{M}\,,\,\phantom{-}\chi \,& & \multicolumn{1}{c|}{\text{o},\text{R}} & \multicolumn{1}{c|}{\text{e},\text{R}} & \multicolumn{1}{c|}{} & \multicolumn{1}{c|}{\text{o},\text{L}} & \multicolumn{1}{c|}{\text{e},\text{L}} & \multicolumn{1}{c|}{\text{o},\text{R}} & \multicolumn{1}{c|}{\text{e},\text{R}} & \multicolumn{1}{c|}{} & \multicolumn{1}{c|}{\text{o},\text{L}} & \multicolumn{1}{c|}{\text{e},\text{L}} & \multicolumn{1}{c|}{\text{o},\text{R}} & \multicolumn{1}{c|}{\text{e},\text{R}} & \multicolumn{1}{c|}{} & \multicolumn{1}{c|}{\text{o},\text{L}} & \multicolumn{1}{c}{\text{e},\text{L}} \\
& & & & & & & & & & & & & & & \\
\mathcal{C}_{V_{c}}\,,\,-1& & \mathcal{X}^{(X)} & \mathcal{X}^{(X)} & \,\mathcal{X} \, & \mathcal{X}^{(X)} & \mathcal{X}^{(X)} & \mathcal{X}^{(X)} & \mathcal{X}^{(X)} & \,\mathcal{X} \, & \mathcal{X}^{(X)} & \mathcal{X}^{(X)} & \mathcal{X}^{(X)} & \mathcal{X}^{(X)} & \,\mathcal{X} \, & \mathcal{X}^{(X)} & \mathcal{X}^{(X)} \\
\zeta_{v_{2}}\,,\,\phantom{-}1& & \textcolor{myCyan}{\mathcal{X}^{(Z)}} & \textcolor{myCyan}{\mathcal{X}^{(Z)}} & \,\mathcal{Z} \, & \mathcal{X}^{(Z)} & \textcolor{myCyan}{\mathcal{X}^{(Z)}} & \mathcal{X}^{(X)} & \textcolor{myCyan}{\mathcal{X}^{(X)}} & \,\mathcal{X} \, & \textcolor{myCyan}{\mathcal{X}^{(X)}} & \mathcal{X}^{(X)} & \textcolor{myCyan}{\mathcal{X}^{(Z)}} & \mathcal{X}^{(Z)} & \,\mathcal{Z} \, & \textcolor{myCyan}{\mathcal{X}^{(Z)}} & \textcolor{myCyan}{\mathcal{X}^{(Z)}}\\
\zeta_{v_{1}}\,,\,\phantom{-}1& & \mathcal{X}^{(X)} & \textcolor{myCyan}{\mathcal{X}^{(X)}} & \,\mathcal{X} \, & \textcolor{myCyan}{\mathcal{X}^{(X)}} & \mathcal{X}^{(X)} & \textcolor{myCyan}{\mathcal{X}^{(Z)}} & \mathcal{X}^{(Z)} & \,\mathcal{Z} \, & \textcolor{myCyan}{\mathcal{X}^{(Z)}} & \textcolor{myCyan}{\mathcal{X}^{(Z)}} & \textcolor{myCyan}{\mathcal{X}^{(Z)}} & \textcolor{myCyan}{\mathcal{X}^{(Z)}} & \,\mathcal{Z} \, & \mathcal{X}^{(Z)} & \textcolor{myCyan}{\mathcal{X}^{(Z)}} \\
\zeta_{v_{3}}\,,\,\phantom{-}1& & \textcolor{myCyan}{\mathcal{X}^{(Z)}} & \mathcal{X}^{(Z)} & \,\mathcal{Z} \, & \textcolor{myCyan}{\mathcal{X}^{(Z)}} & \textcolor{myCyan}{\mathcal{X}^{(Z)}} & \textcolor{myCyan}{\mathcal{X}^{(Z)}} & \textcolor{myCyan}{\mathcal{X}^{(Z)}} & \,\mathcal{Z} \, & \mathcal{X}^{(Z)} & \textcolor{myCyan}{\mathcal{X}^{(Z)}} & \mathcal{X}^{(X)} & \textcolor{myCyan}{\mathcal{X}^{(X)}} & \,\mathcal{X} \, & \textcolor{myCyan}{\mathcal{X}^{(X)}} & \mathcal{X}^{(X)}\\[0.1cm] \cdashline{1-17}
& & & & & & & & & & & & & & & & \\[-0.25cm]
\mathcal{C}^{(Z)}_{v_{2}}\,,\,\phantom{-}1& & \mathcal{X}^{(Y)} & \mathcal{X}^{(Y)} & \,\mathcal{Y} \, & \textcolor{myCyan}{\mathcal{X}^{(Y)}} & \mathcal{X}^{(Y)} & \textcolor{myCyan}{\mathcal{X}^{(Z)}} & \mathcal{X}^{(Z)} & \,\textcolor{myCyan}{\mathcal{Z}} \, & \mathcal{X}^{(Z)} & \textcolor{myCyan}{\mathcal{X}^{(Z)}} & \mathcal{X}^{(Y)} & \textcolor{myCyan}{\mathcal{X}^{(Y)}} & \,\mathcal{Y} \, & \mathcal{X}^{(Y)} & \mathcal{X}^{(Y)} \\
\mathcal{C}^{(X)}_{v_{2}}\,,\,\phantom{-}1& & \mathcal{X}^{(Y)} & \mathcal{X}^{(Y)} & \,\mathcal{Y} \, & \textcolor{myCyan}{\mathcal{X}^{(Y)}} & \mathcal{X}^{(Y)} & \textcolor{myCyan}{\mathcal{X}^{(X)}} & \mathcal{X}^{(X)} & \,\textcolor{myCyan}{\mathcal{X}} \, & \mathcal{X}^{(X)} & \textcolor{myCyan}{\mathcal{X}^{(X)}} & \mathcal{X}^{(Y)} & \textcolor{myCyan}{\mathcal{X}^{(Y)}} & \,\mathcal{Y} \, & \mathcal{X}^{(Y)} & \mathcal{X}^{(Y)}\\
\mathcal{C}^{(Z)}_{v_{1}}\,,\,\phantom{-}1& & \textcolor{myCyan}{\mathcal{X}^{(Z)}} & \mathcal{X}^{(Z)} & \,\textcolor{myCyan}{\mathcal{Z}} \, & \mathcal{X}^{(Z)} & \textcolor{myCyan}{\mathcal{X}^{(Z)}} & \mathcal{X}^{(Y)} & \textcolor{myCyan}{\mathcal{X}^{(Y)}} & \,\mathcal{Y} \, & \mathcal{X}^{(Y)} & \mathcal{X}^{(Y)} & \mathcal{X}^{(Y)} & \mathcal{X}^{(Y)} & \,\mathcal{Y} \, & \textcolor{myCyan}{\mathcal{X}^{(Y)}} & \mathcal{X}^{(Y)}\\
\mathcal{C}^{(X)}_{v_{1}}\,,\,\phantom{-}1& & \textcolor{myCyan}{\mathcal{X}^{(X)}} & \mathcal{X}^{(X)} & \,\textcolor{myCyan}{\mathcal{X}} \, & \mathcal{X}^{(X)} & \textcolor{myCyan}{\mathcal{X}^{(X)}} & \mathcal{X}^{(Y)} & \textcolor{myCyan}{\mathcal{X}^{(Y)}} & \,\mathcal{Y} \, & \mathcal{X}^{(Y)} & \mathcal{X}^{(Y)} & \mathcal{X}^{(Y)} & \mathcal{X}^{(Y)} & \,\mathcal{Y} \, & \textcolor{myCyan}{\mathcal{X}^{(Y)}} & \mathcal{X}^{(Y)} \\
\mathcal{C}^{(Z)}_{v_{3}}\,,\,\phantom{-}1& & \mathcal{X}^{(Y)} & \textcolor{myCyan}{\mathcal{X}^{(Y)}} & \,\mathcal{Y} \, & \mathcal{X}^{(Y)} & \mathcal{X}^{(Y)} & \mathcal{X}^{(Y)} & \mathcal{X}^{(Y)} & \,\mathcal{Y} \, & \textcolor{myCyan}{\mathcal{X}^{(Y)}} & \mathcal{X}^{(Y)} & \textcolor{myCyan}{\mathcal{X}^{(Z)}} & \mathcal{X}^{(Z)} & \,\textcolor{myCyan}{\mathcal{Z}} \, & \mathcal{X}^{(Z)} & \textcolor{myCyan}{\mathcal{X}^{(Z)}} \\
\mathcal{C}^{(X)}_{v_{3}}\,,\,\phantom{-}1& & \mathcal{X}^{(Y)} & \textcolor{myCyan}{\mathcal{X}^{(Y)}} & \,\mathcal{Y} \, & \mathcal{X}^{(Y)} & \mathcal{X}^{(Y)} & \mathcal{X}^{(Y)} & \mathcal{X}^{(Y)} & \,\mathcal{Y} \, & \textcolor{myCyan}{\mathcal{X}^{(Y)}} & \mathcal{X}^{(Y)} & \textcolor{myCyan}{\mathcal{X}^{(X)}} & \mathcal{X}^{(X)} & \,\textcolor{myCyan}{\mathcal{X}} \, & \mathcal{X}^{(X)} & \textcolor{myCyan}{\mathcal{X}^{(X)}} 
\end{array}$
}
\end{center}
Note that, as the odd cycle has translation symmetry, by relabeling the vertices in the proof, we can obtain the anticommutation relations for all power vertices belonging to the odd cycle: For all $k = 1,\dots,\vert V \vert$, we can obtain
\begin{equation}
\{ \mathcal{X}_{v_k}, \mathcal{Z}_{v_k} \}\vert \Psi \rangle=0\,.
\end{equation}
By applying Lemma~\ref{lem:indu}, we obtain the same anticommutation relation for all power vertices in the inflated graph.
This is possible, even if we only consider the anticommutation relation on the power vertex $v_{1}$
\end{proof}

\begin{proof}[Proof [Lemma~\ref{lem:ref1234}, Reference experiment~\ref{def:refe2}]]
If the physical experiment (PE) simulates the reference experiment (RE)~\ref{def:refe2} for $\vert N(v_{c})\vert > 2$, the specified correlations from the measurements fulfill
\begin{align}\label{pe2.1}
\langle \Psi\vert \, \zeta_{u} \, \vert\Psi\rangle &= \phantom{-} 1\,, \qquad \forall u \in V\,,\\ \label{pe2.2}
\langle \Psi\vert \, \Tilde{\zeta}_{v_{c}} \, \vert\Psi\rangle &= \phantom{-} 1\,,\\ \label{pe2.3}
\langle \Psi\vert \, \mathcal{C}_{v_{i}v_{j}} \,\vert\Psi\rangle &= -1\,, \qquad \forall \, v_{i}\neq v_{j} \wedge v_{i},v_{j} \in V_{3} = \{ v_{1},v_{2},v_{3} \} \subseteq N(v_{c}) \,.
\end{align}
Since the norm of all physical measurement observables is at most $1$, Eqs.~\eqref{pe2.1} -- \eqref{pe2.3} imply
\begin{align} \label{pe2.4}
\zeta_{u} \,\vert \Psi \rangle &= \vert\Psi\rangle \,, \qquad \forall u \in V \,, \\ \label{pe2.5}
\Tilde{\zeta}_{v_{c}} \,\vert \Psi \rangle &= \vert\Psi\rangle \,,\\ \label{pe2.6}
(-1)\,\mathcal{C}_{v_{i}v_{j}} \, \vert \Psi \rangle &= \vert\Psi\rangle \,, \qquad \forall \, v_{i}\neq v_{j} \wedge v_{i},v_{j} \in V_{3} = \{ v_{1},v_{2},v_{3} \} \subseteq N(v_{c}) \,.
\end{align}
Recall that any vertex's observable depends on the measurement observables from all vertices $v \in N_d(u)$.
As explained in the main text, whenever a vertex is not asked to measure anything in the RE, in the PE it is given an input, but the results are marginalized.

Then, we get the following chain of relations
\begin{align} \label{ali:starantieval-one}
\vert \Psi \rangle &= (-1) \Tilde{\zeta}_{v_{c}} \mathcal{C}_{v_{1}v_{2}} \mathcal{C}_{v_{2}v_{3}} \mathcal{C}_{v_{1}v_{3}} \vert \Psi \rangle\\ \nonumber
&= (-1) \,\Tilde{\mathcal{X}}_{v_{c}} \, \mathcal{Z}_{v_{1}} \mathcal{Z}_{v_{2}} \mathcal{Z}_{v_{3}} \mathcal{Z}^{N(v_{c})\setminus V_{3}} \, \phantom{\underset{ (v_{c},v_{1}) }{ \Tilde{\mathcal{X}}^{(X)}_{\text{o,L}} \Tilde{\mathcal{X}}^{(X)}_{\text{o,R}} } \, \phantom{\underset{ (v_{c},v_{2}) }{ \Tilde{\mathcal{X}}^{(X)}_{\text{o,L}} \Tilde{\mathcal{X}}^{(X)}_{\text{o,R}} }}\,\underset{ (v_{c},v_{3}) }{ \Tilde{\mathcal{X}}^{(X)}_{\text{o,L}} \Tilde{\mathcal{X}}^{(X)}_{\text{o,R}} } } ~~~~~~ \underset{(v_{c},N(v_{c}))}{\Tilde{\mathcal{X}}^{(X)}_{\text{e},L}\Tilde{\mathcal{X}}^{(Z)}_{\text{e},R}} \\ \nonumber
&\phantom{= (-1)}\cdot\Tilde{\mathcal{X}}_{v_{c}}\, \mathcal{X}_{v_{1}} \mathcal{X}_{v_{2}} \mathcal{Z}_{v_{2}} \,\mathcal{Z}^{N(v_{c}) \setminus V_{3} } \, \underset{ (v_{c},v_{1}) }{ \Tilde{\mathcal{X}}^{(X)}_{\text{o,L}} \Tilde{\mathcal{X}}^{(X)}_{\text{o,R}} } \,\underset{ (v_{c},v_{2}) }{ \Tilde{\mathcal{X}}^{(X)}_{\text{o,L}} \Tilde{\mathcal{X}}^{(X)}_{\text{o,R}} } \, \phantom{\underset{ (v_{c},v_{3}) }{ \Tilde{\mathcal{X}}^{(X)}_{\text{o,L}} \Tilde{\mathcal{X}}^{(X)}_{\text{o,R}} }}\, \underset{ (v_{c},N(v_{c})\setminus \{v_{1},v_{2}\}) }{ \Tilde{\mathcal{X}}^{(X)}_{\text{e,L}}\Tilde{\mathcal{X}}^{(Z)}_{\text{e,R}}} \,\mathcal{S}_{v_{1}} \mathcal{S}_{v_{2}} \\ \label{ali:starantieval-two}
&\phantom{= (-1)}\cdot\Tilde{\mathcal{X}}_{v_{c}}\,\mathcal{Z}_{v_{1}} \mathcal{X}_{v_{2}} \mathcal{X}_{v_{3}} \,\mathcal{Z}^{N(v_{c}) \setminus V_{3}} \, \phantom{\underset{ (v_{c},v_{1}) }{ \Tilde{\mathcal{X}}^{(X)}_{\text{o,L}} \Tilde{\mathcal{X}}^{(X)}_{\text{o,R}} }}\, \underset{ (v_{c},v_{2}) }{ \Tilde{\mathcal{X}}^{(X)}_{\text{o,L}} \Tilde{\mathcal{X}}^{(X)}_{\text{o,R}} } \,\underset{ (v_{c},v_{3}) }{ \Tilde{\mathcal{X}}^{(X)}_{\text{o,L}} \Tilde{\mathcal{X}}^{(X)}_{\text{o,R}} } \, \underset{ (v_{c},N(v_{c})\setminus \{v_{2},v_{3}\}) }{ \Tilde{\mathcal{X}}^{(X)}_{\text{e,L}}\Tilde{\mathcal{X}}^{(Z)}_{\text{e,R}}} \,\mathcal{S}_{v_{2}} \mathcal{S}_{v_{3}} \\ \label{ali:starantieval-three}
&\phantom{= (-1)}\cdot \Tilde{\mathcal{X}}_{v_{c}}\, \mathcal{X}_{v_{1}} \mathcal{Z}_{v_{2}} \mathcal{X}_{v_{3}} \,\mathcal{Z}^{N(v_{c}) \setminus V_{3} } \, \underset{ (v_{c},v_{1}) }{ \Tilde{\mathcal{X}}^{(X)}_{\text{o,L}} \Tilde{\mathcal{X}}^{(X)}_{\text{o,R}} } \, \phantom{\underset{ (v_{c},v_{2}) }{ \Tilde{\mathcal{X}}^{(X)}_{\text{o,L}} \Tilde{\mathcal{X}}^{(X)}_{\text{o,R}} }}\,\underset{ (v_{c},v_{3}) }{ \Tilde{\mathcal{X}}^{(X)}_{\text{o,L}} \Tilde{\mathcal{X}}^{(X)}_{\text{o,R}} } \, \underset{ (v_{c},N(v_{c})\setminus \{v_{1},v_{3}\}) }{ \Tilde{\mathcal{X}}^{(X)}_{\text{e,L}}\Tilde{\mathcal{X}}^{(Z)}_{\text{e,R}}} \,\mathcal{S}_{v_{1}} \mathcal{S}_{v_{3}} \vert \Psi \rangle\\ 
&= (-1 ) \mathcal{Z}_{v_{1}} \mathcal{X}_{v_{1}} \mathcal{Z}_{v_{1}} \mathcal{X}_{v_{1}} \vert \Psi \rangle \,.
\end{align}
Line~\eqref{ali:starantieval-one} was obtained by using Eqs.\,\eqref{pe2.4} -- \eqref{pe1.6}.
The following four lines, constituting Eq.\;\eqref{ali:starantieval-two} are obtained by using Eq.\;\eqref{eq:pe-refe25} and Eq.\;\eqref{eq:pe-refe21}.
Line~\eqref{ali:starantieval-three} was obtained by using that the physical measurement observables square to the identity, so that we are only left with observables on the power vertex $v_{1}$.
We obtained the anticommutation relation between observables on the power vertex $v_{1}$.

A simple visualization of the calculation for $d=2$ can be depicted as an array, where each row contains the measurements $\mathcal{M}$ with submeasurements $\mathcal{C}$, along with the deterministic outcome $\chi = \pm 1$ from $ \mathcal{C} \vert \Psi\rangle = \chi \vert \Psi\rangle$:
\begin{center}
\resizebox{0.99\textwidth}{!}{%
$\begin{array}{ll|lllllllllllllllllllll}
\mathcal{C}(\textcolor{myCyan}{\mathcal{M}}) & \phantom{-} \chi & & v_{c} & & \hspace{-0.1cm}1_{(v_{c},v_{1})} & \hspace{-0.1cm}2_{(v_{c},v_{1})} & \hspace{-0.1cm}3_{(v_{c},v_{1})} & \hspace{-0.1cm}4_{(v_{c},v_{1})} & v_{1} & & \hspace{-0.1cm} 1_{(v_{c},v_{2})} & \hspace{-0.1cm}2_{(v_{c},v_{2})} & \hspace{-0.1cm}3_{(v_{c},v_{2})} & \hspace{-0.1cm}4_{(v_{c},v_{2})} & v_{2} & & \hspace{-0.1cm}1_{(v_{c},v_{3})} & \hspace{-0.1cm}2_{(v_{c},v_{3})} & \hspace{-0.1cm}3_{(v_{c},v_{3})} & \hspace{-0.1cm}4_{(v_{c},v_{3})} & v_{3}\\[0.25cm]
\Tilde{\zeta}_{v_{c}}\,,& \phantom{-}1 & & \Tilde{\mathcal{X}} & & \textcolor{myCyan}{\mathcal{Y}^{(X)}} & \Tilde{\mathcal{X}}^{(X)} & \textcolor{myCyan}{\Tilde{\mathcal{X}}^{(Z)}} & \Tilde{\mathcal{X}}^{(Z)} &\, \mathcal{Z} & & \textcolor{myCyan}{\mathcal{Y}^{(X)}} & \Tilde{\mathcal{X}}^{(X)} & \textcolor{myCyan}{\Tilde{\mathcal{X}}^{(Z)}} & \Tilde{\mathcal{X}}^{(Z)} &\, \mathcal{Z} & & \textcolor{myCyan}{\mathcal{Y}^{(X)}} & \Tilde{\mathcal{X}}^{(X)} & \textcolor{myCyan}{\Tilde{\mathcal{X}}^{(Z)}} & \Tilde{\mathcal{X}}^{(Z)} &\, \mathcal{Z} \\
\mathcal{C}_{v_{1},v_{2}}\,,& - 1& & \Tilde{\mathcal{X}} & & \mathcal{Y}^{(X)} & \textcolor{myCyan}{\Tilde{\mathcal{X}}^{(X)}} & \Tilde{\mathcal{X}}^{(X)} & \textcolor{myCyan}{\Tilde{\mathcal{X}}^{(X)}} &\, \mathcal{X} & & \mathcal{Y}^{(X)} & \textcolor{myCyan}{\Tilde{\mathcal{X}}^{(X)}} & \Tilde{\mathcal{X}}^{(X)} & \textcolor{myCyan}{\Tilde{\mathcal{X}}^{(X)}} &\, \mathcal{X} & & \textcolor{myCyan}{\mathcal{Y}^{(X)}} & \Tilde{\mathcal{X}}^{(X)} & \textcolor{myCyan}{\Tilde{\mathcal{X}}^{(Z)}} & \Tilde{\mathcal{X}}^{(Z)} & \, \mathcal{Z} \\
\mathcal{C}_{v_{2},v_{3}}\,,& - 1 & & \Tilde{\mathcal{X}} & & \textcolor{myCyan}{\mathcal{Y}^{(X)}} & \Tilde{\mathcal{X}}^{(X)} & \textcolor{myCyan}{\Tilde{\mathcal{X}}^{(Z)}} & \Tilde{\mathcal{X}}^{(Z)} & \, \mathcal{Z} & & \mathcal{Y}^{(X)} & \textcolor{myCyan}{\Tilde{\mathcal{X}}^{(X)}} & \Tilde{\mathcal{X}}^{(X)} & \textcolor{myCyan}{\Tilde{\mathcal{X}}^{(X)}} &\, \mathcal{X} & & \mathcal{Y}^{(X)} & \textcolor{myCyan}{\Tilde{\mathcal{X}}^{(X)}} & \Tilde{\mathcal{X}}^{(X)} & \textcolor{myCyan}{\Tilde{\mathcal{X}}^{(X)}} &\, \mathcal{X} \\
\mathcal{C}_{v_{1},v_{3}}\,,& - 1 & & \Tilde{\mathcal{X}} & & \mathcal{Y}^{(X)} & \textcolor{myCyan}{\Tilde{\mathcal{X}}^{(X)}} & \Tilde{\mathcal{X}}^{(X)} & \textcolor{myCyan}{\Tilde{\mathcal{X}}^{(X)}} &\, \mathcal{X} & & \textcolor{myCyan}{\mathcal{Y}^{(X)}} & \Tilde{\mathcal{X}}^{(X)} & \textcolor{myCyan}{\Tilde{\mathcal{X}}^{(Z)}} & \Tilde{\mathcal{X}}^{(Z)} & \, \mathcal{Z} & & \mathcal{Y}^{(X)} & \textcolor{myCyan}{\Tilde{\mathcal{X}}^{(X)}} & \Tilde{\mathcal{X}}^{(X)} & \textcolor{myCyan}{\Tilde{\mathcal{X}}^{(X)}} &\, \mathcal{X} \\[0.1cm] \cdashline{1-22}
& & & & & & & & & & & & & & & & & & & & \\[-0.25cm]
\zeta_{v_{c}}\,,& \phantom{-}1 & & \mathcal{X} & & \textcolor{myCyan}{\mathcal{X}^{(X)}} & \mathcal{X}^{(X)} & \textcolor{myCyan}{\mathcal{X}^{(Z)}} & \mathcal{X}^{(Z)} &\, \mathcal{Z} & & \textcolor{myCyan}{\mathcal{X}^{(X)}} & \mathcal{X}^{(X)} & \textcolor{myCyan}{\mathcal{X}^{(Z)}} & \mathcal{X}^{(Z)} &\, \mathcal{Z} & & \textcolor{myCyan}{\mathcal{X}^{(X)}} & \mathcal{X}^{(X)} & \textcolor{myCyan}{\mathcal{X}^{(Z)}} & \mathcal{X}^{(Z)} &\, \mathcal{Z} \\
\zeta_{v_{1}}\,,& \phantom{-} 1& & \mathcal{Z} & & \mathcal{X}^{(X)} & \textcolor{myCyan}{\mathcal{X}^{(X)}} & \mathcal{X}^{(X)} & \textcolor{myCyan}{\mathcal{X}^{(X)}} &\, \mathcal{X} & & \textcolor{myCyan}{\mathcal{X}^{(Z)}} & \textcolor{myCyan}{\mathcal{X}^{(Z)}} & \textcolor{myCyan}{\mathcal{X}^{(X)}} & \textcolor{myCyan}{\mathcal{X}^{(X)}} & \, \textcolor{myCyan}{\mathcal{X}} & & \textcolor{myCyan}{\mathcal{X}^{(Z)}} & \textcolor{myCyan}{\mathcal{X}^{(Z)}} & \textcolor{myCyan}{\mathcal{X}^{(X)}} & \textcolor{myCyan}{\mathcal{X}^{(X)}} & \, \textcolor{myCyan}{\mathcal{X}} \\
\zeta_{v_{c}}\,,& \phantom{-}1 & & \mathcal{X} & & \textcolor{myCyan}{\mathcal{X}^{(X)}} & \mathcal{X}^{(X)} & \textcolor{myCyan}{\mathcal{X}^{(Z)}} & \mathcal{X}^{(Z)} &\, \mathcal{Z} & & \textcolor{myCyan}{\mathcal{X}^{(X)}} & \mathcal{X}^{(X)} & \textcolor{myCyan}{\mathcal{X}^{(Z)}} & \mathcal{X}^{(Z)} &\, \mathcal{Z} & & \textcolor{myCyan}{\mathcal{X}^{(X)}} & \mathcal{X}^{(X)} & \textcolor{myCyan}{\mathcal{X}^{(Z)}} & \mathcal{X}^{(Z)} &\, \mathcal{Z} \\
\zeta_{v_{1}}\,,& \phantom{-} 1& & \mathcal{Z} & & \mathcal{X}^{(X)} & \textcolor{myCyan}{\mathcal{X}^{(X)}} & \mathcal{X}^{(X)} & \textcolor{myCyan}{\mathcal{X}^{(X)}} &\, \mathcal{X} & & \textcolor{myCyan}{\mathcal{X}^{(Z)}} & \textcolor{myCyan}{\mathcal{X}^{(Z)}} & \textcolor{myCyan}{\mathcal{X}^{(X)}} & \textcolor{myCyan}{\mathcal{X}^{(X)}} & \, \textcolor{myCyan}{\mathcal{X}} & & \textcolor{myCyan}{\mathcal{X}^{(Z)}} & \textcolor{myCyan}{\mathcal{X}^{(Z)}} & \textcolor{myCyan}{\mathcal{X}^{(X)}} & \textcolor{myCyan}{\mathcal{X}^{(X)}} & \, \textcolor{myCyan}{\mathcal{X}} &
\end{array}$
}
\end{center}
\noindent The Pauli measurements are represented as products of all observables shown (in black and light-blue font), while the stabilizer submeasurements consist of the products of the observables in black font.
Multiplying these column-wise yields their action on the state.
The observables in the light-blue font are not directly involved in the calculation of the anticommutation in Eq.\;\eqref{ali:starantieval-one}, but they impact the observable within communication distance and confound classical models based on communication.
All measurement observables are labeled according to the local measurement setting.
The information gained from classical communication is accounted for in superscripts.
To keep these minimal, we only label the chain vertices according to the measurement setting of its nearest power vertex since all other chain vertices measure in the same basis.

If the physical experiment (PE) simulates the reference experiment (RE)~\ref{def:refe2} for $N(v_{c}) = \{v_{1},v_{2} \} $ such that $\vert N(v_{c}) \vert =2 $, the specified correlations from the measurements fulfill
\begin{align}\label{pe22.1}
\langle \Psi\vert \, \zeta_{u} \, \vert\Psi\rangle &= \phantom{-} 1\,, \qquad \forall u \in V\,,\\ \label{pe22.2}
\langle \Psi\vert \, \Tilde{\zeta}_{v_{c}} \, \vert\Psi\rangle &= \phantom{-} 1\,,\\ \label{pe22.3}
\langle \Psi\vert \, \mathcal{C}_{v_{c}} \,\vert\Psi\rangle &= -1\,,\\ \label{pe22.4}
\langle \Psi\vert \, \mathcal{C}_{v_{1}} \,\vert\Psi\rangle &= \phantom{-}1\,,\\ \label{pe22.5}
\langle \Psi\vert \, \mathcal{C}_{v_{2}} \,\vert\Psi\rangle &= \phantom{-}1\,,\\ \label{pe22.6}
\langle \Psi\vert \, \mathcal{C}^{(X)}_{2} \,\vert\Psi\rangle &= \phantom{-}1\,.\\ \label{pe22.7}
\langle \Psi\vert \, \mathcal{C}^{(Y)}_{2} \,\vert\Psi\rangle &= \phantom{-}1\,.
\end{align}
Since the norm of all physical measurement observables is at most $1$, Eqs.~\eqref{pe22.1} -- \eqref{pe22.7} imply
\begin{align} \label{pe22.8}
\zeta_{u} \, \vert\Psi\rangle &= \vert\Psi\rangle\,, \qquad \forall u \in V\,,\\ \label{pe22.9}
\Tilde{\zeta}_{v_{c}} \, \vert\Psi\rangle &= \vert\Psi\rangle \,,\\ \label{pe22.10}
(-1)\mathcal{C}_{v_{c}} \,\vert\Psi\rangle &= \vert\Psi\rangle \,,\\ \label{pe22.11}
\mathcal{C}_{v_{1}} \,\vert\Psi\rangle &= \vert\Psi\rangle\,,\\ \label{pe22.12}
\mathcal{C}_{v_{2}} \,\vert\Psi\rangle &= \vert\Psi\rangle\,,\\ \label{pe22.13}
\mathcal{C}^{(X)}_{2} \,\vert\Psi\rangle &= \vert\Psi\rangle \,.\\ \label{pe22.14}
\mathcal{C}^{(Y)}_{2} \,\vert\Psi\rangle &= \vert\Psi\rangle \,.
\end{align}
Recall that any vertex's observable depends on the measurement observables from all vertices $v \in N_d(u)$.
Whenever a vertex is not asked to measure anything in the RE, in the PE it is given an input, but the results are marginalized.
Thus, physical measurement $\mathcal{C}_{v}$ has two non-equivalent implementations, one in which the vertex which does not measure anything in the RE is asked to perform the measurement corresponding to $X$ and the other in which it is asked to measure $Z$. We write $\mathcal{C}^{(X)}_{2}$ and $\mathcal{C}^{(Y)}_{2}$ for these two physical implementations, respectively.
Both operators have congruent observables except on the chain vertices with the nearest power vertex $v_{c}$, thus
\begin{equation} \label{pe2-overlap}
\mathcal{C}^{(Y)}_{v_{c}} \, \mathcal{C}^{(X)}_{v_{c}} = \underset{(v_{c},N(v_{c}))}{\Tilde{\mathcal{X}}^{(Y)}_{\text{o},\text{L}} \Tilde{\mathcal{X}}^{(X)}_{\text{e},\text{L}}} \,.
\end{equation}
Then, we get the following chain of relations
\begin{align} \label{ali:star2antieval-one}
\vert \Psi \rangle &= (-1) \mathcal{C}_{v_{c}} \Tilde{\zeta}_{v_{c}} \mathcal{C}_{v_{1}} \mathcal{C}_{v_{2}} \mathcal{C}^{(Y)}_{2} \mathcal{C}^{(X)}_{2} \vert \Psi \rangle \\ \nonumber
&= (-1) ~\Tilde{\mathcal{X}}_{v_{c}}\, \mathcal{X}_{v_{1}} \mathcal{X}_{v_{2}} \, \underset{(v_{c},v_{1})}{ \Tilde{\mathcal{X}}^{(X)}_{\text{o,L}} \Tilde{\mathcal{X}}^{(X)}_{\text{o,R}}}\, \underset{(v_{c},v_{2})}{ \Tilde{\mathcal{X}}^{(X)}_{\text{o,L}} \Tilde{\mathcal{X}}^{(X)}_{\text{o,R}}}\, \mathcal{S}_{v_{1}} \mathcal{S}_{v_{2}}\\ \nonumber 
&\phantom{= (-1)}\cdot\Tilde{\mathcal{X}}_{v_{c}} \, \mathcal{Z}_{v_{1}} \mathcal{Z}_{v_{2}} \underset{(v_{c},v_{1})}{\Tilde{\mathcal{X}}^{(X)}_{\text{e},L}\Tilde{\mathcal{X}}^{(Z)}_{\text{e},R}} \, \underset{(v_{c},v_{2})}{\Tilde{\mathcal{X}}^{(X)}_{\text{e},L}\Tilde{\mathcal{X}}^{(Z)}_{\text{e},R}} \\ \nonumber
&\phantom{= (-1)}\cdot \Tilde{\mathcal{Y}}_{v_{c}}\, \mathcal{X}_{v_{1}} \mathcal{Z}_{v_{2}} \, \underset{ (v_{c},v_{1}) }{ \Tilde{\mathcal{X}}^{(Y)}_{\text{o,L}} \Tilde{\mathcal{X}}^{(X)}_{\text{o,R}} } \, \underset{ (v_{c},v_{2}) }{ \Tilde{\mathcal{X}}^{(Y)}_{\text{e,L}}\Tilde{\mathcal{X}}^{(Z)}_{\text{e,R}}} \,\mathcal{S}_{v_{1}}\\ \nonumber
&\phantom{= (-1)}\cdot \Tilde{\mathcal{Y}}_{v_{c}}\, \mathcal{Z}_{v_{1}} \mathcal{X}_{v_{2}} \,\underset{ (v_{c},v_{1}) }{ \Tilde{\mathcal{X}}^{(Y)}_{\text{e,L}}\Tilde{\mathcal{X}}^{(Z)}_{\text{e,R}}} \, \underset{ (v_{c},v_{2}) }{ \Tilde{\mathcal{X}}^{(Y)}_{\text{o,L}} \Tilde{\mathcal{X}}^{(X)}_{\text{o,R}} } \, \mathcal{S}_{v_{2}}\\ \label{ali:star2antieval-two}
&\phantom{= (-1)}\cdot \phantom{\Tilde{\mathcal{X}}_{v_{c}}\, \mathcal{X}_{v_{1}} \mathcal{X}_{v_{2}}} \, \underset{(v_{c},v_{1})}{\Tilde{\mathcal{X}}^{(Y)}_{\text{o},\text{L}} \Tilde{\mathcal{X}}^{(X)}_{\text{e},\text{L}}} \, \underset{(v_{c},v_{2})}{\Tilde{\mathcal{X}}^{(Y)}_{\text{o},\text{L}} \Tilde{\mathcal{X}}^{(X)}_{\text{e},\text{L}}}\\ \label{ali:star2antieval-three}
&=(-1 ) \mathcal{Z}_{v_{1}} \mathcal{X}_{v_{1}} \mathcal{Z}_{v_{1}} \mathcal{X}_{v_{1}} \vert \Psi \rangle \,.
\end{align}
Line~\eqref{ali:star2antieval-one} was obtained by using Eqs.\,\eqref{pe22.8} -- \eqref{pe22.14}. The following five lines, constituting Eq.\;\eqref{ali:star2antieval-two} are obtained by using Eqs.\,\eqref{eq:pe-infstabpauli},~\eqref{eq:pe-refe23} -- \eqref{eq:pe-refe24-2}, and Eq.\;\eqref{pe2-overlap}.
Taking that the physical measurement observables square to the identity gives line~\eqref{ali:star2antieval-three}.
We obtained the anticommutation relation between observables on the power vertex $v_{1}$.

The same anticommutation relation for all other power vertices follows from Lemma~\ref{lem:indu}.
\end{proof}

\begin{proof}[Proof [Lemma~\ref{lem:ref1234}, Reference experiment~\ref{def:refe3}]]
If the physical experiment (PE) simulates the reference experiment (RE)~\ref{def:refe3}, the specified correlations from the measurements fulfill
\begin{align}\label{pe3.1}
\langle \Psi\vert \, \zeta_{u} \, \vert\Psi\rangle &= 1\,, \qquad \qquad \qquad \forall u \in V\,,\\ \label{pe3.2}
\langle \Psi\vert \, \mathcal{I}_{1} \, \vert\Psi\rangle &= (-1)^{d+1}\frac{1}{\sqrt{2}}\,,\\ \label{pe3.3}
\langle \Psi\vert \, \mathcal{I}_{2} \,\vert\Psi\rangle &= (-1)^{d}\frac{1}{\sqrt{2}}\,,\\ \label{pe3.4}
\langle \Psi\vert \,\mathcal{I}_{3} \,\vert\Psi\rangle &= \frac{1}{\sqrt{2}}\,,\\ \label{pe3.5}
\langle \Psi\vert \, \mathcal{I}_{4} \,\vert\Psi\rangle &= \frac{1}{\sqrt{2}}\,,\\ \label{pe3.6}
\langle \Psi\vert \, \mathcal{P}_{1} \,\vert\Psi\rangle &= 1\,.\\ \label{pe3.7}
\langle \Psi\vert \, \mathcal{P}_{2} \,\vert\Psi\rangle &= 1\,.
\end{align}
Then, it follows from Eqs.\,\eqref{pe3.2} -- \eqref{pe3.5} that
\begin{equation}
\langle \Psi \vert \left(2\sqrt{2} \mathds{1} + (-1)^d(\mathcal{I}_{1} - \mathcal{I}_{2}) - \mathcal{I}_{3} - \mathcal{I}_{4} \right) \vert \Psi \rangle = 0 \,.\label{eq:appsos1}
\end{equation}
The operator has at least three decompositions into sums of squares
\begin{align}
4 \mathds{1} + \sqrt{2}\left( (-1)^d(\mathcal{I}_{1} - \mathcal{I}_{2}) - \mathcal{I}_{3} - \mathcal{I}_{4} \right) 
&= \left[ \sqrt{2} \mathds{1} + \frac{1}{2}\left( (-1)^d(\mathcal{I}_{1} - \mathcal{I}_{2}) - \mathcal{I}_{3} - \mathcal{I}_{4} \right) \right]^2 \nonumber\\
&\hspace{1.5cm}+ \left[ \frac{1}{2}\left( (-1)^d(\mathcal{I}_{1} + \mathcal{I}_{2}) - \mathcal{I}_{3} + \mathcal{I}_{4} \right) \right]^2 \nonumber \\
&= \left[ \mathcal{A}_{0} - \frac{\mathcal{B}_{0} - (-1)^d \mathcal{B}_{1}}{\sqrt{2}} \right]^2 + \left[ \mathcal{A}_{1} - \frac{ \mathcal{B}_{1} + (-1)^d \mathcal{B}_{0} }{\sqrt{2}} \right]^2 \label{ali:sosanti1}\\
&= \left[ \mathcal{B}_{0} - \frac{\mathcal{A}_{0} - (-1)^d \mathcal{A}_{1}}{\sqrt{2}} \right]^2 + \left[ \mathcal{B}_{1} -\frac{ \mathcal{A}_{1} + (-1)^d \mathcal{A}_{0} }{\sqrt{2}} \right]^2 \label{ali:sosanti2} \,.
\end{align}
with $\mathcal{I}_{1} = \mathcal{A}_{0}\, \mathcal{B}_{1}$, $\mathcal{I}_{2} = \mathcal{A}_{1} \,\mathcal{B}_{0}$, $\mathcal{I}_{3} = \mathcal{A}_{1} \,\mathcal{B}_{1}$, $\mathcal{I}_{4} = \mathcal{A}_{0} \,\mathcal{B}_{0}$, and, using Eqs.\,\eqref{ali:pcorr41-pe} -- \eqref{ali:pcorr44-pe},
\begin{equation*}
\mathcal{A}_{0} := \Tilde{\mathcal{Z}}^{(R_{z})}_{v_{l}} \underset{(v_{l},v_{r})}{\Tilde{\mathcal{X}}^{(Z,R_{z})}_{\text{o},\text{L}}} \mathcal{R}^{(Z)}_{d_{(v_{l},v_{r})}} \,, \hspace{0.5cm} \mathcal{A}_{1} := S_{v_{l}} \Tilde{\mathcal{O}}^{(R_{z})}_{v_{l}} \underset{(v_{l},v_{r})}{\Tilde{\mathcal{X}}^{(O,R_{z})}_{\text{e},\text{L}}} \mathcal{R}^{(O)}_{d_{(v_{l},v_{r})}} \,, \hspace{0.5cm} \mathcal{B}_{0} := \underset{(v_{l},v_{r})}{\Tilde{\mathcal{X}}^{(Z,R_{z})}_{\text{e},\text{R}}} \Tilde{\mathcal{Z}}_{v_{r}} \,, \hspace{0.5cm} \mathcal{B}_{1} := \underset{(v_{l},v_{r})}{\Tilde{\mathcal{X}}^{(X,R_{z})}_{\text{o},\text{R}}} \Tilde{\mathcal{X}}_{v_{r}} \, \mathcal{S}_{v_{r}} \,.
\end{equation*}
The first decomposition holds since $\mathcal{I}_{1} \mathcal{I}_{4} = \mathcal{I}_{3} \mathcal{I}_{2} = \mathcal{B}_{1} \mathcal{B}_{0}$ and $\mathcal{I}_{4} \mathcal{I}_{1} = \mathcal{I}_{2} \mathcal{I}_{3} = \mathcal{B}_{0} \mathcal{B}_{1}$.
From Eq.\;\eqref{eq:appsos1} it is
\begin{align*}
2\sqrt{2} \mathds{1} + (-1)^d(\mathcal{I}_{1} - \mathcal{I}_{2}) - \mathcal{I}_{3} - \mathcal{I}_{4} \vert \Psi \rangle &= 0 \,, \hspace{1cm} (-1)^d(\mathcal{I}_{1} + \mathcal{I}_{2}) - \mathcal{I}_{3} + \mathcal{I}_{4} \vert \Psi \rangle = 0 \,, \\[0.2cm]
\left[ \mathcal{A}_{0} - \frac{\mathcal{B}_{0} - (-1)^d \mathcal{B}_{1}}{\sqrt{2}} \right] \vert \Psi \rangle &= 0 \,, \hspace{1cm}\left[ \mathcal{A}_{1} - \frac{ \mathcal{B}_{1} + (-1)^d \mathcal{B}_{0} }{\sqrt{2}} \right] \vert \Psi \rangle = 0 \,,\\
\left[ \mathcal{B}_{0} - \frac{\mathcal{A}_{0} - (-1)^d \mathcal{A}_{1}}{\sqrt{2}} \right] \vert \Psi \rangle &= 0 \,, \hspace{1cm} \left[ \mathcal{B}_{1} -\frac{ \mathcal{A}_{1} + (-1)^d \mathcal{A}_{0} }{\sqrt{2}} \right] \vert \Psi \rangle = 0 \,.
\end{align*}
Then, one can write
\begin{align}
\mathcal{A}_{0} \mathcal{A}_{1} + \mathcal{A}_{1} \mathcal{A}_{0} &= \mathcal{A}_{0} \left[ \mathcal{A}_{1} - \frac{ \mathcal{B}_{1} + (-1)^d \mathcal{B}_{0} }{\sqrt{2}} \right] + \mathcal{A}_{1} \left[ \mathcal{A}_{0} - \frac{\mathcal{B}_{0} - (-1)^d \mathcal{B}_{1}}{\sqrt{2}} \right] \nonumber\\
&\hspace{2cm}+ \frac{(-1)^d \left[ (-1)^d ( \mathcal{I}_{1} + \mathcal{I}_{2}) - \mathcal{I}_{3} + \mathcal{I}_{4} \right] }{\sqrt{2}} \,, \label{ali:appsos2}\\
\mathcal{B}_{0} \, \mathcal{B}_{1} + \mathcal{B}_{1} \, \mathcal{B}_{0} &= \mathcal{B}_{0} \left[ \mathcal{B}_{1} -\frac{ \mathcal{A}_{1} + (-1)^d \mathcal{A}_{0} }{\sqrt{2}} \right] + \mathcal{B}_{1} \left[\mathcal{B}_{0} - \frac{\mathcal{A}_{0} - (-1)^d \mathcal{A}_{1}}{\sqrt{2}} \right] \nonumber\\
&\hspace{2cm}+ \frac{(-1)^d \left[ (-1)^d ( \mathcal{I}_{1} + \mathcal{I}_{2}) - \mathcal{I}_{3} + \mathcal{I}_{4} \right] }{\sqrt{2}} \label{ali:appsos3}
\end{align} 
As a result, we obtain \[ \mathcal{A}_{0} \mathcal{A}_{1} + \mathcal{A}_{} \mathcal{A}_{0} \vert \Psi \rangle = 0\,,~\mathcal{B}_{0} \, \mathcal{B}_{1} + \mathcal{B}_{1} \, \mathcal{B}_{0} \vert \Psi \rangle = 0 \,,\]
which imply the anticommutation relations \[\{ \Tilde{\mathcal{O}}_{v_{l}} \mathcal{R}^{(O)}_{d_{(v_{l},v_{r})}},\Tilde{\mathcal{Z}}_{v_{l}} \mathcal{R}^{(Z)}_{d_{(v_{l},v_{r})}} \} \vert \Psi \rangle = 0 \,,~\{ \Tilde{\mathcal{X}}_{v_{r}},\Tilde{\mathcal{Z}}_{v_{r}} \} \vert \Psi \rangle = 0 \,,\]
respectively.
Furthermore, we can obtain the anticommutation relation $\{ \mathcal{O}_{v_{l}} \mathcal{X}^{(O)}_{d_{(v_{l},v_{r})}},\mathcal{Z}_{v_{l}} \mathcal{X}^{(Z)}_{d_{(v_{l},v_{r})}} \} \vert \Psi \rangle = 0$ from the previous and $\mathcal{P}_{1} \, \mathcal{P}_{2} \, \mathcal{P}_{1} \, \mathcal{P}_{2}\vert \Psi \rangle = \vert \Psi \rangle$.

Together with the above anticommutation relations, we obtain $\{ \mathcal{X}_{v}, \mathcal{Z}_{v} \} \vert \Psi \rangle =0$ from
\begin{equation} \label{eq:propagation-pe3}
\zeta_{v} \, \mathcal{P}_{2} \, \zeta_{v} \, \mathcal{P}_{2} \,\vert \Psi \rangle = \vert \Psi \rangle \,,~\text{for }v \in N(v_{l})\,,\qquad \zeta_{v} \, \mathcal{P}_{1} \, \zeta_{v} \, \mathcal{P}_{1} \,\vert \Psi \rangle = \vert \Psi \rangle \,, ~\text{for }v \in N(v_{r}) \,.
\end{equation}
All other anticommutation relations follow from Lemma~\ref{lem:indu}.
\end{proof}

\begin{lem}\label{lem:indu}
Consider a physical experiment with the state $\vert \Psi \rangle $ and correlations
\begin{align}
\zeta_{u} &= \mathcal{X}_{u} \underset{(u,N(u))}{\Tilde{\mathcal{X}}^{(X)}_{\text{e,L}}} \underset{(u,N(u))}{\Tilde{\mathcal{X}}^{(Z)}_{\text{o,R}}} \Tilde{\mathcal{Z}}^{N(u)} \,,\label{eq:pinste-1}\\
\Tilde{\zeta}_{v} &= \Tilde{\mathcal{X}}_{v} \underset{(v,N(v))}{\Tilde{\mathcal{X}}^{(X)}_{\text{e,L}}} \underset{(v,N(v))}{\Tilde{\mathcal{X}}^{(Z)}_{\text{o,R}}} \mathcal{Z}^{N(v)}\label{eq:pinste-2} \,,
\end{align}
such that
\begin{equation}
\langle\Psi \vert \zeta_{u} \vert \Psi \rangle = \langle\Psi \vert \Tilde{\zeta}_{v} \vert \Psi \rangle = 1 \,,
\end{equation}
for $u \in N(v), v \in N(u)$.

If $ \{ \mathcal{X}_{u} , \mathcal{Z}_{u} \} \vert \Psi \rangle = 0$, then it also holds that $ \{ \Tilde{\mathcal{X}}_{v}, \Tilde{\mathcal{Z}}_{v} \} \vert \Psi \rangle = 0$.
\end{lem}
Note that in most applications of Lemma~\ref{lem:indu}, $\Tilde{\zeta}_{v} = \zeta_{v}$.
In general, the correlations $\zeta_{v},\Tilde{\zeta}_{v}$ can stem from different measurements, if the difference is not apparent on the marginalized observables.
In particular, it suffices that the marginalized observables do not overlap but on the power vertices that are of interest for the anticommutation relation.
\begin{proof}
Given the requirements, it is $\zeta_{u} \vert \Psi \rangle = \zeta_{v} \vert \Psi \rangle = \vert \Psi \rangle$ and $ \mathcal{Z}_{u} \mathcal{X}_{u} \mathcal{Z}_{u} \mathcal{X}_{u} \vert \Psi \rangle = (-1) \vert \Psi \rangle$.
Then, we evaluate
\begin{equation*}
\vert \Psi \rangle = \Tilde{\zeta}_{v} \, \zeta_{u} \, \Tilde{\zeta}_{v} \, \zeta_{u}\, \vert \Psi \rangle = (\mathcal{Z}_{u} \mathcal{X}_{u} \mathcal{Z}_{u} \mathcal{X}_{u}) (\Tilde{\mathcal{X}}_{v} \Tilde{\mathcal{Z}}_{v} \Tilde{\mathcal{X}}_{v} \Tilde{\mathcal{Z}}_{v}) \vert \Psi \rangle =(-1)\Tilde{\mathcal{X}}_{v} \Tilde{\mathcal{Z}}_{v} \Tilde{\mathcal{X}}_{v} \Tilde{\mathcal{Z}}_{v} \vert \Psi \rangle \,.
\end{equation*}
We use Eq.\;\eqref{eq:pinste-1} and Eq.\;\eqref{eq:pinste-2}, and that the observables of the physical measurement square to identity.
It follows that \(
\{\Tilde{\mathcal{X}}_{v},\Tilde{\mathcal{Z}}_{v} \} \vert \Psi \rangle = 0 \).
\end{proof}

\section{Self-testing symmetric graph states against communication}\label{app:selfsymmfull}

For the standard self-testing proof, we require observable of the physical experiment to reproduce the graph's stabilizer conditions, specifically that of the generator elements, and anticommutation relations for observables that are involved in the products of observables that reproduce these conditions, specifically corresponding to the Pauli $X$ and $Z$ observables.
In a scenario with bounded classical communication, 
The local observable corresponding to the Pauli $X$ operator must also include the information that its nearest neighbors measure in the Pauli $Z$ basis, and vice versa.

We present the proof for the ideal case $\epsilon=\delta=0$.
For a robust version of the equivalence, we refer to the one of RE~\ref{def:refe1} and RE~\ref{def:refe2} in Tab.~\ref{tab:deltas}.

\begin{proof}[Proof][$\epsilon=\delta=0$]
Given a physical experiment with state $\vert \Psi \rangle$ and measurement observables that are compatible and simulates the RE~\ref{def:fullcircle}.
The RE include RE~\ref{def:refe1} such that Lemma~\ref{lem:ref1234} applies, i.e., $\{ \mathcal{X}_{u},\mathcal{Z}_{u} \} \vert \Psi \rangle =0$ for all $u \in V^{\prime}$.
The observables rely on the vertices within communication distance $d$ to perform measurements in the Pauli $X$ basis.

We denote the tensor product of observables $\Tilde{\zeta}^{(X)}_{u}$ and $\Tilde{\zeta}^{(Z)}_{u}$ correspond to the Pauli submeasurements $f_{u}$ of the measurements $\Tilde{M}^{(X)}_{u}$ in Eq.\;\eqref{eq:measurxzx} and $\Tilde{M}^{(Z)}_{u}$ in Eq.\;\eqref{eq:measurzxz}, respectively.
From the simulation condition,
\begin{equation}
\langle \Psi \vert \Tilde{\zeta}^{(X)}_{u} \vert \Psi \rangle = \langle \Psi \vert \Tilde{\zeta}^{(Z)}_{u} \vert \Psi \rangle =1 \,,
\end{equation}
and, therefore, $\Tilde{\zeta}^{(X)}_{u} \vert \Psi \rangle = \Tilde{\zeta}^{(Z)}_{u} \vert \Psi \rangle = \vert \Psi \rangle$ for all $u \in V^{\prime}$.

We evaluate $\zeta^{(X)}_{u} \zeta^{(Z)}_{v} \zeta^{(X)}_{u} \zeta^{(Z)}_{v} \vert \Psi \rangle = \vert \Psi \rangle$ for any two neighboring vertices with distance $dist(u,v)=2d+1$.
Namely, we multiply the local observables column wise
\begin{equation} \begin{array}{ccccccccccccc}
\zeta^{(X)}_{u} & = \mathcal{X}_{u}& & \underset{(u,v)}{\Tilde{\mathcal{X}}^{(X)}_{\text{e},\text{L}}} & &\underset{(u,v)}{\Tilde{\mathcal{X}}^{(Z)}_{\text{e},\text{R}}} & \Z_{v} & \\
\zeta^{(Z)}_{v} & = \mathcal{Z}_{u} &\underset{(u,v)}{\Tilde{\mathcal{X}}^{(Z)}_{\text{o},\text{L}}} & & \underset{(u,v)}{\Tilde{\mathcal{X}}^{(X)}_{\text{o},\text{R}}} & & \X_{v} & \,,\\
\zeta^{(X)}_{u} & = \mathcal{X}_{u}& & \underset{(u,v)}{\Tilde{\mathcal{X}}^{(X)}_{\text{o},\text{L}}} & &\underset{(u,v)}{\Tilde{\mathcal{X}}^{(Z)}_{\text{o},\text{R}}} & \Z_{v} & \\
\zeta^{(Z)}_{v} & = \mathcal{Z}_{u} &\underset{(u,v)}{\Tilde{\mathcal{X}}^{(Z)}_{\text{o},\text{L}}} & & \underset{(u,v)}{\Tilde{\mathcal{X}}^{(X)}_{\text{o},\text{R}}} & & \X_{v} & \,,
\end{array} \label{eq:oddcirc}\end{equation}
similar to Lemma~\ref{lem:indu}, where $\X_{v},\Z_{v}$ are observables whose neighboring vertices up to distance $d$ measure the observables $X$ and $Z$ in an alternating pattern.
As a result, we obtain Eq.\;\eqref{eq:brandnewAC}
\begin{equation}
\{ \X_{v},\Z_{v} \} \vert \Psi \rangle =0 \label{eq:anticommalt}
\end{equation}
from $\{ \mathcal{X}_{u},\mathcal{Z}_{u} \} \vert \Psi \rangle =0 $.
The submeasurements $g_{u}$ from settings corresponding to $M_{u}^{(alt)}$ in Eq.\;\eqref{eq:ref4alt} lead to Eq.\;\eqref{eq:brandnewSTAB} for the product of observables $\xi_{u}$,
\begin{equation}
\xi_{u} \vert \Psi \rangle = \X_{u} \bigotimes_{v \in N(u)} \Z_{v} \vert \Psi \rangle = \vert \Psi \rangle\,.
\end{equation}
As a result, the physical state reproduces the stabilizer conditions of the inflated circular graph state.

Given a physical experiment with state $\vert \Psi \rangle$ and measurement observables that are compatible and simulates the RE~\ref{def:refhex}.
The RE contain RE~\ref{def:refe2} for a tripoint star graph around vertex $v_{c}$ in the honeycomb lattice, such that Lemma~\ref{lem:ref1234} applies with an anticommutation relation $\{ \mathcal{X}_{u}, \mathcal{Z}_{u}\} \vert \Psi \rangle = 0$ for the power vertices of the inflated star graph $v_{c}$ and $v \in N^{\prime}(v_{c})$, which are the vertices in the corners of the tripoint star.
However, these observables rely on the vertices on the tripoint star within communication distance $d$ to perform measurements in the Pauli $X$ basis.

We denote the tensor product of observables $\Tilde{\zeta}^{(X)}_{v_{0}}$ corresponds to the submeasurement $f_{v_{0}}$ from measurements $\Tilde{M}^{(X)}_{v_{0}}$ in Eq.\;\eqref{eq:ref4altx}, and $\Tilde{\zeta}^{(Z)}_{v_{c}}$ corresponds to $f_{v_{c}}$ from $\Tilde{M}^{(Z)}_{v_{0}}$ in Eq.\;\eqref{eq:ref4altz}.
Additionally, the tensor product of observables $\Tilde{\zeta}^{(alt)}_{u}$ corresponds to the submeasurement $f_{u}$ from measurements $M_{1}^{(alt)}$ in Eq.\;\eqref{eq:measurehex1} if $u \in V_{hex_{1}}$ or from $M_{2}^{(alt)}$ in Eq.\;\eqref{eq:measurehex2} if $u \in V_{hex_{2}}$.
From the simulation condition,
\begin{equation}
\langle \Psi \vert \Tilde{\zeta}^{(X)}_{v_{0}} \vert \Psi \rangle = \langle \Psi \vert \Tilde{\zeta}^{(Z)}_{v_{c}} \vert \Psi \rangle = \langle \Psi \vert \Tilde{\zeta}^{(alt)}_{u} \vert \Psi \rangle = 1\,,
\end{equation}
and, therefore, $\Tilde{\zeta}^{(X)}_{v_{0}} \vert \Psi \rangle = \Tilde{\zeta}^{(Z)}_{v_{c}} \vert \Psi \rangle = \Tilde{\zeta}^{(alt)}_{u} \vert \Psi \rangle = \vert \Psi \rangle$.

First, we evaluate $\zeta^{(X)}_{v_{0}} \zeta^{(Z)}_{v_{c}} \zeta^{(X)}_{v_{0}} \zeta^{(Z)}_{v_{c}} \vert \Psi \rangle = \vert \Psi \rangle$ for any two neighboring vertices with distance $dist(u,v)=2d+1$ in the same way as in the proof of Lemma~\ref{lem:indu}.
As a result, \( \{ \X_{v_{c}},\Z_{v_{c}} \} \vert \Psi \rangle =0 \) from $\{ \mathcal{X}_{v_{0}},\mathcal{Z}_{v_{0}} \} \vert \Psi \rangle =0 $, for observables $\X,\Z$ that correspond to measurement settings where the neighboring vertices up to distance $d$ measure the observables $X$ and $Z$ in an alternating pattern.

We propagate the anticommutation relation~\eqref{eq:anticommalt} to every vertex of the honeycomb lattice starting with $v_{c}$ and consecutively evaluating $\Tilde{\zeta}^{(alt)}_{u} \Tilde{\zeta}^{(alt)}_{v} \Tilde{\zeta}^{(alt)}_{u} \Tilde{\zeta}^{(alt)}_{v} \vert \Psi \rangle = \vert \Psi \rangle$ in the same way as in the proof of Lemma~\ref{lem:indu}.
The vertex $u$ is at the center of a tripoint star and vertex $v$ at distance $dist(u,v)=2d+1$ in one of the corners of the tripoint star. 
Then, it is \(\{ \X_{v},\Z_{v} \} \vert \Psi \rangle =0 \) from $\{ \mathcal{X}_{u},\mathcal{Z}_{u} \} \vert \Psi \rangle =0 $.
The submeasurements $g_{u}$ from settings corresponding to $M^{(alt)}_{hex_{1}}$ in Eq.\;\eqref{eq:measurehex1} if $u \in V_{hex_{1}}$ or~$M^{(alt)}_{hex_{2}}$ in Eq.\;\eqref{eq:measurehex2} if $u \in V_{hex_{2}}$ lead to Eq.\;\eqref{eq:brandnewSTAB} for the product of observables $\xi_{u}$,
\begin{equation}
\xi_{u} \vert \Psi \rangle = \X_{u} \bigotimes_{v \in N(u)} \Z_{v} \vert \Psi \rangle = \vert \Psi \rangle\,.
\end{equation}
As a result, the physical state reproduces the stabilizer conditions of the inflated circular graph state.

We show the equivalence in terms of Def.~\ref{def:equirob} between the PE and the RE with the local isometry is $\Phi = \prod_{v \in V^{\prime}(V_{hex})} \phi_{v}$ with $\phi_{v}$ characterized by the first circuit in Fig.~\ref{fig:swapcirc} using the observables $\X_{u}$ and $\Z_{u}$.
The evaluation of the isometry is the same as in Appendix~\ref{app:line}, but we use $\xi_{u} \vert \Psi \rangle = \X_{u} (\Z)^{N(u)} \vert \Psi \rangle =\vert \Psi \rangle$ and the anticommutation relation in Eq.\;\eqref{eq:anticommalt}.
The tensor product of observables $\xi_{u}$ corresponds to the measurement submeasurements $g_{u}$ from $M^{(X)}_{u}$ for RE~\ref{def:fullcircle} and from $M^{(alt)}$ for RE~\ref{def:refhex}.
\end{proof}

\section{Inflated Isomorphism}
\label{app:isocom}
Here, we prove Theorem~\ref{theo:main} for the ideal case, $\epsilon = 0$ and $\delta = 0$. In Appendix~\ref{app:robref3}, we extend the proof to robust equivalence by trailing the proof for the ideal case and bounding all expressions that result from the accordance of the measurement correlations.

Consider a physical experiment (PE) that is compatible and simulate one of the reference experiments (REs)~\ref{def:refe1},\ref{def:refe2} and~\ref{def:refe3}, such that Lemma~\ref{lem:ref1234} applies, and it is $\zeta_{u} \vert \Psi \rangle = \vert \Psi \rangle$ for all $u \in V$.

Measuring the chain vertices in the measurement bases corresponding to observable $\mathcal{X}$ with outcome $\mathbf{x} = (x_{v})_{v\in V^{\prime} \setminus V}$, transforms the state
\begin{equation}
\vert \Psi \rangle \longrightarrow \vert \Psi^{\prime} \rangle = \Pi_{\mathcal{X}}^{(\mathbf{x})} \vert \Psi \rangle = \bigotimes_{v \in V^{\prime}\setminus V} \left( \mathds{1}_{v} + (-1)^{x_{v}} \mathcal{X}_{v} \right) \vert \Psi \rangle \,.
\end{equation}

The isometry that maps the subsequent state (on the power vertices) to the graph state corresponding to the graph in the RE is a product of local isometries for each power vertex.
The local isometries are described by the first circuit in Fig.~\ref{fig:swapcirc}.
We employ the standard isometry inspired on the SWAP gate in Fig.~\ref{fig:swapcirc}, which was proposed in~\cite{mayers2003self}, used in~\cite{mckague2011self} and in Section~\ref{app:line}.

The isometry acts locally on the post-measurement state and one ancilla qubit for every power vertex,
\begin{equation} \Phi \left( \, \vert \mathbf{0} \rangle \otimes \Pi_{\mathcal{X}}^{(\mathbf{x})} \vert \Psi \rangle \, \right) = \prod_{v \in V} \phi_{v} \hspace{-0.05cm} \left( \vert \mathbf{0} \rangle \otimes \Pi_{\mathcal{X}}^{(\mathbf{x})} \vert \Psi \rangle \right) \,. \label{eq:isogen}
\end{equation} 
We denote $(x_{\text{e}})_{v} := \sum^d_{s=1} x_{2s_{(v,N(v))}}$ as the sum of measurement outcomes of every second chain vertex seen from the power vertex $v$ up to the next power vertices.
Denoting $Z_\mathbf{a}^{\mathbf{x}_{\text{e}}} = \bigotimes_{u \in V} Z^{(x_{\text{e}})_{u}}_{a_{u}}$, it is
\begin{align}
\Phi \left(\vert \mathbf{0} \rangle \otimes \Pi_{\mathcal{X}}^{(\mathbf{x})} \vert \Psi \rangle\right) &= \frac{1}{2^{\vert V \vert}} \sum_{\mathbf{a}\in \mathbb{F}^{\vert V \vert}_{2}} \vert \mathbf{a} \rangle \otimes \Pi_{\mathcal{X}}^{(\mathbf{x})} \bigotimes_{u \in V} \left( \mathds{1}_{u} + \mathcal{Z}_{u} \right) \mathcal{X}_{u}^{a_{u}} \vert \Psi \rangle \label{eq:isoinf}\\
&= \frac{1}{2^{\vert V \vert}} \sum_{\mathbf{a}\in \mathbb{F}^{\vert V \vert}_{2}} (-1)^{g(\mathbf{a})} \vert \mathbf{a} \rangle \otimes \Pi_{\mathcal{X}}^{(\mathbf{x})}\hspace{-0.05cm} \bigotimes_{u \in V} \underset{(u,N(u))}{\left(\mathcal{X}^{(X)}_{\text{e},\text{L}} \mathcal{X}^{(Z)}_{\text{e},\text{R}}\right)^{a_{u}}}\left( \mathds{1}_{u} + \mathcal{Z}_{u} \right) \vert \Psi \rangle \label{eq:isoinfh1} \\
&= \frac{1}{2^{\vert V \vert}}\sum_{\mathbf{a}\in \mathbb{F}^{\vert V \vert}_{2}} (-1)^{g(\mathbf{a}) + \mathbf{x}_{\text{e}}\mathbf{a} } \vert \mathbf{a} \rangle \otimes \Pi_{\mathcal{X}}^{(\mathbf{x})} \bigotimes_{u \in V} \left( \mathds{1}_{u} + \mathcal{Z}_{u} \right) \vert \Psi \rangle \label{eq:isoinfh2}\\ 
&= \vert G^{(\mathbf{x})} \rangle \otimes \vert \mathit{junk} \rangle \label{eq:isoinfh3} \,, 
\end{align}
with the graph state's signature \[g(\mathbf{a}) := \sum_{(u,v) \in E} a_{u}a_{v} \,. \]
Line~\eqref{eq:isoinf} was obtained by evaluating the action of the circuit $\phi_{v}$, which has been done in Eq.\;\eqref{eq:isocirceval} in Appendix~\ref{app:line}.

To arrive at Eq.\;\eqref{eq:isoinfh1}, consider an arbitrary but fixed order of vertices and sequentially perform the following three steps. 
First, from $\zeta_{u} \vert \Psi \rangle = \vert \Psi \rangle$ and Eq.\;\eqref{eq:pe-infstabpauli}, apply \[\mathcal{X}^{a_{u}}_{u} \vert \Psi \rangle = \bigg( \underset{(u,N(u))}{\left(\mathcal{X}^{(X)}_{\text{e},\text{L}}\mathcal{X}^{(Z)}_{\text{e},\text{R}}\right)} Z^{N(u)} \bigg)^{a_{u}} \vert \Psi \rangle \,.\] 
With the anticommutations relations, permute $\mathcal{Z}^{N(u)}$ from right to left. Depending on the order of vertices, if $u < v \in N(u)$, $(\mathcal{Z}_{v})^{a_{u}}$ anticommutes with $(\mathcal{X}_{v})^{a_{v}}$, leading to a phase $(-1)^{a_{u} a_{v}}$ which is acquired if and only if two vertices share an edge.
Then, absorb $\left( \mathds{1}_{u} + \mathcal{Z}_{u} \right) ( \mathcal{Z}_{v} )^{a_{u}} = \left( \mathds{1}_{u} + \mathcal{Z}_{u} \right)$.

For Eq.\;\eqref{eq:isoinfh2}, use $\Pi^{(x)}_{\mathcal{X}} \mathcal{X} = (-1)^x \Pi^{(x)}_{\mathcal{X}}$ for the operators on the chain vertices.
To arrive at Eq.\;\eqref{eq:isoinfh3}, define the locally rotated graph state
\begin{equation}
\vert G^{(\mathbf{x})} \rangle = \frac{1}{2^{\vert V \vert/2}}\sum_{\mathbf{a}\in \mathbb{F}^{\vert V \vert}_{2}} (-1)^{g(\mathbf{a}) + \mathbf{x}_{\text{e}}\mathbf{a} } \vert \mathbf{a} \rangle \,,
\end{equation}
and a residual state
\begin{equation}
\vert junk \rangle = 2^{\vert V \vert/2} \Pi_{\mathcal{X}}^{(\mathbf{x})} \bigotimes_{u \in V} \left( \mathds{1}_{u} + \mathcal{Z}_{u} \right) \vert \Psi \rangle \,.
\end{equation}

There are three caveats regarding RE~\ref{def:refe3}. First, for odd $d$, the correction operation $Z^{(x_{\text{e}})_{v_{l}}}_{a_{v_{l}}}$ on vertex $v_{l}$ depends additionally on the measurement outcomes of chain vertex $d_{(v_{l}.v_{r})}$, i.e., $(x_{\text{e}})_{v_{l}} \rightarrow (x_{\text{e}})_{v_{l}} + x_{d_{(v_{l},v_{r})}}$. 
Second, $\mathcal{X}_{v_{l}}\mathcal{Z}_{v_{l}} \vert \Psi \rangle = (-1)\mathcal{Q} \, \mathcal{Z}_{v_{l}} \mathcal{X}_{v_{l}} \vert \Psi \rangle $ with $\mathcal{Q} := \left(\mathcal{X}^{(Z)} \mathcal{X}^{(O)} \mathcal{X}^{(Z)} \mathcal{X}^{(O)}\right)_{d_{(v_{l},v_{r})}}$.
Throughout the above considerations, $\mathcal{Q}$ occurs whenever one uses the anticommutation relation.
However, it can immediately be absorbed in the measurement projectors since $\Pi_{\mathcal{X}}^{(\mathbf{x})} \mathcal{Q} = \Pi_{\mathcal{X}}^{(\mathbf{x})}$ and therefore does not affect the result. Third, for $d=1$, we replace the observable $\mathcal{X}_{v_{l}}$ by $\mathcal{Y}_{v_{l}}$ in the isomorphism.

For any measurement of the PE $\mathcal{M}_k$, the action under the isometry is $\Phi^{(\mathbf{x})} \left( \vert 0 \rangle \otimes \Pi_{\mathcal{X}}^{(\mathbf{x})} \mathcal{M}_k \vert \Psi \rangle \right)$. Since all measurements can be composed into products of observables $\mathcal{Z}$ and $\mathcal{X}$ on the power vertices and operators $\mathcal{X}$ on the chain vertices, we consider these separately. 
For some power vertex $v$ and $\Phi^{(\mathbf{x})} \left( \vert 0 \rangle \otimes \Pi_{\mathcal{X}}^{(\mathbf{x})} \mathcal{X}_{v} \vert \Psi^{\prime} \rangle \right)$, insert $X_{v}$ to the right of Eq.\;\eqref{eq:isoinf} to see that it shifts $a_{v} \rightarrow (a_{v} + 1) \bmod 2$. 
This is exactly the action of a Pauli $X$ operator on $\vert a_{v} \rangle$. For $\Phi^{(\mathbf{x})} \left( \vert 0 \rangle \otimes \Pi_{\mathcal{X}}^{(\mathbf{x})} \mathcal{X}_{v} \vert \Psi^{\prime} \rangle \right)$, inserting $Z_{v}$ to the right of Eq.\;\eqref{eq:isoinf}, and anticommute it with $\mathcal{X}_{v}^{a_{v}}$ for a sign factor $(-1)^{a_{v}}$, which is exactly the action of a Pauli $Z$ operator on $\vert a_{v} \rangle$. 
For any chain vertex $v$, $\Pi_{\mathcal{X}}^{(\mathbf{x})} \mathcal{X}_{v} = (-1)^{x_{v}} \Pi_{\mathcal{X}}^{(\mathbf{x})}$ leads to a global phase, undetectable by any measurement. 
By linearity, it is then \( \Phi^{(\mathbf{x})} \left( \vert 0 \rangle \otimes \Pi_{\mathcal{X}}^{(\mathbf{x})} \mathcal{M}_k \vert \Psi \rangle \right) = M_k \vert G \rangle \otimes \vert \mathit{junk}\rangle \,,\) for any measurement setting $M_k$ reduced to its action on the power vertices. This concludes our proof of Theorem~\ref{theo:main}.

In Appendix~\ref{app:robref3}, we extend Theorem~\ref{theo:main} to a robust self-testing procedure with techniques present in~\cite{mckague2011self} and in~\cite{SOS2015}. The bounds on the isometry for all discussed REs are in Tab.~\ref{tab:deltas}. They scale with $\sqrt{\epsilon}$ where $\epsilon$ is the deviation of the measurement results from the ideal ones, as in Def.~\ref{def:selfrob}. The bounds on the self-testing procedure with inflated graph states allowing bounded classical communication are slightly worse than for the non-inflated graph states, which is mainly due to the higher number of measurements required to derive the anticommutation relations.

\section{Self-testing robustness }
\label{app:robref3}
\begin{table}
\begin{tabular}{Sl|Sc|Sl}
Reference experiment & $c$ & Robustness \\ \cline{1-3}
RE~\ref{def:refe0} & $4$ & $~\delta= \sqrt{\epsilon}\, \vert V \vert (2 \vert V \vert^2 + 6 \vert V \vert + 6)$ \\
RE~\ref{def:refe1} & $3m+1$ & $~\delta= \sqrt{\epsilon}\, \vert V \vert (\frac{7}{2} \vert V \vert^2 + \frac{15}{2} \vert V \vert +2) $ \\
RE~\ref{def:refe2} & &\\ 
$~\vert N(v_{c}) \vert = 2$ & $6$ & $~\delta= \sqrt{\epsilon}\, \vert V \vert ( 2 \vert V \vert^2 + 7 \vert V \vert + 8 )$ \\ 
$~\vert N(v_{c}) \vert > 2$ & $4 + 4\vert N(v_{c}) \vert$ & $~\delta= \sqrt{\epsilon}\, \vert V \vert ( 2 \vert V \vert^2 + 6 \vert V \vert + 6 )$ \\ 
RE~\ref{def:refe3} & $~\sqrt{8 \sqrt{2}}(1 + \sqrt{2})~$ & $~\delta= \sqrt{\epsilon}\, \vert V \vert[ 2 \vert V \vert^2+4 \vert V \vert +1 + (2 + \vert V \vert)\sqrt{2\sqrt{2}}(1+\sqrt{2}) ]$ \\\cline{1-3}
Def.\,3 in~\cite{mckague2011self} & $m+1$ &$~\delta= \sqrt{\epsilon}\, \vert V \vert (\frac{5}{2} \vert V \vert^2 + \frac{11}{2} \vert V \vert +2)$ \\
Def.\,4 in~\cite{mckague2011self} & $-$ &$~\delta= \sqrt{\epsilon}\, \vert V \vert(2 \vert V \vert^2 + 4 \vert V \vert + 1)+ \vert V \vert \sqrt[4]{\epsilon} ( \frac{13}{2} \vert V \vert + 13)$ 
\end{tabular}
\caption{Table of bounds on robustness $\delta$ in Def.~\ref{def:selfrob} from Eq.\;\eqref{eq:robgenb}, parameter $c$, with $\vert E \vert \leq \vert V \vert^2$ and $l\leq \vert V \vert$.
Note that $m$ is the number of vertices in the circle in RE~\ref{def:refe1}.}
\label{tab:deltas}
\end{table}
We show that our self-testing proofs from Theorem~\ref{theo:main}, Theorem~\ref{theo:entire}, and Theorem~\ref{theo:ref0} are robust in terms of Def.~\ref{def:selfrob}.
For this purpose, we mainly use the triangle inequality and the following inequalities.

From $\langle \Psi \vert \mathcal{M} \vert \Psi \rangle > 1 - \epsilon $, it is $\vert \Psi \rangle - \mathcal{M} \vert \Psi \rangle \Vert < \sqrt{2\epsilon}$, and it is $\vert \Psi \rangle - \mathcal{M}_{1} \mathcal{M}_{2} \vert \Psi \rangle \Vert < \alpha + \beta$ for $\vert \Psi \rangle - \mathcal{M}_{1} \vert \Psi \rangle \Vert$, $\vert \Psi \rangle - \mathcal{M}_{2} \vert \Psi \rangle \Vert$ and if $\Vert \mathcal{M}_{1} \Vert_\infty = 1$.

We first bound the anticommutation relations from Lemma~\ref{lem:ref1234} and Lemma~\ref{lem:indu} by counting the number of measurements involved in obtaining the anticommutation relation on the first vertex, and then adding the number of times Lemma~\ref{lem:indu} has to be applied, which is the distance along the edges between the first vertex and any other vertex in the graph.
The only exceptions are the vertices $v_{l},v_{r}$ in RE~\ref{def:refe3}, where we use techniques from~\cite{SOS2015}.
Consider Eq.\;\eqref{eq:robdef0} for the observables in the physical experiment $\mathcal{I}_i$ simulating the $I_i$ from Eqs.\,\eqref{ali:pcorr41} -- \eqref{ali:pcorr44} for $i=1,2,3,4$.
Explicitly, it is $ \vert \sqrt{2} \langle \Psi \vert \mathcal{I}_{1} \vert \Psi\rangle - (-1)^d\vert < \sqrt{2} \epsilon$, $\vert \sqrt{2} \langle \Psi \vert \mathcal{I}_{2} \vert \Psi\rangle + (-1)^d \vert < \sqrt{2}\epsilon$ and $\vert \sqrt{2} \langle \Psi \vert \mathcal{I}_i \vert \Psi\rangle - 1 \vert < \sqrt{2} \epsilon$ for $i = 3,4$.
It follows that $(-1)^d \sqrt{2} \langle \Psi \vert \mathcal{I}_{1} \vert \Psi\rangle < -1 + \sqrt{2}\epsilon$, $(-1)^{d} \sqrt{2} \langle \Psi \vert \mathcal{I}_{2} \vert \Psi\rangle > 1 - \sqrt{2}\epsilon$ and $\sqrt{2} \langle \Psi \vert \mathcal{I}_i \vert \Psi\rangle > 1 - \sqrt{2} \epsilon$ for $i=3,4$.
As a result, \[ \langle \Psi \vert 4 \mathds{1} + \sqrt{2} \left( (-1)^d (\mathcal{I}_{1} - \mathcal{I}_{2}) - \mathcal{I}_{3} - \mathcal{I}_{4} \right) \vert \Psi \rangle < 4 \sqrt{2} \,\epsilon \,.\]
Due to the operator's sums of squares decompositions, every squared element is equally bounded by 
\begin{align*} 4 \sqrt{2} \, \epsilon > & \, \langle\Psi \vert \left[ \frac{1}{2}\left( (-1)^d(\mathcal{I}_{1} + \mathcal{I}_{2}) - \mathcal{I}_{3} + \mathcal{I}_{4} \right) \right]^2 \vert \Psi \rangle = \left\Vert \left[\frac{1}{2} \left( (-1)^d(\mathcal{I}_{1} + \mathcal{I}_{2}) - \mathcal{I}_{3} + \mathcal{I}_{4} \right) \right]\vert \Psi \rangle \right\Vert^2_{1}\,.\end{align*}
It follows from Eq.~\ref{ali:sosanti1} and Eq.~\ref{ali:sosanti2} that
\begin{align*} 4 \sqrt{2} \, \epsilon &> \langle\Psi \vert \left[ \mathcal{A}_{0} - \frac{\mathcal{B}_{0} - (-1)^d \mathcal{B}_{1}}{\sqrt{2}} \right]^2 \vert \Psi \rangle = \left\Vert \left[\mathcal{A}_{0} - \frac{\mathcal{B}_{0} - (-1)^d \mathcal{B}_{1}}{\sqrt{2}} \right]\vert \Psi \rangle \right\Vert^2_{1}\,, \\
4 \sqrt{2} \epsilon &> \langle \Psi \vert \left[ \mathcal{A}_{1} - \frac{ \mathcal{B}_{1} + (-1)^d \mathcal{B}_{0} }{\sqrt{2}} \right]^2 \vert \Psi \rangle = \left\Vert \left[ \mathcal{A}_{1} - \frac{ \mathcal{B}_{1} + (-1)^d \mathcal{B}_{0} }{\sqrt{2}} \right] \vert \Psi \rangle \right\Vert^2_{1} \,,\\
4 \sqrt{2} \epsilon &> \langle \Psi \vert \left[ \frac{\mathcal{A}_{0} - (-1)^d \mathcal{A}_{1}}{\sqrt{2}} - \mathcal{B}_{0} \right]^2 \vert \Psi \rangle = \left\Vert \left[ \frac{\mathcal{A}_{0} - (-1)^d \mathcal{A}_{1}}{\sqrt{2}} - \mathcal{B}_{0}\right] \vert \Psi \rangle \right\Vert^2_{1} \,,\\
4 \sqrt{2} \epsilon &> \langle \Psi \vert \left[\frac{ \mathcal{A}_{1} + (-1)^d \mathcal{A}_{0} }{\sqrt{2}} - \mathcal{B}_{1} \right]^2 \vert \Psi \rangle = \left\Vert \left[ \frac{ \mathcal{A}_{1} + (-1)^d \mathcal{A}_{0} }{\sqrt{2}} - \mathcal{B}_{1} \right] \vert \Psi \rangle \right\Vert^2_{1} \,.\end{align*}
To bound the anticommutation relations $ \left\Vert \left\{ \mathcal{O}_{v_{l}} \mathcal{R}^{(O)}_{d_{(v_{l},v_{r})}},\mathcal{Z}_{v_{l}} \mathcal{R}^{(Z)}_{d_{(v_{l},v_{r})}} \right\} \vert \Psi \rangle \right\Vert_{1}$ and $ \left\Vert \{ \mathcal{X}_{v_{r}},\mathcal{Z}_{v_{r}} \} \vert \Psi \rangle \right\Vert_{1}$, we use Eq.\;\eqref{ali:appsos2} and Eq.\;\eqref{ali:appsos3}, respectively, together with the triangle inequality.
As a result, it is
\begin{align}
\left\Vert \left\{ \mathcal{O}_{v_{l}} \mathcal{R}^{(O)}_{d_{(v_{l},v_{r})}},\mathcal{Z}_{v_{l}} \mathcal{R}^{(Z)}_{d_{(v_{l},v_{r})}} \right\} \vert \Psi \rangle \right\Vert_{1} &< \sqrt{8 \sqrt{2}}(1 + \sqrt{2}) \sqrt{\epsilon} \,, \\
\left\Vert \{ \mathcal{X}_{v_{r}},\mathcal{Z}_{v_{r}} \} \vert \Psi \rangle \right\Vert_{1} &< \sqrt{8 \sqrt{2}}(1 + \sqrt{2}) \sqrt{\epsilon} \,.
\end{align}

With the bound on the anticommutation relations, we can bound the isometry using the triangle inequality whenever the anticommutation relations or relations directly from the measurements are used with more details in~\cite{mckague2011self}.
Note that, when evaluating the isometry on the measurements applied to the state, we do not use additional anticommutation relations; therefore, the bounds of Eq.\;\eqref{ali:robdef1} and Eq.\;\eqref{ali:robdef2} are the same.
Then, a general bound is
\begin{equation}
\delta = \frac{\sqrt{\epsilon}}{2} \Big( (2 \vert V \vert + \vert E \vert ) \, ( c +4l) \Big) \,,
\label{eq:robgenb}
\end{equation}
with the number of vertices $\vert V \vert $ and number of edges $\vert E \vert$ in the (inflated) graph.
The parameter $c_{1}$ depends on the number of measurements needed to establish the first anticommutation relation, and $l$ is the largest distance along the edges from this vertex to any other.
The parameter $c_{2}$ is the number of measurements involved in the mapping of $\mathcal{X}_{u} (\mathcal{Y}_{u})$ to $\mathcal{Z}^{N(u)}$.
For the different reference experiments, the parameters and bounds on $\delta$ are shown in Fig.~\ref{tab:deltas}.

The bounds on $\delta$ are not optimal and can be improved, especially if more is known about the structure of the graph.
Specifically, for larger graphs, it might be useful to combine several reference experiments when required subgraphs occur multiple times in the graph.

\end{document}